\def\submission{0}
\ifnum\submission=0
\documentclass[11pt,letterpaper]{article}
\usepackage[margin=1in]{geometry}
\else
\documentclass{llncs}
\fi
\usepackage[utf8]{inputenc}
\usepackage[colorlinks=true,allcolors=blue]{hyperref}
\usepackage{amsmath}
\usepackage{amsfonts}
\ifnum\submission=0
\usepackage{amsthm}
\fi
\usepackage{braket}
\usepackage{authblk}
\usepackage{comment}
\usepackage[colorlinks=true,allcolors=blue]{hyperref}
\usepackage{breakcites}
\usepackage{cleveref}
\usepackage[sanserif,basic]{complexity}
\usepackage{graphicx}

\ifnum\submission=0
\newtheorem{theorem}{Theorem}[section]
\newtheorem{conjecture}[theorem]{Conjecture}
\newtheorem{lemma}[theorem]{Lemma}
\newtheorem{claim}[theorem]{Claim}

\newtheorem{corollary}[theorem]{Corollary}
\newtheorem{definition}[theorem]{Definition}

\newtheorem{remark}{Remark}

\fi
\DeclareMathOperator*{\Ex}{\mathbb{E}}

\crefname{claim}{Claim}{Claims}

\ifnum\submission=0
\title{Verifiable Quantum Advantage without Structure}
\author[1]{Takashi Yamakawa\thanks{This work was done in part while the author was visiting Princeton University.}}
\author[2]{Mark Zhandry}
\affil[1]{NTT Social Informatics Laboratories}
\affil[2]{NTT Research}

\else
\title{Toward Classically Verifiable Quantum Advantage from Random Oracles}
\author{Takashi Yamakawa\inst{1}\thanks{This work was done in part while the author was visiting Princeton University.} \and
Mark Zhandry\inst{2,3}}
\authorrunning{T. Yamakawa et al.}
\institute{
NTT Social Informatics Laboratories, Tokyo, Japan 
\email{takashi.yamakawa.ga@hco.ntt.co.jp}
\and
Princeton University, Princeton, USA
\email{mzhandry@princeton.edu}
\and
NTT Research, Palo Alto, USA
}
\fi

\begin{document}
\maketitle
\ifnum\submission=1
\vspace{-7mm}
\fi

\begin{abstract}
 We show the following hold, unconditionally unless otherwise stated, relative to a random oracle: 
\begin{itemize}
    \item There are $\NP$ \emph{search} problems solvable by 
    quantum polynomial-time machines but not classical probabilistic polynomial-time machines.
    \item There exist functions that are one-way, and even collision resistant, against classical adversaries but are easily inverted quantumly. Similar counterexamples exist for digital signatures and CPA-secure public key encryption (the latter requiring the assumption of a classically CPA-secure encryption scheme). Interestingly, the counterexample does not necessarily extend to the case of other cryptographic objects such as PRGs.
    \item There are unconditional publicly verifiable proofs of quantumness with the minimal rounds of interaction: for uniform adversaries, the proofs are non-interactive, whereas for non-uniform adversaries the proofs are two message public coin.
    \item Our results do not appear to contradict the Aaronson-Ambanis conjecture. Assuming this conjecture, there exist publicly verifiable certifiable randomness, again with the minimal rounds of interaction.
\end{itemize}
By replacing the random oracle with a concrete cryptographic hash function such as SHA2, we obtain plausible Minicrypt instantiations of the above results. Previous analogous results all required substantial structure, either in terms of highly structured oracles and/or algebraic assumptions in Cryptomania and beyond. 
\end{abstract}
\thispagestyle{empty}
\clearpage 

\newcommand{\markz}[1]{{\color{blue}(\textbf{Mark}: #1)}}
\newcommand{\takashi}[1]{{\color{red}(\textbf{Takashi}: #1)}}
\newcommand{\revise}[1]{{\color{red}#1}}

\def \sample { \overset{\hspace{0.1em}\mathsf{\scriptscriptstyle\$}}{\leftarrow} }
\newcommand{\ra}{\rightarrow}
\newcommand{\la}{\leftarrow}
\newcommand{\negl}{\mathsf{negl}}
\newcommand{\secpar}{\lambda}
\newcommand{\A}{\mathcal{A}}
\newcommand{\B}{\mathcal{B}}
\newcommand{\siml}{\mathcal{S}}
\newcommand{\ora}{\mathcal{O}}
\newcommand{\bit}{\{0,1\}}
\newcommand{\td}{\mathsf{td}}
\newcommand{\defeq}{:=}
\newcommand{\oracle}{\mathcal{O}}
\newcommand{\enc}{\mathsf{Enc}}
\newcommand{\dec}{\mathsf{Dec}}
\newcommand{\calI}{\mathcal{I}}
\newcommand{\calS}{\mathcal{S}}
\newcommand{\calX}{\mathcal{X}}
\newcommand{\calY}{\mathcal{Y}}
\newcommand{\calZ}{\mathcal{Z}}
\newcommand{\calR}{\mathcal{R}}
\newcommand{\vecr}{{\mathbf{r}}}
\newcommand{\vecx}{\mathbf{x}}
\newcommand{\vecy}{\mathbf{y}}
\newcommand{\vecz}{\mathbf{z}}
\newcommand{\matA}{{\mathbf{A}}}
\newcommand{\matB}{{\mathbf{B}}}
\newcommand{\matD}{{\mathbf{D}}}
\newcommand{\reprogram}{\mathsf{Reprogram}}
\newcommand{\Func}{\mathsf{Func}}
\newcommand{\TD}{\mathsf{TD}}
\newcommand{\Samp}{\mathsf{Samp}}
\newcommand{\experiment}{\mathsf{Exp}}
\newcommand{\semiconst}{\mathsf{SC}}
\newcommand{\perm}{\mathsf{Perm}}
\newcommand{\FF}{\mathbb{F}}
\newcommand{\vecv}{\mathbf{v}}
\newcommand{\vu}{\mathbf{u}}
\newcommand{\vx}{\mathbf{x}}
\newcommand{\vy}{\mathbf{y}}
\newcommand{\vz}{\mathbf{z}}
\newcommand{\ve}{\mathbf{e}}
\newcommand{\vc}{\mathbf{c}}
\newcommand{\vm}{\mathbf{m}}
\newcommand{\mG}{\mathbf{G}}
\newcommand{\mH}{\mathbf{H}}
\newcommand{\valpha}{\boldsymbol \alpha}
\newcommand{\GRS}{\mathrm{GRS}}
\newcommand{\GRSDecode}{\mathrm{GRSDecode}}
\newcommand{\GRSDecodeError}{\mathrm{GRSDecodeError}}
\newcommand{\GRSListDecode}{\mathrm{GRSListDecode}}
\newcommand{\on}{\mathsf{on}}
\newcommand{\off}{\mathsf{off}}
\newcommand{\hw}{\mathsf{hw}}
\newcommand{\HW}{\mathcal{H}\mathcal{W}}
\newcommand{\HWl}{\HW_{\leq \alpha n }}
\newcommand{\HWg}{\HW_{> \alpha n }}
\newcommand{\etag}{\eta_{1}}
\newcommand{\etale}{\eta_{0}}
\newcommand{\bad}{\mathsf{BAD}}
\newcommand{\good}{\mathsf{GOOD}}
\newcommand{\Tr}{\mathrm{Tr}}
\newcommand{\mfol}{(m)}
\newcommand{\decode}{\mathsf{Decode}}
\newcommand{\encode}{\mathsf{Encode}}
\newcommand{\RS}{\mathrm{RS}}
\newcommand{\vzero}{\mathbf{0}}
\newcommand{\keylength}{{\ell_\mathsf{key}}}
\newcommand{\inlength}{{\ell_\mathsf{in}}}
\newcommand{\outlength}{{\ell_\mathsf{out}}}
\newcommand{\pilength}{{\ell_{\pi}}}
\newcommand{\accept}{\mathsf{acc}}
\newcommand{\poq}{\mathsf{poq}}
\newcommand{\concat}{||}
\newcommand{\Approx}{\mathsf{Approx}}
\newcommand{\gooderrors}{\mathcal{G}}
\newcommand{\baderrors}{\mathcal{B}}
\newcommand{\dist}{\mathcal{D}}
\newcommand{\bardist}{\widetilde{\mathcal{D}}}
\newcommand{\bardistprime}{\widetilde{\mathcal{D}'}}
\newcommand{\hashset}{\widetilde{\mathcal{H}}}
\newcommand{\Col}{\mathrm{Col}}

\newcommand{\poqsetup}{\mathsf{PoQRO}.\mathsf{Setup}}
\newcommand{\poqprove}{\mathsf{PoQRO}.\mathsf{Prove}}
\newcommand{\poqverify}{\mathsf{PoQRO}.\mathsf{Verify}}
\newcommand{\pk}{\mathsf{pk}}
\newcommand{\sk}{\mathsf{sk}}
\newcommand{\gen}{\mathsf{Gen}}
\newcommand{\prove}{\mathsf{Prove}}
\newcommand{\verify}{\mathsf{Verify}}
\newcommand{\prob}{\mathsf{prob}}
\newcommand{\sol}{\mathsf{sol}}
\newcommand{\QFT}{\mathsf{QFT}}
\newcommand{\learner}{\mathcal{L}}
\newcommand{\efficientlearner}{\widetilde{\mathcal{L}}}
\newcommand{\Lout}{L_{\mathsf{out}}}

\newcommand{\ecrhsetup}{\mathsf{ECRH}.\mathsf{Setup}}
\newcommand{\ecrhgen}{\mathsf{ECRH}.\mathsf{Gen}}
\newcommand{\ecrheval}{\mathsf{ECRH}.\mathsf{Eval}}
\newcommand{\ecrhequiv}{\mathsf{ECRH}.\mathsf{Equiv}}
\newcommand{\crs}{\mathsf{crs}}

\newcommand{\sigkeygen}{\mathsf{Sig}.\mathsf{KeyGen}}
\newcommand{\sigsign}{\mathsf{Sig}.\mathsf{Sign}}
\newcommand{\sigverify}{\mathsf{Sig}.\mathsf{Verify}}
\newcommand{\vk}{\mathsf{vk}}
\newcommand{\sigk}{\mathsf{sigk}}

\newcommand{\pkekeygen}{\mathsf{PKE}.\mathsf{KeyGen}}
\newcommand{\pkeenc}{\mathsf{PKE}.\mathsf{Enc}}
\newcommand{\pkedec}{\mathsf{PKE}.\mathsf{
Dec}}
\newcommand{\ek}{\mathsf{ek}}
\newcommand{\dk}{\mathsf{dk}}
\newcommand{\ct}{\mathsf{ct}}
\newcommand{\st}{\mathsf{st}}

\newcommand{\win}{\mathsf{win}}
\newcommand{\chal}{\mathcal{C}}
\newcommand{\inp}{\mathsf{inp}}

\tableofcontents
\thispagestyle{empty}
\clearpage
\pagenumbering{arabic}

\section{Introduction}

{\it Can $\NP$ search problems have a super-polynomial speed-up on quantum computers?} This is one of the oldest and most important questions in quantum complexity.

The first proposals for such quantum advantage were relative to highly structured oracles. Examples include Simon's oracle~\cite{Simon97}, or more generally periodic oracles, as well as the Bernstein–Vazirani oracle~\cite{STOC:BerVaz93} and welded trees~\cite{STOC:CCDFGS03}.

The first non-oracular quantum advantage for $\NP$ problems is due to Shor's famous algorithm for factoring integers and computing discrete logarithms~\cite{FOCS:Shor94}. Since Shor's algorithm, other non-oracular $\NP$ problems with quantum advantage include solving Pell's equation~\cite{STOC:Hallgren02} and matrix group membership~\cite{STOC:BabBeaSer09}. While the technical details of all these examples are very different, these problems can all be seen as non-oracular instantiations of \emph{periodic} oracles.

While the above non-oracular problems are certainly easy on a quantum computer, the classical hardness can only be conjectured since, in particular, the classical hardness would imply $\P\neq\NP$, or an analogous statement if one considers probabilistic algorithms. The problem is that, when instantiating an oracle with real-world computational tasks, non-black-box algorithms may be available that render the problem classically easy, despite the oracle problem being hard. For example, index calculus methods~\cite{Adleman79} yield sub-exponential time classical attacks for factoring and discrete logarithms, despite black box period-finding being classically exponentially hard.

To make matters worse, for the known $\NP$ search problems with plausible quantum advantage, the classical hardness is widely believed to be a much stronger assumption than $\P\neq\NP$, since the problems have significant algebraic structure and are not believed to be $\NP$-complete. In particular, all $\NP$ search problems we are aware of yielding a super-polynomial quantum advantage rely on \emph{Cryptomania} tools~\cite{Impagliazzo95}, in the sense that their classical hardness can be used to build public key encryption.\footnote{Matrix group membership includes discrete logarithms as a special case. For a public key system based on Pell's equations, see~\cite{EPRINT:Padhye06}.} This puts the assumptions needed for an $\NP$ quantum advantage quite high in the assumption hierarchy.

\paragraph{Quantum speed-ups and structure.} The above tasks demonstrating speed-ups, both oracular and non-oracular, all have one thing in common: significant ``structure.'' It is natural to wonder whether such structure is necessary. In the non-oracular setting, a natural interpretation of this question could be if Minicrypt assumptions---those that give symmetric key but not public key cryptography---can be used to give a quantum advantage. Minicrypt assumptions, such as the one-wayness of SHA2, lack the algebraic structure needed in typical super-polynomial quantum speed-ups. In the oracle setting, this could mean, for example, proving unconditional quantum advantage relative to a uniformly \emph{random} oracle, which is generally seen as beeing structure-less.

Prior work on this topic could be interpreted as negative. As observed above, all non-oracular $\NP$ problems demonstrating quantum advantage imply, or are closely related to problems that imply, public key cryptography. In the random oracle setting, the evidence is even stronger. The most natural problems to reason about---one-wayness and collision resistance of the random oracle, and generalizations---provably only have a polynomial quantum advantage~\cite{BBBV97,AS04,Yuen14,Zhandry15}. Additional evidence is given by Aaronson and Ambanis~\cite{AA14}, who build on work of Beals et al.~\cite{FOCS:BBCMW98}. They consider the following conjecture, dating back to at least 1999:

\begin{quote}\begin{conjecture}[Paraphrased from \cite{AA14}]\label{conj:aa14}Let $Q$ be a quantum algorithm with \emph{Boolean} output that makes $T$ queries to a random oracle $\oracle$, and let $\epsilon,\delta>0$. Then there exists a deterministic classical algorithm $C$ that makes $\poly(T,1/\epsilon,1/\delta)$ queries, such that 
\[\Pr_\oracle\left[\;\left|\;C^\oracle()-\Pr[Q^\oracle()=1]\;\right|\leq \epsilon\;\right]\geq 1-\delta\enspace ,\]
where the inner probability is over the randomness of $Q$.
\end{conjecture}
\end{quote}
Aaronson and Ambanis give some evidence for Conjecture~\ref{conj:aa14}, by reducing it to a plausible \emph{mathematical} conjecture closely related to known existing results. If Conjecture~\ref{conj:aa14} is true, any quantum \emph{decision} algorithm $Q$ making queries to a random oracle can be simulated classically with only polynomially-more queries. 

Note that the conjectured classical simulator may be \emph{computationally} inefficient, and indeed we would expect it to be if, say, $Q$ ignored its oracle and just factored integers. But for any particular algorithm $Q$, proving computational inefficiency amounts to an unconditional hardness result, which is beyond the reach of current complexity theory. Thus, Conjecture~\ref{conj:aa14}, if true, essentially shows that random oracles are equivalent to the non-oracular world with respect to $\NP$ decision problems, and cannot be used to provide provable quantum advantage for such problems.

\subsection{Our Results} In this work, we make progress toward justifying super-polynomial quantum advantage for $\NP$ problems, under less structured oracles or milder computational assumptions. We show, perhaps surprisingly, that for certain \emph{search} problems in $\NP$, random oracles do in fact give provable unconditional super-polynomial quantum speed-ups.

\paragraph{Random oracles.} Our starting point is to prove the following theorem:

\begin{theorem}[Informal]\label{thm:poqinf} Relative to a random oracle,  
there exists a non-interactive proof of quantumness, with unconditional security against any computationally-unbounded adversary making a polynomial number of classical queries.
\end{theorem}

Here, a proof of quantumness~\cite{FOCS:BCMVV18} is a protocol between a quantum prover and classical verifier (meaning in particular that messages are classical) where no cheating classical prover can convince the verifier. By being non-interactive, our protocol is also publicly verifiable. Prior LWE-based proofs of quantumness~\cite{FOCS:BCMVV18,BKVV20} lacked public verifiability. The only previous publicly verifiable proof of quantumness~\cite{STOC:AGKZ20} required highly non-trivial structured oracles.

\begin{remark}We note the restriction to uniform adversaries is necessary in the non-interactive setting, as a non-uniform adversary (that may take oracle-dependent advice) can simply have a proof hardcoded. Our protocol also readily gives a two-message public coin (and hence also publicly verifiable) protocol against non-uniform adversaries, which is the best one can hope for in the non-uniform setting.
\end{remark}

\noindent Theorem~\ref{thm:poqinf} has a number of interesting immediate consequences:
\begin{corollary}\label{cor:npinf} Relative to a random oracle, there exists an $\NP$ search problem that is solvable by 
quantum polynomial-time (QPT) machines but not by classical probabilistic polynomial-time (PPT) machines. 
\end{corollary}
Our construction also readily adapts to give one-way functions that are classically secure but quantum insecure. We can alternatively use minimal-round proofs of quantumness generically to give a one-way function counterexample, and even a collision resistance counterexample:
\begin{theorem}\label{thm:colinf}Relative to a random oracle, there exists a compressing function that is collision resistant against any computationally unbounded adversary making a polynomial number of classical queries, but is not even one-way against quantum adversaries. 
\end{theorem}
Using results from~\cite{EC:YamZha21}, we also obtain an unconditional analogous counterexample for digital signatures and CPA-secure public key encryption (the latter requiring assuming classically CPA-secure public key encryption). Previous such results required LWE (in the case of signatures) or highly structured additional oracles (in the case of CPA-secure encryption).


\medskip

Our results do not appear to contradict Conjecture~\ref{conj:aa14}, since they are for \emph{search} problems as opposed to \emph{decision} problems. In particular, our quantum algorithm for generating proofs of quantumness/breaking the one-wayness does not compute a function, but rather samples from a set of possible values. Assuming Conjecture~\ref{conj:aa14} shows that this is inherent. We leverage this feature to yield the following:
\begin{theorem}Assuming Conjecture~\ref{conj:aa14}, relative to a random oracle there exists a one- (resp. \mbox{two-)} message certifiable randomness protocol against a single uniform (resp. non-uniform) quantum device. By adding a final message from the verifier to the prover, our protocols become public coin and publicly verifiable.
\end{theorem}
Here, certifiable randomness~\cite{FOCS:BCMVV18} means the classical verifier, if it accepts, is able to expand a small random seed $s$ into a truly random bit-string $x,|x|\gg |s|$, with the aid of a single quantum device. Conditioned on the verifier accepting, $x$ remains truly random even if the device is adversarial. 
We remark that $|x|\gg |s|$ is the key property that makes certifiable randomness non-trivial: It enables the verifier to create a large random string $x$ from a much smaller random seed $s$.  In addition, we remark that the random seed $s$ is used only in the verifier's postprocessing for deriving $x$ and not used during the protocol execution in our construction.

We note that our results are the best possible: if the final message is from prover to verifier, the protocols cannot be publicly verifiable. Indeed, the prover could force, say, the first output bit to be 0 by generating a candidate final message, computing the what the outputted string would be, and then re-sampling the final message until the first output bit is 1. Our one- and two-message protocols therefore require verifier random coins that are kept from the prover. In our protocols, however, these secret random coins can be sampled and even published after the prover's message. The result is that, by adding a final message from the verifier, our protocols are public coin and publicly verifiable.

\paragraph{Instantiating the random oracle.} We next instantiate the random oracle in the above construction with a standard-model cryptographic hash, such as SHA2. We cannot hope to prove security unconditionally. Nevertheless, the resulting construction is quite plausibly secure. Indeed, it is common practice in cryptography to prove security of a hash-based protocol relative to random oracles~\cite{CCS:BelRog93}, and then assume that security also applies when the random oracle is replaced with a concrete well-designed cryptographic hash. While there are known counter-examples to the random oracle assumption~\cite{STOC:CanGolHal98}, they are quite contrived and are not known to apply to our construction.

We thus obtain a plausible construction of non-interactive proofs of quantumness based on a cryptographic hash, such as SHA2.
This gives a completely new approach to non-oracular quantum advantage. What's more, it is widely believed that SHA2 is only capable of yielding symmetric key cryptosystems. Impagliazzo and Rudich~\cite{STOC:ImpRud89} show that there is no classical black box construction of public key encryption from cryptographic hash functions, and no quantum or non-black box techniques are known to overcome this barrier\footnote{There is also some evidence that quantum black box techniques cannot overcome this barrier~\cite{C:ACCFLM22}.}. In fact, what~\cite{STOC:ImpRud89} show is that, in the world of computationally unbounded but query bounded (classical) attackers, random oracles cannot be used to construct public key encryption. But this is exactly the setting of the random oracle model we consider.

Therefore, by instantiating the random oracle with a well-designed hash such as SHA2, we obtain a Minicrypt construction of a proof of quantumness. We likewise obtain candidate Minicrypt examples of $\NP$ search problems in $\BQP\setminus\BPP$, functions that are classically one-way but quantumly easy, and even certifiable randomness.

\subsection{Discussion}

\paragraph{Other sources of quantum advantage.} Other candidates for super-polynomial quantum speed-ups are known. Aaronson and Arkhipov~\cite{STOC:AarArk11} and Bremner, Jozsa, and Shepherd~\cite{BJS10} give a sampling task with such a speed-up, based on plausible complexity-theoretic constructions. Similar sampling tasks are at the heart of current real-world demonstrations of quantum advantage. More recently, Brakerski et al.~\cite{FOCS:BCMVV18} provided a proof of quantumness from the Learning With Errors (LWE) assumption, 
Kalai et al. \cite{STOC:KLVY23} give a construction from general quantum homomorphic encryption, and 
Morimae and Yamakawa~\cite{ITCS:MorYam23} give a construction from general trapdoor permutations.

We note, however, that none of the these alternate sources of quantum advantage correspond to $\NP$ search problems, as there is no way to verify the output. In the case of~\cite{STOC:AarArk11,BJS10}, this is because the task is to sample from a distribution, and it is in general hard to tell if an algorithm samples from a given distribution. In the case of~\cite{FOCS:BCMVV18,STOC:KLVY23,ITCS:MorYam23}, this is due to the interactive protocols being private coin.

\paragraph{Why $\NP$ search problems?} Most real-life problems of interest can be phrased as $\NP$ search problems, so it is a natural class of problems to study. Our work gives the first evidence besides period finding of a quantum advantage for this class.

Moreover, $\NP$ means that solutions can be efficiently verified. For existing sampling-based demonstrations of quantum advantage~\cite{STOC:AarArk11,BJS10}, verification is roughly as hard as classically sampling. Proofs of quantumness from cryptographic assumptions~\cite{FOCS:BCMVV18,STOC:KLVY23,ITCS:MorYam23} do admit verification, but the verifier must use certain secrets computed during the protocol in order to verify. This means that only the verifier involved in the protocol is convinced of the quantumness of the prover.

In contrast, using an $\NP$ problem means anyone can look at the solution and verify that it is correct. Moreover, our particular instantiation allows for sampling the problems obliviously, meaning we obtain a \emph{public coin} proof of quantumness where the verifier's message is simply uniform random coins. Against uniform adversaries, we can even just set the verifier's message to $000\cdots$, eliminating the verifier's message altogether.

\paragraph{The QROM} In classical cryptography, the Random Oracle Model (ROM)\cite{CCS:BelRog93} models a hash function as a truly random function, and proves security in such a world. This model is very important for providing security justifications of many practical cryptosystems.

Boneh et al.~\cite{AC:BDFLSZ11} explain that, when moving to the quantum setting, one needs to model the random oracle as a \emph{quantum random oracle model} (QROM). Many works (e.g.~\cite{C:Zhandry12,TCC:TarUnr16,EC:SaiXagYam18,EC:KilLyuSch18,AC:KatYamYam18,C:LiuZha19,C:DFMS19,TCC:ChiManSpo19}) have been devoted to lifting classical ROM results to the QROM. Ambainis, Rosmanis, and Unruh \cite{FOCS:AmbRosUnr14} demonstrated that some random-oracle-based constructions that are known to be secure against classical adversaries are insecure against quantum adversaries. However, their counterexamples are insecure even against quantum adversaries in the classical ROM (i.e., those that only make classical queries), and thus they do not indicate a difference between the classical ROM and QROM.  
To date, most of the main classical ROM results have successfully been lifted. 
This leads to a natural question: do all ROM results lift to the QROM?

Recently, Yamakawa and Zhandry~\cite{EC:YamZha21}, leveraging recent proofs of quantumness~\cite{BKVV20} in the random oracle, give a counter-example assuming the hardness of learning with errors (LWE). Their counter-examples were limited to highly interactive security models such as digital signatures and CCA-secure public key encryption.

By relying on LWE, \cite{EC:YamZha21} left open the possibility that \emph{unconditional} ROM results may all lift to the QROM. Our proof of quantumness refutes this, showing that the ROM and QROM are separated even in the unconditional setting. Our results also give counterexamples for many more objects, especially for objects like one-way functions and collision resistance which have essentially non-interactive security experiments.

\paragraph{Subsequent work.} Our techniques have already been used in many subsequent works. Liu~\cite{EC:Liu23} uses our construction to give an exponential separation between classical and quantum advice, relative to a random oracle. 
Li, Liu, Pelecanos, Yamakawa~\cite{ITCS:LLPY24} and Ben-David and Kundu~\cite{ICALP:BenKun24} extended this idea to show a separation between $\mathsf{QMA}$ and $\mathsf{QCMA}$ relative to a classical oracle in restricted models.
Arora et al.~\cite{STOC:ACCGSW23} use our construction to give a proof of quantum depth relative to random oracles.
Jordan et al.~\cite{jordan2024optimizationdecodedquantuminterferometry} extend our idea to give a new quantum algorithm for optimization problems. 
Göös et al.~\cite{GGJL24} use our construction to show a new quantum advantage in the context of communication complexity.
Li~\cite{ITCS:Li24} and Jain et al.~\cite{FOCS:JLRX24} study the complexity of our problem in terms of subclasses of $\mathsf{TFNP}$.  
\subsection{Overview}

Let $\Sigma$ be an exponentially-sized alphabet, and $C\subseteq \Sigma^n$ be an error correcting code over $\Sigma$. Let $O:\Sigma\rightarrow\{0,1\}$ be a function. Consider the following function $f_C^O:C\rightarrow\{0,1\}^n$ derived from $C,O$:
\[
    f_C^O(c_1,\dots,c_n)=(O(c_1),\dots,O(c_n))
\]
In other words, $f_C^O$ simply applies $O$ independently to each symbol in the input codeword. We will model $O$ as a uniformly random function. Note that if $f$ were applied to arbitrary words in $\Sigma^n$, then it would just be the parallel application of a function with one-bit outputs, which can be trivially inverted. By restricting the domain to only codewords, we show, under a suitable choice of code elaborated on below, that:
\begin{itemize}
    \item $f_C^O$ is unconditionally one-way against classical probabilistic algorithms making polynomially-many queries to $O$. It is even infeasible to find $c\in C$ such that $f_C^O(c)=0^n$.
    \item There exists a quantum algorithm which, given any $y\in\{0,1\}^n$, samples statistically close to uniformly from the set of pre-images $c\in C$ such that $f_C^O(c)=y$.
\end{itemize}
From these properties, we immediately obtain a weak version of Theorem~\ref{thm:colinf} which only considers classical one-wayness. We explain in Section~\ref{sec:CRH} how to obtain the full Theorem~\ref{thm:colinf}. To prove quantumness, one simply produces $c\in C$ such that $f_C^O(c)=0^n$, giving Theorem~\ref{thm:poqinf}. Since inverting one-way functions is in $\NP$, this also immediately gives Corollary~\ref{cor:npinf}. We now explain how we justify these facts about $f_C^O$.

\paragraph{Classical hardness.} Assume $C$ satisfies the following properties: (1) the set of symbols obtained at each position are distinct, and (2) $C$ is information-theoretically {\bf list-recoverable}.\footnote{List-recoverable codes have been used in cryptography in the contexts of domain extension of hash functions \cite{C:HIOS15,EC:KomNaoYog18,STOC:BitKalPan18} and the Fiat-Shamir transform \cite{STOC:HolLomRot21}.} 
Here, we take list-recoverability to mean that, given polynomial-sized sets $S_i,i\in[n]$ of possible symbols for each position, there exist a sub-exponential sized (in $n$) list of codewords $c$ such that $c_i\in S_i$ for all $i\in [n]$. The list size remains sub-exponential even if we include codewords such that $c_i\notin S_i$ for a few positions.

Property (1) can be obtained generically by replacing $\Sigma\mapsto[n]\times \Sigma$, where $(c_1,\dots,c_n)\mapsto ((1,c_1),\dots,(n,c_n))$. Property (2) is satisfied by folded Reed-Solomon codes, as shown by Guruswami and Rudra~\cite{GR08}. 

Assuming (1) and (2), we can show classical hardness. Fix an image $y$. We can assume without loss of generality that the adversary always evaluates $f_C^O(c)$ for any pre-image $c$ it outputs. Suppose for our discussion here that all queries to $O$ were made in parallel. Then any polynomial-sized set of queries corresponds to a collection of $S_i$. List recoverability means that there are at most $2^{n^c},c<1$ codewords consistent with the $S_i$. For each consistent codeword, the probability of being a pre-image of $y$ is at most $2^{-n}$ over the choice of random oracle. Union-bounding over the list of consistent codewords shows that the probability that \emph{any} consistent codeword is a pre-image is exponentially small. With some effort, we can show the above holds even for adaptively chosen queries.

\begin{remark}Haitner et al.~\cite{C:HIOS15} construct a very similar hash function from list-recoverable codes. Their hash functions assumes a multi-bit $O$, but then XORs the results together, rather than concatenating them. 
They prove that their hash function is collision-resistant. Our proof of one-wayness is based on a similar idea to their proof of collision-resistance.
Our novelty, and what does not appear to be possible for their construction, is the quantum pre-image finder, which we discuss next. 

We note that we could, similar to~\cite{C:HIOS15}, prove the collision resistance of $f_C^O$ by choosing $C$ to have an appropriate rate. However, our quantum pre-image finder constrains $C$ to having a rate where we only know how to prove one-wayness. Proving Theorem~\ref{thm:colinf} therefore requires a different construction, which we elaborate on in Section~\ref{sec:CRH}.
\end{remark}

\paragraph{Quantum easiness.} Our algorithm can be seen as loosely inspired by Regev's quantum reduction between SIS and LWE~\cite{STOC:Regev05}. Given an image $y$, our goal will be to create a uniform superposition over pre-images of $y$:
\[|\psi_y\rangle\propto \sum_{c\in C:f_C^O(c)=y}|c\rangle\]
We can view $|\psi_y\rangle$ as the point-wise product of two vectors:
\begin{align*}
    |\phi\rangle&\propto \sum_{c\in C}|c\rangle\;,\;\;\;\;\;\;\;\;\;\;\text{ and }\;\;\;\;\;\;\;\;\;\;
    |\tau_y\rangle\propto \sum_{c\in{\mathbf{\Sigma^n}}:f_C^O(c)=y}|c\rangle
\end{align*}
Observe that $|\tau_y\rangle$ looks like $|\psi_y\rangle$, except that the domain is no longer constrained to codewords. Once we have the state $|\psi_y\rangle$, we can simply measure it to obtain a random pre-image of $y$. We will show how to construct $|\psi_y\rangle$ in reverse: we will show a sequence of reversible transformations that transform $|\psi_y\rangle$ into states we can readily construct. By applying these transformations in reverse we obtain $|\psi_y\rangle$. To do so, we will now impose that $\Sigma$ is a vector space over $\FF_q$ for some prime $q$, and that $C$ is {\bf linear} over $\FF_q$.\footnote{In the main body, we use an extension field $\FF_q$ (i.e., $q$ is a prime power) for an appropriate parameter choice, but one can think of it as a prime field for the purpose of this overview.} This means there is a dual code $C^\perp$, such that $c\cdot d=0$ for all $c\in C,d\in C^\perp$.

We now consider the quantum Fourier transform $\QFT$ of $|\psi_y\rangle$.\footnote{
Note that an element of $\Sigma^n$ can be written as a vector over $\FF_q$. Here, we simply write $\QFT$ to mean the operation that applies QFT over the additive group of $\FF_q$ for each coordinate.} Write:
\begin{align*}
|\widehat{\phi}\rangle&:=\QFT|\phi\rangle\propto \sum_{c\in\Sigma^n}\alpha_c |c\rangle=\sum_{c\in C^\perp}|c\rangle\\
|\widehat{\tau_y}\rangle&:=\QFT|\tau_y\rangle\propto \sum_{c\in\Sigma^n}\beta_{y,c} |c\rangle\\
\end{align*}
Above, we used the fact that the QFT of a uniform superposition over a linear space is just the uniform superposition over the dual space. Then, by the Convolution Theorem, the QFT of $|\psi_y\rangle$ is the convolution of $|\widehat{\phi}\rangle$ and $|\widehat{\tau}_y\rangle$:
\[|\widehat{\psi}_y\rangle:=\QFT|\psi_y\rangle\propto \sum_{c,e\in\Sigma^n}\alpha_c\beta_{y,e}|c+e\rangle=\sum_{c\in C^\perp,e\in\Sigma^n}\beta_{y,e} |c+e\rangle\]
The next step is to decode $c$ and $e$ from $c+e$; assuming we had such a decoding, we can apply it to obtain the state proportional to
\[\sum_{c\in C^\perp,e\in\Sigma^n}\beta_{y,e}|c,e\rangle=|\widehat{\phi}\rangle|\widehat{\tau}_y\rangle\]
We can then construct $|\widehat{\phi}\rangle$ as the QFT of $|\phi\rangle$, which we can generate using the generator matrix for $C$. We will likewise construct $|\widehat{\tau}_y\rangle$ as the QFT of $|\tau_y\rangle$. To construct $|\tau_y\rangle$, we note that $|\tau_y\rangle$ is a product of $n$ states that look like:
\[|\tau_{i,y_i}\rangle\propto\sum_{\sigma\in\Sigma:O(\sigma)=y_i}|\sigma\rangle\]
Since each $y_i$ is just a single bit, we can construct such states by applying $O$ to a uniform superposition of inputs, measuring the result, and starting over if we obtain the incorrect $y_i$.

\medskip

It remains to show how to decode $c,e$ from $c+e$. We observe that $|\widehat{\tau}_{i,y_i}\rangle$ has roughly half of its weight on 0, whereas the remaining half the weight is essentially uniform (though with complex phases) since $O$ is a random function. This means we can think of $e$ as a vector where each symbol is 0 with probability $1/2$, and random otherwise. In other words, $c+e$ is a noisy version of $c$ following an analog of the binary symmetric channel generalized to larger alphabets. If the dual code $C^\perp$ were efficiently decodable under such noise, then one can decode $c$ (and hence $e$) from $c+e$.

Toward that end, we show that $c$ is uniquely and efficiently decodable (with high probability) provided the rate of $C^\perp$ is not too high. In our case where $C$ is a folded Reed-Solomon code, $C^\perp$ is essentially another Reed-Solomon code, and we can decode efficiently using {\bf list-decoding} algorithms~\cite{GS99}. We can show that the list-decoding results in a unique codeword (with high probability) for the above described error distribution assuming $C$ to have an appropriate rate.

There are a couple important caveats with the above. First is that, to use list-recoverability to prove one-wayness, we actually needed to augment $C$, which broke linearity. This is easily overcome by only applying the QFT to the linear part of $C$.

More importantly, and much more challenging, we can only decode $c+e$ as long as $e$ has somewhat small Hamming weight. While such $e$ occur with overwhelming probability, there will always be a negligible fraction of decoding errors. The problem is that the constant of proportionality in the Convolution Theorem is exponentially large, and therefore the negligible decoding errors from our procedure could end up being blown up and drowning out $|\widehat{\psi}_y\rangle$. This is not just an issue with our particular choice of decoding algorithm, as for large enough Hamming weight decoding errors are guaranteed. What this means is that the map $|\widehat{\phi}\rangle|\widehat{\tau}_y\rangle\mapsto |\widehat{\psi}_y\rangle$ is not even unitary, and $|\widehat{\psi}_y\rangle$ is not even unit norm. 

By exploiting the particular structure of our coding problem and the uniform randomness of the oracle $O$, we are able to resolve the above difficulties and show that our algorithm does, in fact, produce pre-images of $y$ as desired.

\paragraph{Certifiable randomness.} We next explain that \emph{any} efficient quantum algorithm for inverting $f_C^O$ likely produces random pre-images. After all, suppose there was an alternative quantum algorithm which inverted $f_C^O$, such that it finds a deterministic pre-image on any given $y$. If we look at any single bit of the pre-image, then Conjecture~\ref{conj:aa14} would imply that this bit can be simulated by a polynomial-query classical algorithm. By applying Conjecture~\ref{conj:aa14} to every bit of the pre-image, we thus obtain a classical query algorithm for inverting $f_C^O$, which we know is impossible.

This immediately gives us a proof of entropy: the prover generates a pre-image $c$ of an arbitrary $y$ (even $y=0^n$), and the verifier checks that $f_C^O(c)=y$. If the check passes, the verifier can be convinced that $c$ was not deterministically generated, and therefore has some randomness. Though this only ensures that $c$ is not completely deterministic, 
by using the fact that $f_C^O$ is one-way even against sub-exponential-query algorithms, we can extend the above argument to show that the min-entropy must be polynomial. 

Once we have a string with min-entropy, we can easily get uniform random bits by having the verifier extract using a private random seed.

\paragraph{Extension to non-uniform adversaries.} Note that the above results all considered fixing an adversary first, and then sampling a random oracle. A standard complexity theoretic argument shows that, in the case of uniform adversaries, we can switch the order of quantifiers, and choose the random oracle first and then the adversary.

For non-uniform adversaries, we have to work harder, and direct analogs of the results above may in fact be impossible: for example, a non-uniform adversary (chosen after the random oracle) could have a valid proof of quantumness hardcoded. 

For proofs of quantumness, we can leverage the ``salting defeats preprocessing'' result of~\cite{EC:CDGS18,FOCS:CGLQ20} to readily get a two-message public coin proof of quantumness against non-uniform attackers. For certifiable entropy/randomness, this also works, except the known bounds would end up requiring the verifiers message to be longer than the extracted string. This is a consequence of leveraging the sub-exponential one-wayness of $f_C^O$ to obtain polynomially-many random bits. Since the verifier's message must be uniform, this would somewhat limit the point of a proof of randomness. We show via careful arguments how to overcome this limitation, obtaining two message proofs of randomness where the verifier's message remains small in the classical advice setting. We leave it open to extend our result to construct proofs of randomness that are secure against non-uniform adversaries with quantum advice.

\paragraph{Extension to worst-case completeness.} Our analysis of the quantum algorithm seems to inherently rely on the oracle being uniformly random.  We show how to tweak our scheme so that correctness holds for \emph{any} oracle. The idea is to set $O=O'\oplus P$, where $O'$ is the oracle, and where $P$ is a $k$-wise independent function for some sufficiently large $k$. The point is that $P$ is supplied as part of the problem solution, and so is chosen by the quantum algorithm. This makes $O$ $k$-wise independent regardless of $O'$, which is sufficient for the analysis.

Of course, introducing $P$ makes the classical problem easier, since now the classical adversary has some flexibility in constructing $O$. We handle this by asking the adversary to find many solutions relative to different $O'$, but the same $P$. This amplifies hardness, after which we can union-bound over all possible $P$ and still maintain classical hardness. The quantum algorithm, on the other hand, can solve each of the individual instances with high probability, so it can easily solve all instances.

This gives the following conceptual implication: 
By regarding the oracle as an $N=2^n$-bit input, 
we obtain a relational problem $R\subseteq \{0,1\}^N \times \{0,1\}^{m}$ 
for $m=\poly(n)$ such that 
\begin{enumerate}
\item 
$R$ is classically efficiently verifiable, i.e., 
we can test if $(x,w)\in R$ given $w$ and $\poly(n)$ classical queries to $x$, and  
\item finding $w$ such that $(x,w)\in R$ is easy with $\poly(n)$ quantum queries to for all $x$ but hard with $\poly(n)$ classical queries on average over $x$.  
\end{enumerate} 

Note that this is a slightly different setting than our $\NP$ relation above, where the instances and witnesses were both polynomial-length strings, and the oracle is used to determine which witnesses are valid for a given instance.

\subsection{Acknowledgements} We thank Scott Aaronson for helpful suggestions, including the conceptual implication of worst-case completeness. We thank anonymous reviewers of FOCS 2022, QIP 2023, and Journal of the ACM for their helpful comments. Mark Zhandry is supported in part by an NSF CAREER award.

\subsection{Organization}

The remainder of the paper is organized as follows. Section~\ref{sec:prelim} gives some basic preliminaries, including for quantum computation. Section~\ref{sec:ROM} defines the various objects we will be considering and gives some basic relations. Section~\ref{sec:codes} discusses the properties of error correcting codes we will need. Section~\ref{sec:technical_lemma} gives a technical lemma that is needed to prove the correctness of our protocol, that may be more broadly useful. Section~\ref{sec:PoQ} gives our proof of quantumness, while Section~\ref{sec:separation_primitive} uses this to give counterexamples for various cryptographic primitives. Finally, Section~\ref{sec:PoR} gives our proofs of randomness.

\section{Preliminaries}\label{sec:prelim}
\paragraph{Basic notations.}
We use $\secpar$ to mean the security parameter throughout the paper.
For a set $X$, $|X|$ is the cardinality of $X$.
For a non-empty finite set $X$, 
we denote by $x\sample X$ to mean that $x$ is uniformly taken from $X$.
For a distribution $D$ over a set $X$, we denote by $x\sample D$ to mean that $x\in X$ is taken according to the distribution $D$. 
For sets $\calX$ and $\calY$, $\Func(\calX,\calY)$ denotes the set of all functions from $\calX$ to $\calY$.
For a positive integer $n$, $[n]$ means a set $\{1,...,n\}$.
For a random variable $X$, $\Ex[X]$ denotes its expected value.
For random variables $X$ and $X'$, $\Delta(X,X')$ denotes the statistical distance between $X$ and $X'$.
For a random variable $X$, $H_\infty(X)$ denotes the min-entropy of $X$, i.e., $H_\infty(X)=-\log\max_{x}\Pr[X=x]$.
For a quantum or randomized classical algorithm $\A$, we denote $y\sample \A(x)$ to mean that $\A$ outputs $y$ on input $x$.  
For a randomized classical algorithm $\A$, we denote $y\leftarrow \A(x;r)$ to mean that $\A$ outputs $y$ on input $x$ and randomness $r$. 

\paragraph{Notations for quantum states.}
For a not necessarily normalized state $\ket{\psi}$, we denote by $\|\ket{\psi}\|$ to mean its Euclidean norm. 
For not necessarily normalized quantum states $\ket{\psi}$ and $\ket{\phi}$ and $\epsilon>0$, we denote by $\ket{\psi}\approx_{\epsilon} \ket{\phi}$ to mean $\|\ket{\psi}-\ket{\phi}\|\leq \epsilon$.
We simply write $\ket{\psi}\approx \ket{\phi}$ to mean $\ket{\psi}\approx_{\negl(\secpar)} \ket{\phi}$.
By the triangle inequality, if we have $\ket{\psi}\approx_\epsilon \ket{\phi}$ and $\ket{\phi}\approx_\delta \ket{\tau}$, then we have $\ket{\psi}\approx_{\epsilon+\delta} \ket{\tau}$. 

For not necessarily normalized quantum states $\ket{\psi}$ and $\ket{\phi}$, we denote by $\ket{\psi}\propto \ket{\phi}$ to mean that $\ket{\psi}=C\ket{\phi}$ for some $C\in \mathbb{C}\setminus \{0\}$. 

\subsection{Finite Fields}\label{sec:finite_field}
For a prime power $q=p^r$, $\FF_q$ denotes a field of order $q$. We use this notation throughout the paper, and whenever we write $\FF_q$, $q$ should be understood as a prime power. 
We denote by $\vzero$ to mean $(0,...,0)\in \FF_q^n$ where $n$ will be clear from the context. 
For $\vx=(x_1,...,x_n)\in \FF_q^n$ and $\vy=(y_1,...,y_n)\in \FF_q^n$, we define $\vx\cdot \vy\defeq \sum_{i=1}^{n}x_iy_i$.

We often consider vectors $\vx\in \Sigma^n$ over the alphabet $\Sigma=\FF_q^m$. 
We identify $\Sigma^n$ and $\FF_q^{nm}$ in the canonical way, i.e., we identify $((x_1,\ldots,x_m),\ldots,(x_{(n-1)m+1},\ldots,x_{nm}))\in \Sigma^n$ and $(x_1,x_2,\ldots, x_{nm})\in \FF_q^{nm}$. 
For $\vx=(\vx_1,...,\vx_n)\in \Sigma^n$ and $\vy=(\vy_1,...,\vy_n)\in \Sigma^n$, we define $\vx\cdot \vy\defeq \sum_{i=1}^{n}\vx_i\cdot \vy_i$.

The trace function $\Tr:\FF_q\rightarrow \FF_p$ is defined by\footnote{It may not be immediately clear from the definition below that $\Tr(x)\in\FF_p$, but this is a well-known fact~\cite{1997-lidl}.} 
\begin{align*}
    \Tr(x)\defeq \sum_{i=0}^{r-1}x^{p^{i}}. 
\end{align*}
The trace function is $\FF_p$-linear, i.e., for any $a,b\in \FF_p$ and $x,y\in \FF_q$, we have 
\begin{align*}
    \Tr(ax+by)=a\Tr(x)+b\Tr(y).
\end{align*}
We let $\omega_p\defeq e^{2\pi i/p}$. 
For any $\vx\in \FF_q^n \setminus \{\vzero\}$, we have 
\begin{align}\label{eq:sum_is_zero}
\sum_{\vy\in \FF_q^n}\omega_p^{\Tr(\vx\cdot \vy)}=0.     
\end{align}
The multiplicative group $\FF_q^*$ of $\FF_q$ is cyclic, and thus there is an element $\gamma\in \FF_q^*$ such that 
\begin{align*}
    \{\gamma^{i}\}_{i\in [q-1]}=\FF_q^*.
\end{align*} 
For $\vx\in \FF_q^n$, we denote by $\hw(\vx)$ to mean the Hamming weight of $\vx$, i.e., 
$\hw(\vx)\defeq |\{i\in [n]: x_i \neq 0\}|$ where $\vx=(x_1,\ldots,x_n)$.
For $\vx=(x_1,\ldots,x_n) \in \FF_q^n$ and a subset $S\subseteq [n]$, we denote by $\vx_S$ to mean $(x_i)_{i\in S}$.

\subsection{Quantum Fourier Transform over Finite Fields}
We review known facts on quantum Fourier transform over finite fields. 
On a quantum system over a finite field $\FF_q$, a quantum Fourier transform is a unitary denoted by $\QFT_{\FF_q}$ such that for any $x\in \FF_q$, 
\begin{align*}
    \QFT_{\FF_q}\ket{x} = \frac{1}{\sqrt{q}}\sum_{z\in \FF_q}\omega_p^{\Tr(x \cdot z)}\ket{z}. 
\end{align*} 
A quantum Fourier transform over $\FF_q$ can be approximated to within error $\epsilon$ in time polynomial in $\log q$ and $\log 1/\epsilon$~\cite{BCW02,DHI06}. For ease of exposition, we ignore the approximation error in the rest of the paper since it can be made exponentially small by a polynomial-size quantum circuit.  

We often consider quantum systems over the alphabet $\Sigma=\FF_q^m$ for some positive integer $m$. We define the QFT over $\Sigma$ to be the $m$-tensor product of $\QFT_{\FF_q}$: For $\vx=(x_1,...,x_m)\in \Sigma$, 
\begin{align*}
    \QFT_{\Sigma}\ket{\vx}&:=
    \QFT_{\FF_q}^{\otimes m}\ket{x_1}\ket{x_2}...\ket{x_m}\\
    &=\frac{1}{\sqrt{|\Sigma|}}\sum_{\vz\in \Sigma}\omega_p^{\Tr(\vx \cdot \vz)}\ket{\vz}
\end{align*} 
where the second equality follows from the definition of $\QFT_{\FF_q}$ and linearity of $\Tr$. 
Similarly, for any positive integer $n$ and $\vx\in \Sigma^n$, we have 
\begin{align*}
    \QFT_{\Sigma}^{\otimes n}\ket{\vx} = \frac{1}{|\Sigma|^{n/2}}\sum_{\vz\in \Sigma^n}\omega_p^{\Tr(\vx \cdot \vz)}\ket{\vz}
\end{align*} 
by the definition of $\QFT_{\Sigma}$ and linearity of $\Tr$. 

For a function $f:\Sigma^n\ra \mathbb{C}$, we define 
\begin{align*}
    \hat{f}(\vz)\defeq \frac{1}{|\Sigma|^{n/2}}\sum_{\vx\in \Sigma^n}f(\vx)\omega_p^{\Tr(\vx \cdot \vz)}. 
\end{align*}
Then it is easy to see that we have 
\begin{align*}
    \QFT_{\Sigma}^{\otimes n}\sum_{\vx\in \Sigma^n} f(\vx)\ket{\vx} = \sum_{\vz\in \Sigma^n} \hat{f}(\vz)\ket{\vz}.
\end{align*}
For functions $f:\Sigma^n\ra \mathbb{C}$ and $g:\Sigma^n\ra \mathbb{C}$, $f\cdot g$ and $f\ast g$ denote the point-wise product and convolution of $f$ and $g$, respectively, i.e.,
\begin{align*}
&(f\cdot g)(\vx)\defeq f(\vx)\cdot g(\vx)\\
    &(f\ast g)(\vx)\defeq \sum_{\vy\in \Sigma^n}f(\vy)\cdot g(\vx-\vy).
\end{align*}

We have the following standard lemmas.
We include the proofs for completeness.

\begin{lemma}[Parseval's equality]\label{lem:Parseval}
For any $f:\Sigma^n\ra \mathbb{C}$, we have 
\begin{align*} 
    \sum_{\vx\in \Sigma^n}|f(\vx)|^2=\sum_{\vz\in \Sigma^n}|\hat{f}(\vz)|^2.
\end{align*}
\end{lemma}
\begin{proof}
Since $\QFT_{\FF_q}$ is unitary,  $\QFT_{\Sigma}^{\otimes n}$ is also unitary. 
This immediately implies \Cref{lem:Parseval}.
\end{proof}

\begin{lemma}\label{lem:QFT_prod}
Let $m$ be a positive integer that divides $n$. 
Suppose that we have $f_i:\Sigma \rightarrow \mathbb{C}$ for $i\in [n]$ and $f:\Sigma^n\ra \mathbb{C}$ is defined by
\begin{align}\label{eq:f_product}
    f(\vx)\defeq \prod_{i\in [n]}f_i(\vx_i)
\end{align}
where $\vx=(\vx_1,\vx_2,...,\vx_n)$. 
Then, we have 
\begin{align*}
    \hat{f}(\vz)=\prod_{i\in [n]}\hat{f}_i(\vz_i)
\end{align*}
where $\vz=(\vz_1,\vz_2,...,\vz_n)$. 
\end{lemma}
\begin{proof}
This can be proven by the following equalities:
\begin{align*}
\hat{f}(\vz)
&=\frac{1}{|\Sigma|^{n/2}}\sum_{\vx\in \Sigma^n}f(\vx)\omega_p^{\Tr(\vx\cdot \vz)}\\
&=\frac{1}{|\Sigma|^{n/2}}\sum_{\vx_1\in \Sigma}...\sum_{\vx_{n}\in \Sigma}\prod_{i\in [n]} f_i(\vx_i)\omega_p^{\Tr(\vx_i\cdot \vz_i)}\\
&=\prod_{i\in [n]}\frac{1}{|\Sigma|^{1/2}}\sum_{\vx_i\in \Sigma}f_i(\vx_i)\omega_p^{\Tr(\vx_i\cdot \vz_i)}\\
&=\prod_{i\in [n]}\hat{f}_i(\vz_i)
\end{align*}
where the second equality follows from \Cref{eq:f_product} and the linearity of $\Tr$.
\end{proof}

\begin{lemma}[Convolution theorem]\label{lem:convolution}
For functions $f:\Sigma^n\ra \mathbb{C}$, $g:\Sigma^n\ra \mathbb{C}$, and $h:\Sigma^n\ra \mathbb{C}$, the following equations hold.
\begin{align}\label{eq:conv_one}
    \widehat{f\cdot g} = \frac{1}{|\Sigma|^{n/2}}(\hat{f} \ast \hat{g}), 
\end{align}
\begin{align}\label{eq:conv_two}
    \widehat{f\ast g} = |\Sigma|^{n/2}(\hat{f} \cdot \hat{g}),
\end{align}
\begin{align}\label{eq:fgh}
    \widehat{f\cdot(g\ast h)} = (\hat{f} \ast (\hat{g}\cdot \hat{h})).
\end{align}
\end{lemma}
\begin{proof}

For any $\vx\in \Sigma^n$, we have 
\begin{align*}
(\hat{f}\ast \hat{g})(\vx)
&=\sum_{\vy\in\Sigma^n}\hat{f}(\vy)\hat{g}(\vx-\vy)\\
&=\sum_{\vy\in\Sigma^n}\left(\frac{1}{|\Sigma|^{n/2}}\sum_{\vz\in\Sigma^n}f(\vz)\omega_p^{\Tr(\vy\cdot \vz)}\right)\left(\frac{1}{|\Sigma|^{n/2}}\sum_{\vz'\in\Sigma^n}g(\vz')\omega_p^{\Tr((\vx-\vy)\cdot \vz')}\right)\\
&=\frac{1}{|\Sigma|^{n}}\sum_{\vy\in\Sigma^n}\sum_{\vz\in\Sigma^n}\sum_{\vz'\in\Sigma^n}f(\vz)g(\vz')\omega_p^{\Tr(\vx\cdot \vz')}\omega_p^{\Tr(\vy\cdot (\vz-\vz'))}\\
&=\frac{1}{|\Sigma|^{n}}\sum_{\vz\in\Sigma^n}\sum_{\vz'\in\Sigma^n}\left(f(\vz)g(\vz')\omega_p^{\Tr(\vx\cdot \vz')}\sum_{\vy\in\Sigma^n}\omega_p^{\Tr(\vy\cdot (\vz-\vz'))}\right)\\
&=\frac{1}{|\Sigma|^{n}}\cdot |\Sigma|^{n} \sum_{\vz\in\Sigma^n}f(\vz)g(\vz)\omega_p^{\Tr(\vx\cdot \vz)}\\
&=\sum_{\vz\in\Sigma^n}f(\vz)g(\vz)\omega_p^{\Tr(\vx\cdot \vz)}\\
&=|\Sigma|^{n/2}(\widehat{f \cdot g})(\vx)
\end{align*}
where the third equality follows from the linearity of $\Tr$ and the fifth equality follows from \Cref{eq:sum_is_zero}.
This implies \Cref{eq:conv_one}.

For any $\vx\in \Sigma^n$, we have 
\begin{align*}
\widehat{(f \ast g)}(\vx)
&=\frac{1}{|\Sigma|^{n/2}}\sum_{\vz\in \Sigma^n}(f \ast g)(\vz)\omega_p^{\Tr(\vx\cdot \vz)}\\
&=\frac{1}{|\Sigma|^{n/2}}\sum_{\vz\in \Sigma^n}\sum_{\vy\in \Sigma^n}f(\vy)g(\vz-\vy)\omega_p^{\Tr(\vx\cdot \vy)}\omega_p^{\Tr(\vx\cdot (\vz-\vy))}\\
&=\frac{1}{|\Sigma|^{n/2}}\left(\sum_{\vy\in \Sigma^n}f(\vy)\omega_p^{\Tr(\vx\cdot \vy)}\right)\left(\sum_{\vz'\in \Sigma^n}g(\vz')\omega_p^{\Tr(\vx\cdot \vz')}\right)\\
&=|\Sigma|^{n/2}(\hat{f}\cdot \hat{g})(\vx)
\end{align*}
where the second equality follows from the linearity of $\Tr$.
This implies \Cref{eq:conv_two}.
\Cref{eq:fgh} immediately follows from \Cref{eq:conv_one,eq:conv_two}. 
\end{proof}


\subsection{Other Lemmas}
We rely on the following well-known lemmas.
\begin{lemma}[Chernoff Bound]\label{lem:Chernoff}
Let $X_1,...,X_n$ be independent random variables taking values in $\bit$, $X\defeq \sum_{i\in [n]}X_i$, and $\mu\defeq \Ex[X]$.
For any $\delta\geq 0$, it holds that 
\begin{align*}
    \Pr[X\geq (1+\delta)\mu]\leq e^{-\frac{\delta^2 \mu}{2+\delta}}.
\end{align*}
\end{lemma}

\begin{lemma}[\cite{C:Zhandry12}]\label{lem:simulation_QRO}
For any sets $\calX$ and $\calY$ of classical strings and $q$-quantum-query algorithm $\A$, we have
\[
\Pr[\A^{H}=1:H\sample \Func(\calX,\calY)]= \Pr[\A^{H}=1:H\sample \mathcal{F}]
\]
where  $\mathcal{F}$ is a family of $2q$-wise independent hash functions from $\calX$ to $\calY$.
\end{lemma}
\section{Cryptographic Definitions in the Random Oracle Model}\label{sec:ROM}

Here, we define various cryptographic notions we will be constructing. We consider the following variations of the random oracle model.  
\begin{itemize}
    \item {\bf Classical random oracle model (CROM)~\cite{CCS:BelRog93}.} 
    In this model, a uniformly random function $H:\bit^n \rightarrow \bit^m$ is chosen at the beginning where $n=n(\secpar)$ and $m=m(\secpar)$ are polynomials in the security parameter $\secpar$ (that may vary depending on the protocol), and the adversary is allowed to make classical queries to $H$.\footnote{The classical random oracle model is often just referred to as the ROM, but we call it CROM to emphasize that the oracle access is classical.} When we consider probabilities over the random oracle $H$, it should be understood to be uniformly chosen from the set of all functions from $\bit^n$ to $\bit^m$ unless otherwise stated.  
    We often refer to adversaries in the CROM as uniform classical adversaries. 
    \item {\bf Quantum random oracle model (QROM)~\cite{AC:BDFLSZ11}.} This is identical to the CROM except that queries to $H$ can now be quantum.  In other words, a quantum oracle that applies a unitary $\ket{x}\ket{y}\mapsto \ket{x}\ket{y\oplus H(x)}$ is available.  
     We often refer to adversaries in the QROM as uniform quantum adversaries.
    \item {\bf Classical random oracle model with auxiliary-inputs (AI-CROM)~\cite{C:Unruh07}.}  This is identical to the CROM except that the adversary is allowed to take a polynomial-size classical advice that depends on the random oracle. 
   We often refer to adversaries in the AI-CROM as non-uniform classical adversaries. 
    \item {\bf Quantum random oracle model with (classical) auxiliary-inputs (AI-QROM)~\cite{AC:HhaXagYam19}.\footnote{We could also consider the QROM with quantum auxiliary-inputs, but we do not consider it in this paper.}}  This is identical to the QROM except that the adversary is allowed to take a polynomial-size classical advice that depends on the random oracle. 
   We often refer to adversaries in the AI-QROM as non-uniform quantum adversaries. 
\end{itemize}

\begin{remark}
In this paper, we treat random oracles as functions defined over a finite-size domain that depends on the security parameter. This treatment is more common in cryptography. On the other hand, in complexity theory, random oracles are often treated as functions over the infinite set $\{0,1\}^*$. By standard arguments, we can translate our results into those in the complexity theoretic setting (e.g., relative to a random oracle with probability $1$, proofs of quantumness exist etc.).   
\end{remark}

\begin{definition}[Family of oracle-aided functions.] For functions $\keylength=\keylength(\lambda),\inlength=\inlength(\lambda),\outlength=\outlength(\lambda)$,
a family $\{f_\secpar:\bit^{\keylength}\times\bit^{\inlength} \rightarrow \bit^{\outlength}\}_{\secpar\in\mathbb{N}}$ of \emph{efficiently computable oracle-aided keyed functions relative to oracles} $H:\bit^n \rightarrow \bit^m$
is a family of functions $f_{\secpar}$ that is implemented by a polynomial-time (deterministic) classical machine with an oracle access to $H$. The family of functions is \emph{keyless} if $\keylength=0$. If we do not specify keyed or keyless, we mean keyless. We denote by $f_{\secpar}^H$ to mean $f_\secpar$ relative to a specific oracle $H$.
\end{definition}

\paragraph{One-way functions.} We now define what it means for an oracle-aided function to be one-way relative to a random oracle. 
For one-way functions, we only consider keyless functions, as it is well known that keyless and keyed one-way functions are equivalent.

\begin{definition}[One-way functions with random oracles]
We say that a family $\{f_\secpar:\bit^{\inlength} \rightarrow \bit^{\outlength}\}_{\secpar\in\mathbb{N}}$ of efficiently computable oracle-aided functions relative to oracles $H:\bit^n \rightarrow \bit^m$  
is one-way in the CROM (resp. QROM) if for all unbounded-time $\A$ that make $\poly(\secpar)$ classical (resp. quantum) queries to $H$, there exists a negligible function $\negl$ such that:
\begin{align}
    \Pr_H[y= f_\secpar^{H}(x')
    :x\sample \bit^{\inlength}, y= f_\secpar^{H}(x), 
    x'\sample \A^{H}(1^\secpar,y)]<\negl(\secpar).\label{eq:owf}
\end{align} 

We say that $\{f_\secpar:\bit^{\inlength} \rightarrow \bit^{\outlength}\}_{\secpar\in\mathbb{N}}$ is one-way in the AI-CROM (resp. AI-QROM) if for all unbounded-time $\A$ that make $\poly(\secpar)$ classical (resp. quantum) queries to $H$ and polynomial-size classical advice $\{a(H)\}_H$, there exists a negligible function $\negl$ such that:
\begin{align}
    \Pr_H[y= f_\secpar^{H}(x')
    :x\sample \bit^{\inlength}, y= f_\secpar^{H}(x), 
    x'\sample \A^{H}(a(H),1^\secpar,y)]<\negl(\secpar).\label{eq:owf_AI_CROM}
\end{align} 
\end{definition}

\paragraph{Collision-resistance.} We now define collision-resistant hashing. 
\begin{definition}[Collision-resistance with random oracles]
We say that a family $\{f_\secpar:\bit^{\keylength}\times\bit^{\inlength} \rightarrow \bit^{\outlength}\}_{\secpar\in\mathbb{N}}$ of efficiently computable oracle-aided \emph{keyed} functions relative to oracles $H:\bit^n\rightarrow \bit^m$ 
is collision-resistant in the CROM (resp. QROM) if for all unbounded-time adversaries $\A$ that make $\poly(\secpar)$ classical (resp. quantum) queries to $H$, there exists a negligible function $\negl$ such that:
\begin{align*}
    \Pr_H[f_\secpar^{H}(k,x_0)=f_\secpar^{H}(k,x_1)~\land~x_0\neq x_1  :
    k\sample \bit^\keylength, 
    (x_0,x_1)\sample \A^{H}(k)]=\negl(\secpar).
\end{align*}
Collision-resistance in the AI-CROM and AI-QROM is defined analogously.
\end{definition}
A \emph{keyless} hash function has $\keylength=0$. Note that unlike one-way functions, \emph{keyless} collision resistant hash functions cannot have security  
against non-uniform adversaries since collisions may be hardcoded in the advice. 

\paragraph{Proofs of quantumness.} We now define proofs of quantumness, which have a quantum prover prove that they are quantum to a classical verifier. Like before, we will consider various definitions.

\begin{definition}\label{def:poqro}
A (keyed non-interactive publicly verifiable) proof of quantumness with key length $\keylength=\poly(\secpar)$ 
relative to a random oracle consists of algorithms $(\prove,\verify)$.
\begin{description}
\item[$\prove^{H}(1^\secpar,k)$:]
This is a QPT algorithm that takes the security parameter $1^\secpar$ and a key $k\in\{0,1\}^\keylength$ as input, makes $\poly(\secpar)$ quantum queries to the random oracle $H$, and outputs a classical proof $\pi$.  
\item[$\verify^{H}(1^\secpar,k,\pi)$:]
This is a deterministic classical polynomial-time algorithm that takes the security parameter $1^\secpar$, $k$ and a proof $\pi$, makes $\poly(\secpar)$ queries to the random oracle $H$, and outputs $\top$ indicating acceptance or $\bot$ indicating rejection.
\end{description}
We require a proof of quantumness to satisfy the following properties.\\

\noindent\textbf{Correctness.}
We have 
\[
\Pr_{H,k}\left[\verify^{H}(1^\secpar,k,\pi)=\bot:
\begin{array}{l}
\pi \sample \prove^{H}(1^\secpar,k)
\end{array}
\right]\leq \negl(\secpar).
\]

\noindent\textbf{Soundness.} A proof of quantumness is $(Q(\secpar),\epsilon(\secpar))$-sound in the CROM if, for any unbounded-time adversary $\A$ that makes  $Q(\secpar)$ \emph{classical} oracle queries to $H$, we have
\[
\Pr_{H,k}\left[\verify^{H}(1^\secpar,k,\pi^*)=\top:
\begin{array}{l}
\pi^* \sample \A^H(1^\secpar,k)
\end{array}
\right]\leq \epsilon(\secpar).
\]
When we do not specify $Q$ and $\epsilon$, we require that for any unbounded-time adversary $\A$ that makes $\poly(\secpar)$ queries, the above probability is $\negl(\secpar)$. 
Soundness in the AI-CROM is defined analogously.  
A \emph{keyless} proof of quantumness has $\keylength=0$. 
\end{definition}

Note that, as with collision resistance, there cannot be keyless proofs of quantumness with soundness against non-uniform adversaries. 
Indeed, a valid proof $\pi$ could be hardcored in the advice.

\paragraph{Proofs of randomness.} We now define proofs of (min-)entropy and proofs of randomness, also referred to as certifiable randomness. These are protocols by which a classical verifier with very little entropy can produce significant entropy with the help of a potentially untrusted quantum device.

We note that Brakerski et al.'s~\cite{FOCS:BCMVV18} work giving the first certifiable randomness protocol for a single device actually did not provide a formal definition. The work of Amos et al.~\cite{STOC:AGKZ20} provide a definition of certifiable min-entropy, but we observe that it is inappropriate.  
Their definition says that, conditioned on the verifier accepting, the string produced by the verifier must have min-entropy. We note, however, that a malicious device may always output a deterministic value. This value may be accepted with negligible but non-zero probability. Conditioned on accepting, the entropy remains zero. We give new definitions for certifiable entropy and randomness, overcoming this limitation.

We also note that defining certifiable randomness relative to a random oracle is subtle, since the random oracle itself is an infinite source of randomness. To accurately model entropy that comes from the protocol as opposed to the random oracle, we insist that the random string produced by the verifier has min-entropy or is uniformly random, even conditioned on the random oracle.

We note that for a proof of min-entropy, the situation is analogous to collision resistance where it is potentially feasible in the uniform setting or with a key, but trivially impossible in the non-uniform keyless setting. However, for a proof of randomness, it is inherent in the non-interactive setting that the verifier must have some local randomness. This is because, in the non-interactive setting without verifier randomness, a malicious prover can keep generating samples until, say, the first bit of the output is 0. Such a string clearly will not be uniformly random. This shows that the actual string obtained by the verifier must be kept secret from the prover, at least until after the prover's message is sent. 

We now give the definitions.

\begin{definition}\label{def:pominentropy}
A (keyed non-interactive publicly verifiable) proof of min-entropy relative to a random oracle with key length $\keylength=\poly(\secpar)$
consists of algorithms $(\prove,\verify)$.
\begin{description}
\item[$\prove^{H}(1^\secpar,k,1^h)$:]
This is a QPT algorithm that takes the security parameter $1^\secpar$, key $k\in\{0,1\}^\keylength$, and a min-entropy threshold $1^h$ as input. It makes $\poly(\secpar,h)$ quantum queries to the random oracle $H$, and outputs a classical proof $\pi$.  
\item[$\verify^{H}(1^\secpar,k,1^h,\pi)$:]
This is 
a deterministic classical polynomial-time algorithm that takes  
$1^\secpar,k,1^h$, and a proof $\pi$; 
it makes $\poly(\secpar,h)$ queries to the random oracle $H$, and outputs either a string $x$ (whose length may depend on $h$), or $\bot$ indicating rejection.
\end{description}

We require a proof of min-entropy to satisfy the following properties:\\

\noindent\textbf{Correctness.} 
 For any $h=h(\secpar)$,  
we have 
\[
\Pr_{H,k}\left[\verify^{H}(1^\secpar,k,1^h,\pi)=\bot:
\begin{array}{l}
\pi \sample \prove^{H}(1^\secpar,k,1^h)
\end{array}
\right]\leq \negl(\secpar).
\]

\noindent\textbf{Min-entropy.} 
For any polynomially-bounded $h=h(\lambda)$, any unbounded-time adversary $\A$ that makes $\poly(\secpar)$ quantum oracle queries to $H$, and for any inverse polynomial $\delta$, there is a negligible $\negl$ such that the following holds. Let $\A^H_{\top}(1^\secpar,k,1^h)$ be the distribution $\verify^H(1^\secpar,k,1^h,\A^H(1^\secpar,k,1^h))$, conditioned on the output not being $\bot$. Then:

\[\Pr_{H,k}\left[\Pr[\verify^H(1^\secpar,k,1^h,\A^H(1^\secpar,k,1^h))\neq\bot]\geq\delta(\lambda)\wedge H_\infty\left(\A^H_{\top}(1^\secpar,k,1^h)\right)\leq h(\lambda)\right]\leq\negl(\lambda)
\] 
The min-entropy requirement in the AI-QROM is defined analogously. 
A \emph{keyless} proof of min-entropy has $\keylength=0$ in which case we omit $k$ from the input of $\prove$ and $\verify$. 
\end{definition}
Note that min-entropy and correctness together imply that the output of $\verify$ when interacting with the honest $\prove$ algorithm must have min-entropy at least $h$ for an overwhelming fraction of $H,k$.

\medskip

\begin{definition}\label{def:porandomness}
A (keyed non-interactive publicly verifiable) proof of \emph{randomness} relative to a random oracle
has the same syntax as a proof of min-entropy (\Cref{def:pominentropy}), except that we 
allow $\verify$ to be randomized and 
require the output of $\verify$ to be exactly $h$ bits unless its output is $\bot$. 
We require a proof of randomness to satisfy the following properties:

\noindent\textbf{Correctness.} 
 For any $h=h(\secpar)$,  
we have 
\[
\Pr_{H,k,r}\left[\verify^{H}(1^\secpar,k,1^h,\pi;r)=\bot:
\begin{array}{l}
\pi \sample \prove^{H}(1^\secpar,k,1^h)
\end{array}
\right]\leq \negl(\secpar).
\]

\noindent\textbf{Succinct randomness.}
The length of the randomness $r$ used by $\verify$ is $\poly(\lambda,\log h)$ bits.  

\noindent\textbf{True randomness.} 
For any polynomially-bounded $h=h(\lambda)$ and any unbounded-time  adversary $\A$ that makes $\poly(\secpar)$ quantum oracle queries to $H$, and for any inverse polynomial $\delta$, there is a negligible $\negl$ such that the following holds for a $(1-\negl(\secpar))$-fraction of $(H,k)$. If it holds that $\Pr[\verify^H(k,h,\A^H(k,h);r)\neq\bot]\geq\delta$, then
\[
\Delta\left(\;(r,U)\;,\;(r,\A^H_{\top}(1^\secpar,k,1^h;r))\;)\right)\leq\negl(\lambda)
\]
where 
$\A^H_{\top}(1^\secpar,k,1^h;r)$ is the distribution $\verify^H(1^\secpar,k,1^h,\A^H(1^\secpar,k,1^h);r)$, conditioned on the output not being $\bot$,  
and $U$ is the uniform distribution over $h$-bit-strings. 
In other words, provided that $\verify$ actually outputs a string with inverse polynomial probability, that string will be statistically close to random for an overwhelming fraction of $H,k$. 

The true randomness requirement in the AI-QROM is defined analogously. 
A \emph{keyless} proof of randomness has $\keylength=0$ in which case we omit $k$ from the input of $\prove$ and $\verify$.
\end{definition}

\paragraph{From min-entropy to true randomness.} Here we discuss how proofs of min-entropy imply proofs of randomness. This is an immediate application of extractors:
\begin{theorem}\label{thm:min-entropy_to_randomness}
If proofs of min-entropy in the QROM (resp. AI-QROM) exist, then so do proofs of randomness in the QROM (resp. AI-QROM). 
If the proof of min-entropy is keyless, then so is the proof of randomness. 
\end{theorem}
\begin{proof}We simply have a new $\verify'$ which chooses a random seed for a strong extractor, which it applies to the result of $\verify$, outputting whatever the extractor outputs. By choosing the min-entropy $h$ sufficiently higher than the desired output length according to the parameters of the extractor, the output of $\verify'$ will be statistically close to random and the desired length.
\end{proof}

We note that the verifier's random seed for the extractor can be sampled after the prover's message, and can also be made public afterward. The result is that if the proof of min-entropy is public coin and publicly verifiable, the proof of randomness will be as well, at the cost of a single final message from the verifier.

\subsection{From Uniform to Non-Uniform Security} 
Clearly, security against non-uniform adversaries implies security against uniform adversaries.  
For the other direction, we can use known results of ~\cite{EC:CDGS18}  and~\cite{FOCS:CGLQ20} 
that show that salting generically lifts uniform security to non-uniform security in the classical and quantum random oracle models, respectively. Note that the results require it to be efficiently verifiable when the adversary wins; this applies to one-way functions, collision resistance, and proofs of quantumness, but not to proofs of min-entropy/randomness, where it cannot be efficiently checked if the adversary produced a low entropy or non-uniform string. As immediate corollaries of these results, we obtain the following: 
\begin{theorem}\label{thm:uniform_to_non-uniform_OW}
If $\{f_\lambda\}_\lambda$ is one-way in the CROM (resp. QROM), then $\{g_\lambda\}_\lambda$ where $g_\lambda^H(s,x)=s||f_\lambda^{H(s||\cdot)}(x)$ and where $s\in\bit^\secpar$ is one-way in the AI-CROM (resp. AI-QROM).
\end{theorem}
\begin{theorem}\label{thm:uniform_to_non-uniform_CR}
If $\{f_\lambda\}_\lambda$ is a potentially keyed function family that is collision resistant in the CROM (resp. QROM), then the \emph{keyed} function $\{g_\lambda\}_\lambda$ where $g_\lambda(k_0||k_1,x)=f_\lambda^{H(k_1||\cdot)}(k_0,x)$ and where $k_1\in\bit^\secpar$ is collision resistant against in the AI-CROM (resp. AI-QROM). 
\end{theorem}
\begin{theorem}\label{thm:poq_uniform_to_non-uniform} 
If $(\prove_0,\verify_0)$ is a proof of quantumness that satisfies soundness in the CROM, then $(\prove,\verify)$ satisfies soundness in the AI-CROM, where $\prove^H(1^\secpar,k_0||k_1)=\prove_0^{H(k_1||\cdot)}(1^\secpar,k_0)$ and $\verify^H(1^\secpar,k_0||k_1,\pi)=\verify_0^{H(k_1||\cdot)}(1^\secpar,k_0,\pi)$ and where $k_1\in\bit^\secpar$.
\end{theorem}

We next discuss how salting actually does lift security for proofs of min-entropy and randomness from the uniform to non-uniform case in the classical advice setting. We note that~\cite{FOCS:CGLQ20} actually \emph{does} work, by fixing a particular string, and having the adversary win if it can cause the verifier to output that string. This event occurs with exponentially-small probability, but~\cite{FOCS:CGLQ20} would handle exponentially small probabilities by setting the salt to be appropriately larger than the min-entropy requirement. This limits the utility of a proof of min-entropy, since the large salt could have just been used as the source of randomness. In the following, we show that small salts can, in fact, be used, though it requires a more careful proof and cannot simply rely on the prior theorem statements.

\begin{theorem}\label{thm:randomness_uniform_to_non-uniform} 
If $(\prove_0,\verify_0)$ is a proof of min-entropy (resp. proof of randomness) in the QROM, then $(\prove,\verify)$ is a proof of min-entropy (resp. proof of randomness) in the AI-QROM, where $\prove^H(1^\secpar,k_0||k_1,1^h)=\prove_0^{H(k_1||\cdot)}(1^\secpar,k_0,1^{h+1})$ and $\verify^H(1^\secpar,k_0||k_1,1^h,\pi)=\verify_0^{H(k_1||\cdot)}(1^\secpar,k_0,1^{h+1},\pi)$ and where $k_1\in\bit^\secpar$.
\end{theorem}

We defer the proof to \Cref{sec:proof_randomness_uniform_to_non-uniform}.

\section{Error Correcting Codes.}\label{sec:codes}
In this section, we first review basic definitions and facts on error correcting codes. 
Then, we state requirements of codes that are needed for our purpose. 
Then, we show that such a code exists based on known results. 


\subsection{Definitions}\label{sec:codes_definitions}
A code of length $n\in \mathbb{N}$ over an alphabet $\Sigma$ (which is a finite set) is a subset $C\subseteq \Sigma^n$.

\paragraph{Linear codes.}
A code $C$ is said to be a linear code if its alphabet is $\Sigma=\FF_q$ for some prime power $q$ and $C\subseteq \FF_q^n$ is a linear subspace of $\FF_q^n$. 

\paragraph{Folded linear codes.} 
A code $C$ is said to be a folded linear code \cite{Krachkovsky03,GR08} if its alphabet is $\Sigma=\FF_q^m$ for some prime power $q$ and a positive integer $m$ and $C\subseteq \Sigma^{n}$ is a linear subspace of $\FF_q^{nm}$ where $n$ is the length of $C$ and we embed $C$ into $\FF_q^{nm}$ in the canonical way.
Linear codes are the special case of folded linear codes where $m=1$. 
For a linear code $C\subseteq \FF_q^n$ and a positive integer $m$ that divides $n$, we define its $m$-folded version $C^{\mfol}$ as follows:
\begin{align*}
    C^{\mfol}\defeq \{((x_1,\ldots,x_m),(x_{m+1},\ldots,x_{2m})\ldots,(x_{n-m+1},\ldots,x_{n})):(x_1,\ldots,x_n)\in C\}.
\end{align*}
Clearly, $C^{\mfol}$ is a folded linear code. 
Conversely, any folded linear code can be written as $C^{\mfol}$ for some linear code $C$ and a positive integer $m$. 


\paragraph{Dual codes.}
Let $C$ be a linear code of length $n$ and dimension $k$ over $\FF_q$. The \emph{dual code} $C^\perp$ of $C$ is defined as the orthogonal complement of $C$ as a linear space over $\FF_q$, i.e., 
$$C^{\perp}\defeq \{\vz\in \FF_q^n: \vx \cdot \vz = 0 \text{~for~all~}\vx \in C\}.$$
$C^{\perp}$ is a linear code of length $n$ and dimension $n-k$ over $\FF_q$.\footnote{Note that it does not always hold that $\FF_q^n=C \oplus C^\perp$ since the bilinear form $(\vx,\vy)\mapsto \vx\cdot\vy$ does not satisfy the axioms of the inner product (i.e., there may exist $\vx\neq 0$ such that $\vx\cdot \vx =0$).}

We define dual codes for folded linear codes similarly. That is, for a folded linear code $C\in \Sigma^n$ over the alphabet $\Sigma=\FF_q^m$, its dual $C^\perp$ is defined as 
$$C^{\perp}\defeq \{\vz\in \Sigma^n: \vx \cdot \vz = 0 \text{~for~all~}\vx \in C\}.\footnote{Recall that $\vx \cdot \vz$ for $\vx,\vz\in \Sigma^n$ is defined in \Cref{sec:finite_field}.}$$
It is clear from the definition that $(C^{\perp})^{\mfol}=(C^{\mfol})^{\perp}$ for any linear codes $C$ of length $n$ and positive integer $m$ that divides $n$.

\begin{lemma}\label{lem:fourier_dual}
For a folded linear code $C\subseteq \Sigma^n$, if we define 
\[
f(\vx)\defeq 
\begin{cases}
\frac{1}{\sqrt{|C|}}& \vx\in C\\
0& \text{otherwise}
\end{cases},
\]
then we have 
\[
\hat{f}(\vz)= 
\begin{cases}
\frac{1}{\sqrt{|C^{\perp}|}}& \vz\in C^{\perp}\\
0& \text{otherwise}
\end{cases}.
\]
\end{lemma}
\begin{proof}
For $\vz\in C^{\perp}$, we have 
\begin{align*}
    \hat{f}(\vz)
    &=\frac{1}{|\Sigma|^{n/2}}\sum_{\vx\in \Sigma^n}f(\vx)\omega_p^{\Tr(\vx \cdot \vz)}\\
    &=\frac{1}{|\Sigma|^{n/2}}\sum_{\vx\in C}\frac{1}{\sqrt{|C|}}\\
    &=\frac{1}{\sqrt{|C^{\perp}|}}
\end{align*}
where the final equality follows from $|C|\cdot |C^\perp|=|\Sigma|^n$. 
That $\hat{f}(\vz)=0$ for $\vz \notin C^{\perp}$ immediately follows from the above and \Cref{lem:Parseval}. 
\end{proof}

\paragraph{List recovery.}
We say that a code $C\subseteq \Sigma^n$ is $(\zeta,\ell,L)$-list recoverable if for any subsets $S_i\subseteq \Sigma$ such that $|S_i|\leq \ell$ for $i\in [n]$, we have 
\begin{align*}
    |\{(x_1,...,x_n) \in C:|\{i\in [n]:x_i\in S_i\}|\geq (1-\zeta)n\}|\leq L.
\end{align*}
Note that list recoverability in the literature usually requires that the list of all codewords $(x_1,...,x_n) \in C$ satisfying $|\{i\in [n]:x_i\in S_i\}|\geq (1-\zeta)n$ can be computed from $\{S_i\}_{i\in [n]}$ in time polynomial in $|\Sigma|,n,\ell$. 
However, we will not require this.

\subsection{Suitable Codes}
The following lemma claims the existence of codes that are suitable for our purpose. 
\begin{lemma}[Suitable codes]\label{lem:good_codes}
For any constants
$0<c <c' <1$, there is an explicit family $\{C_\secpar\}_{\secpar \in \mathbb{N}}$ of folded linear codes over the alphabet $\Sigma=\FF_q^m$ of length $n$ where
$|\Sigma|=2^{\secpar^{\Theta(1)}}$,   $n=\Theta(\secpar)$, 
and $|C_\secpar|\geq 2^{n+\secpar}$ 
that satisfies the following.\footnote{\Cref{item:hw} is not needed for the construction of a proof of quantumness given in \Cref{sec:PoQ}. It is used only in the counterexample for one-way functions given in \Cref{sec:separation_OWF}.}
\begin{enumerate}
    \item \label{item:list_recovery}
    $C_\secpar$ is $(\zeta,\ell,L)$-list recoverable where $\zeta=\Omega(1)$, $\ell=2^{\secpar^c}$ and $L=2^{\tilde{O}(\secpar^{c'})}$.
    \item \label{item:dual_decode}
 There is an efficient deterministic decoding algorithm $\decode_{C_\secpar^\perp}$ for $C_\secpar^\perp$ that satisfies the following.
  Let $\dist$ be a distribution over $\Sigma$ that takes $\vzero$ with probability $1/2$ and otherwise takes a uniformly random element of $\Sigma\setminus \{\vzero\}$. 
  Then, it holds that
        \begin{align*}
       \Pr_{\ve\sample \dist^n}[\forall \vx\in C_\secpar^\perp,~\decode_{C_\secpar^\perp} (\vx+\ve)=\vx]=1-2^{-\Omega(\secpar)}.
    \end{align*}
 \item \label{item:hw}
 For all $j\in [n-1]$, $\Pr_{\vx\sample C_\secpar}[\hw(\vx)=n-j]\leq \left(\frac{n}{|\Sigma|}\right)^{j}$.
\end{enumerate}
\end{lemma} 
Our instantiation of $C_\secpar$ is just folded Reed-Solomon codes with an appropriate parameter setting. 
\Cref{item:list_recovery} is a direct consequence of the list recoverability of folded Reed-Solomon codes in a certain parameter regime \cite{GR08,RudraThesis}. 
For proving \Cref{item:dual_decode}, we first remark that the duals of folded Reed-Solomon codes are folded \emph{generalized} Reed-Solomon codes, which have efficient list decoding algorithms~\cite{GS99}. 
Then, we prove that the list decoding algorithm returns a unique decoding result when the error comes from the distribution $\dist^n$. \Cref{item:hw} follows from a simple combinatorial argument.  
The proof of  \Cref{lem:good_codes} is given in \Cref{sec:construction_code}.  
\begin{remark}
Folded Reed-Solomon codes are the only instantiation of $C_\lambda$ which we are aware of. Especially, we are not aware of any other codes that satisfy list-recoverability with appropriate parameters for our purpose. 
\end{remark}



\subsection{Proof of \Cref{lem:good_codes}}\label{sec:construction_code}
In this subsection, we prove \Cref{lem:good_codes}, i.e., we give a construction of codes that satisfy the properties stated in \Cref{lem:good_codes}.  

\subsubsection{Preparation}
Before giving the construction, we need some preparations. 
\paragraph{Generalized Reed-Solomon codes.}
 We review the definition and known facts on (generalized) Reed-Solomon codes. 
 See e.g., \cite[Section 6]{Lindell_Coding} for more details. 

A generalized Reed-Solomon code $\GRS_{\FF_q,\gamma,k,\vecv}$ over $\FF_q$ w.r.t. a generator $\gamma$ of $\FF_q^*$, the degree parameter $0\leq k\leq  N$, and $\vecv=(v_1,...,v_N)\in {\FF_q^*}^N$ where $N\defeq q-1$ is defined as follows:
\[
\GRS_{\FF_q,\gamma,k,\vecv}\defeq \{(v_1f(\gamma),v_2f(\gamma^2)...v_Nf(\gamma^{N})):f\in \FF_q[x]_{deg\leq k}\}
\]
where $\FF_q[x]_{deg\leq k}$ denotes the set of polynomials over $\FF_q$ of degree at most $k$.\footnote{Reed-Solomon codes whose length $N$ is smaller than $q-1$ are often considered. But we focus on the case of $N=q-1$.} 
We remark that $\GRS_{\FF_q,\gamma,k,\vecv}$ is a linear code over $\FF_q$ that has length $N=q-1$ and dimension $k+1$.  
A Reed-Solomon code is a special case of a generalized Reed-Solomon code where  $\vecv=(1,1,\ldots,1)$. 
We denote it by $\RS_{\FF_q,\gamma,k}$ (which is equivalent to $\GRS_{\FF_q,\gamma,k,(1,1,\ldots,1)})$. 
The dual of $\RS_{\FF_q,\gamma,k}$ is $\GRS_{\FF_q,\gamma,N-k-2,\vecv}$ for some $\vecv\in \FF_q^N$~\cite[Claim 6.3]{Lindell_Coding}.\footnote{Recall that the dimension of (generalized) Reed-Solomon codes is the degree parameter $k$ plus one.}  

There is a classical deterministic list decoding algorithm $\GRSListDecode_{\FF_q,\gamma,k,\vecv}$ for $\GRS_{\FF_q,\gamma,k,\vecv}$ that corrects up to $N-\sqrt{kN}$ errors in polynomial time in $N$~\cite{GS99}.\footnote{\cite{GS99} described the list decoding algorithm for Reed-Solmon codes, but that can be extended to one for generalized Reed-Solomon codes in a straightforward manner since scalar multiplications in each coordinate do not affect the decodability.} 
More precisely, for any $\vz\in \FF_q^N$, $\GRSListDecode_{\FF_q,\gamma,k,\vecv}(\vz)$ returns the list of all  $\vx\in \GRS_{\FF_q,\gamma,k,\vecv}$ such that 
$\hw(\vx-\vz)< N-\sqrt{kN}$. 



\paragraph{Folded Reed-Solomon codes.}
Let $m$ be a positive integer that divides $N=q-1$. 
The $m$-folded version $\RS_{\FF_q,\gamma,k}^{\mfol}$ of $\RS_{\FF_q,\gamma,k}$ is a folded linear code of length $n=N/m$.\footnote{We remark that the roles of $n$ and $N$ are swapped compared with \cite{GR08,RudraThesis}.\label{footnote:n_and_N}} 
It is known that $\RS_{\FF_q,\gamma,k}^{\mfol}$ is list recoverable in the following parameter regime~\cite{GR08,RudraThesis}.\footnote{The following lemma is based on Rudra's PhD thesis~\cite{RudraThesis}. 
The same result is also presented in the journal version~\cite{GR08}, but note that there is a notational difference in the definition of list recovery: the definition of $(\zeta,\ell,L)$-list recovery of \cite{GR08} means $((1-\zeta),\ell,L)$-list recovery of \cite{RudraThesis} and this paper. Also remark \Cref{footnote:n_and_N}.} 
\begin{lemma}[{\cite[Sec. 3.6]{RudraThesis}}]\label{lem:FRS_list_recovery}
Let $q$ be a prime power, $\gamma\in \FF_q^*$ be a generator, $N\defeq q-1$, $k<N$ be a positive integer, and $m$ be a positive integer that divides $N$.
For positive integers $\ell$, $r$, and $s\leq m$ and a real $0<\zeta<1$, suppose that the following inequalities hold:
\begin{align}
    \frac{(1-\zeta)N}{m}\geq \left(1+\frac{s}{r}\right)\frac{\sqrt[s+1]{N\ell k^s}}{m-s+1} \label{eq:condition_one}\\
    (r+s)\sqrt[s+1]{\frac{N\ell}{k}}<q. \label{eq:condition_two}
\end{align}
Then, $\RS_{\FF_q,\gamma,k}^{\mfol}$ is $(\zeta,\ell,L)$-list recoverable where $L=q^s$.
\end{lemma}

\subsubsection{Construction}\label{subsec:code_construction}
We show that folded Reed-Solomon codes satisfy the requirements of \Cref{lem:good_codes} if we set parameters appropriately.
In the following, whenever we substitute non-integer values into integer variables, there is an implicit flooring to integers which we omit writing.
Fix $0<c<c'<1$, which defines $\ell=2^{\secpar^c}$.  
Our choices of parameters are as follows:
\begin{itemize}
    \item $q=2^{2\lfloor \log \secpar\rfloor}$ (which automatically defines $N=q-1$), $m=2^{\lfloor\log \secpar\rfloor}+1$, and $n=N/m=2^{\lfloor\log \secpar\rfloor}-1$.\footnote{This is an example of the parameter choice. Any prime power of the form $q=nm+1$ where $n$ and $m$ are positive integers such that $n=\Omega(\secpar)$ and $m=\Omega(\secpar)$ suffices.}  
    
    \item $\gamma$ is an arbitrary generator of $\FF_q^*$. Note that we can find $\gamma$ in polynomial time in $\lambda$ since $q=\poly(\lambda)$.
    \item 
    $k=\alpha N$ for an arbitrary constant $5/6<\alpha<1$. 
\end{itemize}
We set $C_\secpar \defeq \RS_{\FF_q,\gamma,k}^{\mfol}$. 
By the above parameter setting, it is easy to see that we have $|\Sigma|=2^{\secpar^{\Theta(1)}}$,  $n=\Theta(\secpar)$, 
and $|C_\secpar|=q^{k+1}\geq 2^{n+\secpar}$. 
We show that 
$\{C_\secpar\}_{\secpar\in \mathbb{N}}$ satisfies 
the requirements of \Cref{lem:good_codes}. 
For notational simplicity, we omit $\secpar$ from the subscript of $C$. 


\paragraph{First item.}
We prove \Cref{item:list_recovery} of \Cref{lem:good_codes}. First, we remark that we only have to prove that the requirement is satisfied for sufficiently large $\secpar$ since  we can set $L=q^{N}$ for finitely many $\secpar$ for which $(\zeta,\ell,L)$-list recoverability is trivially satisfied for any $\zeta$ and $\ell$.  
We apply \Cref{lem:FRS_list_recovery} with the following parameters:
\begin{itemize}
\item $s=\lambda^{c'}$.  Note that this satisfies the requirement $s\leq m$ in \Cref{lem:FRS_list_recovery} for sufficiently large $\secpar$ since $m=\Omega(\secpar)$ and  $c'<1$.
\item $r=\lambda^{c''}$ for a constant $c'<c''<1$.
\item $0<\zeta<1-\alpha$ is an arbitrary constant.
\end{itemize}
Based on the above parameter setting, we have $\lim_{\secpar \rightarrow \infty}(1+\frac{s}{r})=1$, $\lim_{\secpar \rightarrow \infty}\frac{m}{m-s+1}=1$, and $\lim_{\secpar \rightarrow \infty}\sqrt[s+1]{\ell}=1$ where we used $\ell=2^{\secpar^{c}}$ and $c<c'$.  
Therefore, \Cref{eq:condition_one} can be rearranged as follows:
\begin{align}
    1-\zeta \geq (1+o(1))\left(\frac{k}{N}\right)^{\frac{s}{s+1}}
\end{align}
This is satisfied for sufficiently large $s$ (which occurs for sufficiently large $\secpar$) since we assume 
$k=\alpha N$ and 
$\zeta<1-\alpha$. 

Similarly, by our choice of parameters, the LHS of 
\Cref{eq:condition_two} is $O(\secpar^{c''})$ and the RHS is $\Omega(\secpar^2)$. 
Since $c''<1$, \Cref{eq:condition_two} also holds for sufficiently large $\secpar$.

Thus, by \Cref{lem:FRS_list_recovery}, $\RS_{\FF_q,\gamma,k}^{\mfol}$ with the above parameter setting is $(\zeta,\ell,L)$-list recoverable where $L=q^s\leq (\secpar^2)^{\secpar^{c'}}=2^{\widetilde{O}(\secpar^{c'})}$.
This means that  \Cref{item:list_recovery} of \Cref{lem:good_codes} is satisfied.

\paragraph{Second item.}
Next, we prove \Cref{item:dual_decode} of \Cref{lem:good_codes}.
Since $C=\RS_{\FF_q,\gamma,k}^{\mfol}$ is a folded Reed-Solomon code, its dual $C^\perp$ is a folded generalized Reed-Solomon code $\GRS_{\FF_q,\gamma,N-k-2,\vecv}^{\mfol}$ for some $\vecv\in \FF_q^N$.
In the following, we think of an element of $\Sigma^n$ as an element of $\FF_q^{N}$ in the canonical way.
Then, $C^\perp=\GRS_{\FF_q,\gamma,N-k-2,\vecv}^{\mfol}$ is identified with $\GRS_{\FF_q,\gamma,N-k-2,\vecv}$.  
Let $d\defeq N-k-2$ and $0<\epsilon<0.09$ be a constant specified later.
We define  
$\decode_{C^\perp}$ as follows.
\begin{description}
\item[$\decode_{C^\perp}(\vz)$:] 
On input $\vz \in \FF_q^N$, 
it runs the list decoding algorithm $\GRSListDecode_{\FF_q,\gamma,N-k-2,\vecv}(\vz)$ to get a list of codewords. 
If there is a unique $\vx$ in the list such that $\hw(\vz-\vx)\leq (1/2+\epsilon)N$, it outputs $\vx$, 
and otherwise outputs $\bot$.   
\end{description}

We define a subset $\gooderrors\subseteq \FF_q^N$ as follows.
\begin{align*}
 \gooderrors\defeq \{\ve \in \FF_q^N: \hw(\ve)\leq (1/2+\epsilon)N ~\land~ \forall\vy\in C^\perp\setminus\{\vzero\},~\hw(\ve-\vy)>(1/2+\epsilon)N  \}.  
\end{align*}

 For any $\vx\in C^\perp$ and $\ve\in \gooderrors$, by the definition of $\gooderrors$, $\vx$ is the only codeword of $C^\perp$ whose Hamming distance from $\vx+\ve$ is smaller than or equal to $(1/2+\epsilon)N$.
  Moreover, since $k=\alpha N$ for $\alpha>5/6$ and $\epsilon<0.09$, it holds that $N-\sqrt{dN}=N-\sqrt{(1-\alpha)N^2-2N}\geq (1-\sqrt{1-\alpha})N>0.59 N> (1/2+\epsilon)N$. 
 Thus, for any $\vx\in C^\perp$ and $\ve\in \gooderrors$, the list output by $\GRSListDecode_{\FF_q,\gamma,N-k-2,\vecv}(\vx+\ve)$ must contain $\vx$, which implies 
 \[
 \decode_{C^\perp}(\vx+\ve)=\vx.
 \]
Thus, it suffices to prove
\begin{align*}
        \Pr_{\ve\sample \dist^n}[\ve \notin \gooderrors]=2^{-\Omega(\secpar)}
\end{align*}
where $\dist$ is the distribution as defined in \Cref{lem:good_codes}.\footnote{$\dist^n$ is defined as a distribution over $\Sigma^n$, but its sample can be interpreted as an element of $\FF_q^N$ in the canonical way.} 
For $\ve\in \FF_q^N$, we parse it as  $\ve=(\ve_1,...,\ve_n)\in \Sigma^n$ and define $S_\ve\subseteq [N]$ as the set of indices on which $\ve_i=\vzero$, i.e., 
\[
S_\ve\defeq \bigcup_{i\in[n]:\ve_i=\vzero}\{(i-1)m+1,(i-1)m+2,\ldots,im\}.
\]
By the definition of $\dist$ and $n=\Theta(\secpar)$, the Chernoff bound (\Cref{lem:Chernoff}) gives
\begin{align*}
\Pr_{\ve\sample \dist^n}\left[(1/2-\epsilon)N\leq |S_\ve|\leq(1/2+\epsilon)N\right]\geq 1-2^{\Omega(\secpar)}.
\end{align*} 
Therefore, it suffices to prove 
\begin{align} \label{eq:gooderrors_prob}
        \Pr_{\ve\sample \dist^n}[\ve \notin \gooderrors\mid S_\ve=S^*]=2^{-\Omega(\secpar)}
\end{align}
for all $S^*\subseteq [N]$ such that  $(1/2-\epsilon)N\leq |S^*|\leq (1/2+\epsilon)N$.  
Fix such $S^*$. 
When $S_\ve=S^*$, it is clear that we have $\hw(\ve)\leq (1/2+\epsilon)N$ since $|S^*|\geq (1/2-\epsilon)N$.  
Thus, when $S_\ve=S^*$ and $\ve\notin \gooderrors$, there exists $\vy \in C^\perp\setminus \{\vzero\}$ such that 
\begin{align}\label{eq:upperbound_hw_e-y}
    \hw(\ve-\vy)\leq (1/2+\epsilon)N.
\end{align}
Let $\bar{S}^*\defeq [N]\setminus S^*$. 
Note that $|\bar{S}^*|> d+2\epsilon N$ holds by our parameter choices. 
It holds that\footnote{Recall the notation $\vx_S=(x_i)_{i\in S}$ for $\vx=(x_1,\ldots,x_N) \in \FF_q^N$ and $S\subseteq [N]$.}
\begin{align} \label{eq:hw_e-y_decompose}
\hw(\ve-\vy)= \hw(\ve_{S^*}-\vy_{S^*}) +  \hw(\ve_{\bar{S}^*}-\vy_{\bar{S}^*}).
\end{align}
Since we assume $S^*=S_{\ve}$, we have $\ve_{S^*}=\vzero$. On the other hand, since $\vy\neq \vzero$ and degree $d$ non-zero polynomials have at most $d$ roots, $\vy$ can take $0$ on at most $d$ indices. In particular, we have  
\begin{align} \label{eq:hw_e-y_on_S_star}
\hw(\ve_{S^*}-\vy_{S^*})\geq |S^*|-d.
\end{align}
By combining \Cref{eq:upperbound_hw_e-y,eq:hw_e-y_decompose,eq:hw_e-y_on_S_star}, we have 
\begin{align}\label{eq:lower_bound_hw_e-y_on_S_star}
 \hw(\ve_{\bar{S}^*}-\vy_{\bar{S}^*})\leq  (1/2+\epsilon)N- (|S^*|- d)\leq d +2\epsilon N 
\end{align}
where we used $|S^*|\geq (1/2-\epsilon)N$. 
That is, conditioned on $S_\ve=S^*$, \Cref{eq:lower_bound_hw_e-y_on_S_star} holds for some $\vy \in C^\perp\setminus \{\vzero\}$ whenever $\ve \notin \gooderrors$. 
Moreover, conditioned on $S_\ve=S^*$, the distribution of $\ve_{\bar{S}^*}$ is a direct product of $|\bar{S}^*|/m$ copies of the uniform distribution over $\FF_q^m\setminus \{\vzero\}$ by the definition of $\dist$. Since $q^m=2^{\Omega(\secpar)}$, the distribution is statistically $2^{-\Omega(\secpar)}$-close to the uniform distribution over $\FF_q^N$.  
Combining these observations, it holds that\footnote{We can take $\exists \vy \in C^\perp$ instead of $\exists \vy \in C^\perp\setminus \{\vzero\}$ in the RHS since this does not decrease the probability. Indeed, one can see that the probabilities are the same noting that $\ve_{\bar{S}^*}$ does not take $0$ on any index and $|\bar{S}^*|> d+2\epsilon N$.}
\begin{align} \label{eq:upper_bound_gooderrors}
     \Pr_{\ve\sample \dist^n}[\ve \notin \gooderrors\mid S_\ve=S^*]\leq \Pr_{\ve_{\bar{S}^*}\sample \FF_q^{\left|\bar{S}^*\right|}}[\exists \vy \in C^\perp~\hw(\ve_{\bar{S}^*}-\vy_{\bar{S^*}})\leq d+2\epsilon N]+2^{-\Omega(\secpar)}. 
\end{align}
When there exists $\vy \in C^\perp$ such that $\hw(\ve_{\bar{S}^*}-\vy_{\bar{S^*}})\leq d+2\epsilon N$, 
there is a subset $T\subseteq \bar{S}^*$ such that $|T|= |\bar{S}^*|-\lceil d+2\epsilon N \rceil$ and $\ve_{T}=\vy_{T}$ since we have $|\bar{S}^*|> \lceil d+2\epsilon N \rceil$.
On the other hand, since a codeword of $C^\perp$ is determined by values on $d+1$ indices, 
for any fixed $T\subseteq S^*$, we have 
\begin{align} \label{eq:upperbound_equal_on_T}
   \Pr_{\ve_{\bar{S}^*}\sample \FF_q^{\left|\bar{S}^*\right|}}[\exists \vy \in C^\perp~\ve_{T}=\vy_{T}]= q^{-(|T|-(d+1))}\leq  q^{-\left(\frac{1}{2}-3\epsilon\right)N+2d+1}
\end{align}
where we used $|T|\geq  |\bar{S}^*|-d-2\epsilon N$ and $|\bar{S}^*|\geq (1/2-\epsilon)N$. 
Since there are ${|\bar{S}^*| \choose \lceil d+2\epsilon N \rceil}$ possible choices of $T$, combined with \Cref{eq:upperbound_equal_on_T}, it holds that 
\begin{align}
    \Pr_{\ve_{\bar{S}^*}\sample \FF_q^{\left|\bar{S}^*\right|}}[\exists \vy \in C^\perp~\hw(\ve_{\bar{S}^*}-\vy_{\bar{S^*}})\leq d+2\epsilon N] 
    &\leq {|\bar{S}^*| \choose \lceil d+2\epsilon N \rceil}\cdot q^{-\left(\frac{1}{2}-3\epsilon\right)N+2d+1} \notag \\
    &\leq q^{d+2\epsilon N +1}\cdot q^{-\left(\frac{1}{2}-3\epsilon\right)N+2d+1} \notag \\
    &\leq q^{-\left(\frac{1}{2}-3(1-\alpha)-5\epsilon\right)N-4} \label{eq:upperbound_prob_S_star}
\end{align}
where we used $|\bar{S}^*|\leq N<q$ in the second inequality 
and $d=N-k-2=(1-\alpha)N-2$ in the third inequality.  
Since $5/6<\alpha<1$, we can choose $0<\epsilon<0.09$ in such a way that  $\frac{1}{2}-3(1-\alpha)-5\epsilon>0$. (For example, $\epsilon\defeq -\frac{1}{4}+\frac{3}{10}\alpha$ suffices.)
Then, by combining \Cref{eq:upper_bound_gooderrors,eq:upperbound_prob_S_star} together with $q=\Omega(\secpar)$ and $\frac{1}{2}-3(1-\alpha)-5\epsilon=\Omega(1)$, we obtain \Cref{eq:gooderrors_prob}.  

\paragraph{Third item.}
Finally, we prove \Cref{item:hw} of \Cref{lem:good_codes}. 
For $\lceil \frac{k+1}{m}\rceil< j < n$, there does not exist a codeword $\vx$ such that $\hw(\vx)=n-j$. 
This is because if $\hw(\vx)=n-j$, 
the polynomial $f$ corresponding to $\vx$ 
has at least $mj\geq k+1$ roots, which means that $\vx=\vzero$ since the degree of $f$ is at most $k$. This contradicts $\hw(\vx)=n-j>0$. 

The case of $j\leq \lceil \frac{k+1}{m}\rceil$ is proven below. 
In this case, since a polynomial of degree at most $k$ is determined by evaluated values on $k+1$ points, for any subset $S\subseteq [n]$ such that $|S|=j$, $\vx_S$ is uniformly distributed over $\Sigma^j$ when $\vx\sample C_\secpar$. 
Therefore, we have 
\begin{align*}
\Pr_{\vx\sample C_\secpar}[\hw(\vx)=n-j]
&\leq \sum_{S\subseteq [n]\text{~s.t.~}|S|=j}\Pr_{\vx\sample C_\secpar}[\vx_{S}=\vzero]\\
&\leq {n\choose j}|\Sigma|^{-j}\\
&\leq \left(\frac{n}{|\Sigma|}\right)^j. 
\end{align*}

This completes the proof of 
\Cref{lem:good_codes}.
\section{Technical Lemma}\label{sec:technical_lemma}
We prepare a lemma that is used in the proof of correctness of our proof of quantumness constructed in \Cref{sec:PoQ}.  
The lemma is inspired by the quantum step of Regev's reduction from LWE to worst-case lattice problems \cite{STOC:Regev05}. 

\begin{lemma}\label{lem:Regev_like}
Let $\ket{\psi}$ and $\ket{\phi}$ be quantum states on a quantum system over an alphabet $\Sigma=\FF_q^m$ written as 
\begin{align*}
    &\ket{\psi}=\sum_{\vx \in \Sigma^n}V(\vx)\ket{\vx}\\
    &\ket{\phi}=\sum_{\ve \in \Sigma^n}W(\ve)\ket{\ve}.
\end{align*}
Let $F:\Sigma^n \rightarrow \Sigma^n$ be a function. 
Let $\good\subseteq \Sigma^n \times \Sigma^n$ be a subset such that for any $(\vx,\ve)\in \good$, we have $F(\vx+\ve)=\vx$. 
Let $\bad$ be the complement of $\good$, i.e., $\bad\defeq(\Sigma^n \times \Sigma^n)\setminus \good$.
Suppose that we have 
\begin{align}
&\sum_{(\vx,\ve)\in \bad}|\hat{V}(\vx)\hat{W}(\ve)|^2\leq \epsilon \label{eq:hatV_hatW}\\
&\sum_{\vz\in \Sigma^n}\left|\sum_{(\vx,\ve)\in \bad: \vx+\ve=\vz}\hat{V}(\vx)\hat{W}(\ve)\right|^2\leq \delta.
\label{eq:hatV_hatW_two}
\end{align}
Let $U_{\mathsf{add}}$ and $U_{F}$ be unitaries defined as follows:
\begin{align*}
&\ket{\vx}\ket{\ve}
\xrightarrow{U_{\mathsf{add}}}
\ket{\vx}\ket{\vx+\ve}\xrightarrow{U_{F}}
\ket{\vx-F(\vx+\ve)}\ket{\vx+\ve}.
\end{align*}
Then we have 
\begin{align*}
   (I\otimes (\QFT_{\Sigma}^{-1})^{\otimes n})U_{F}U_{\mathsf{add}}(\QFT_{\Sigma}^{\otimes n}\otimes \QFT_{\Sigma}^{\otimes n})\ket{\psi}\ket{\phi} \approx_{\sqrt{\epsilon}+\sqrt{\delta}} |\Sigma|^{n/2}\sum_{\vz \in \Sigma^n}(V\cdot W)(\vz)\ket{0}\ket{\vz}.
\end{align*}
\end{lemma}
\begin{proof}
\Cref{eq:hatV_hatW,eq:hatV_hatW_two} immediately imply the following inequalities, respectively:
\begin{align*} 
    \left\|\sum_{(\vx,\ve) \in \bad}\hat{V}(\vx)\hat{W}(\ve)\ket{\vx}\ket{\ve}\right\|\leq \sqrt{\epsilon}
\end{align*}
and 
\begin{align*} 
    \left\|\sum_{(\vx,\ve) \in \bad}\hat{V}(\vx)\hat{W}(\ve)\ket{\vx+\ve}\right\|\leq \sqrt{\delta}.
\end{align*}
Since $\bad$ is the complement of $\good$, the above imply the following:
\begin{align} \label{eq:approx_one}
\sum_{(\vx,\ve) \in \Sigma^n\times \Sigma^n}\hat{V}(\vx)\hat{W}(\ve)\ket{\vx}\ket{\ve}
\approx_{\sqrt{\epsilon}} 
    \sum_{(\vx,\ve) \in \good}\hat{V}(\vx)\hat{W}(\ve)\ket{\vx}\ket{\ve}
\end{align}
and 
\begin{align} \label{eq:approx_two}
\sum_{(\vx,\ve) \in \Sigma^n\times \Sigma^n}\hat{V}(\vx)\hat{W}(\ve)\ket{\vx+\ve}
\approx_{\sqrt{\delta}} 
\sum_{(\vx,\ve) \in \good}\hat{V}(\vx)\hat{W}(\ve)\ket{\vx+\ve}.
\end{align}

Then, we have 
\begin{align*}
   U_{F}U_{\mathsf{add}}(\QFT_{\Sigma}^{\otimes n}\otimes \QFT_{\Sigma}^{\otimes n})\ket{\psi}\ket{\phi}
   &=
   U_{F}U_{\mathsf{add}}\sum_{(\vx,\ve) \in \Sigma^n\times \Sigma^n}\hat{V}(\vx)\hat{W}(\ve)\ket{\vx}\ket{\ve}\\
   &\approx_{\sqrt{\epsilon}} U_{F}U_{\mathsf{add}}\sum_{(\vx,\ve) \in \good}\hat{V}(\vx)\hat{W}(\ve)\ket{\vx}\ket{\ve}\\
   &=\sum_{(\vx,\ve) \in \good}\hat{V}(\vx)\hat{W}(\ve)\ket{0}\ket{\vx+\ve}\\
   &\approx_{\sqrt{\delta}}\sum_{(\vx,\ve) \in \Sigma^n\times \Sigma^n}\hat{V}(\vx)\hat{W}(\ve)\ket{0}\ket{\vx+\ve}\\
   &=\sum_{\vz \in \Sigma^n}(\hat{V}*\hat{W})(\vz)\ket{0}\ket{\vz}\\
   &=|\Sigma|^{n/2}\sum_{\vz \in \Sigma^n}\widehat{(V\cdot W)}(\vz)\ket{0}\ket{\vz}\\
   &=(I\otimes \QFT_\Sigma^{\otimes n})|\Sigma|^{n/2}\sum_{\vz \in \Sigma^n}(V\cdot W)(\vz)\ket{0}\ket{\vz}
\end{align*}
where we used 
\Cref{eq:approx_one} for the second line, 
\Cref{eq:approx_two} for the fourth line, and 
the convolution theorem (\Cref{eq:conv_one} in \Cref{lem:convolution}) for the sixth line. 
This completes the proof of \Cref{lem:Regev_like}. 
\end{proof}
\section{Proofs of Quantumness}\label{sec:PoQ}
In this section, we give a construction of proofs of quantumness in the QROM, which is the main result of this paper.

\begin{theorem}\label{thm:PoQ}
There exists a keyless proof of quantumness relative to a random oracle that satisfies soundness in the CROM.
\end{theorem}
 
By \Cref{thm:poq_uniform_to_non-uniform}, we immediately obtain the following corollary. 
\begin{corollary}\label{cor:PoQ_non-uniform}
There exists a keyed proof of quantumness relative to a random oracle that satisfies soundness in the AI-CROM. 
\end{corollary}

The rest of this subsection is devoted to a proof of \Cref{thm:PoQ}.

\paragraph{Construction.}
Let $\{C_\secpar\}_\secpar$ be a family of codes over an alphabet $\Sigma=\FF_q^m$ that satisfies the requirements of \Cref{lem:good_codes} with arbitrary $0<c<c'<1$. 
In the following, we omit $\secpar$ from the subscript of $C$ since it is clear from the context. 
We use notations defined in \Cref{lem:good_codes} (e.g., $n,m,\zeta,\ell,L$ etc). 
Let $H:\Sigma\rightarrow \{0,1\}^{n}$ be a random oracle.\footnote{Strictly speaking, we consider a random oracle with the domain $\bit^*$. However, since our construction only makes queries to $H$ on (bit representaions of) elements of $\Sigma$ for a fixed security parameter, we simply denote by $H$ to mean the restriction of $H$ to (bit representations of) $\Sigma$.}
For $i\in [n]$, let $H_i:\Sigma \ra \bit$ be a function that on input $x$ outputs the $i$-th bit of $H(x)$. 
Then, we construct a proof of quantumness $\Pi=(\prove,\verify)$ in the QROM as follows.
\begin{description}
\item[$\prove^{H}(1^\secpar)$:]
For $i\in [n]$, it generates a state
\[
\ket{\phi_i}\propto \sum_{
\ve_i\in \Sigma \text{~s.t.~} H_i(\ve_i)= 1
}  \ket{\ve_i}.
\]
This is done as follows. 
It generates a uniform superposition over $\Sigma$, coherently evaluates $H$, and measures its value. 
If the measurement outcome is $1$,  then it succeeds in generating the above state.  
It repeats the above procedure until it succeeds or it fails $\secpar$ times. 
If it fails to generate $\ket{\phi_i}$ within $\secpar$ trials for some $i\in [n]$, it just aborts.
Otherwise, it sets 
\[
\ket{\phi}\defeq \ket{\phi_1}\otimes\ket{\phi_2}\otimes\ldots\otimes \ket{\phi_n}. 
\]
Note that we have 
\[
\ket{\phi}\propto\sum_{\substack{\ve=(\ve_1,\ldots,\ve_n)\in \Sigma^n\text{~s.t.~}\\
H_i(\ve_i)=1\text{~for~all~}i\in[n]}}\ket{\ve}.
\]
It generates a state 
\[
\ket{\psi}\propto \sum_{\vx \in C}\ket{\vx}.
\]
Then it applies $\QFT_\Sigma^{\otimes n}$ to both $\ket{\psi}$ and $\ket{\phi}$.
At this point, it has the state
\[
\ket{\eta}\defeq \QFT_\Sigma^{\otimes n} \ket{\psi} \otimes \QFT_\Sigma^{\otimes n} \ket{\phi}.
\]
Let  $U_{\mathsf{add}}$ and $U_{\mathsf{decode}}$ be unitaries on the Hilbert space of $\ket{\eta}$ defined by the following:
\[
\ket{\vx}\ket{\ve}
\xrightarrow{U_{\mathsf{add}}}
\ket{\vx}\ket{\vx+\ve}
\xrightarrow{U_{\mathsf{decode}}}
\ket{\vx-\decode_{C^\perp}(\vx+\ve)}\ket{\vx+\ve}
\]
where $\decode_{C^\perp}$ is the decoder for $C^\perp$ as required in \Cref{item:dual_decode} of \Cref{lem:good_codes}.
Then it applies $(I\otimes(\QFT_\Sigma^{-1})^{\otimes n}) U_{\mathsf{decode}}U_{\mathsf{add}}$  to $\ket{\eta}$, measures the second register, and outputs the measurement outcome $\vx\in \Sigma^n$ as $\pi$.  A diagram showing how to compute $\pi$ is given in Figure~\ref{fig:prove}.

\item[$\verify^{H}(1^\secpar,\pi)$:]
It parses $\pi=\vx=(\vx_1,\ldots,\vx_n)$ and
outputs $\top$ if $\vx\in C$ and $H_i(\vx_i)= 1$ for all $i\in [n]$ 
and $\bot$ otherwise.
\end{description}
\begin{figure}
    \centering
    \includegraphics[width=\textwidth]{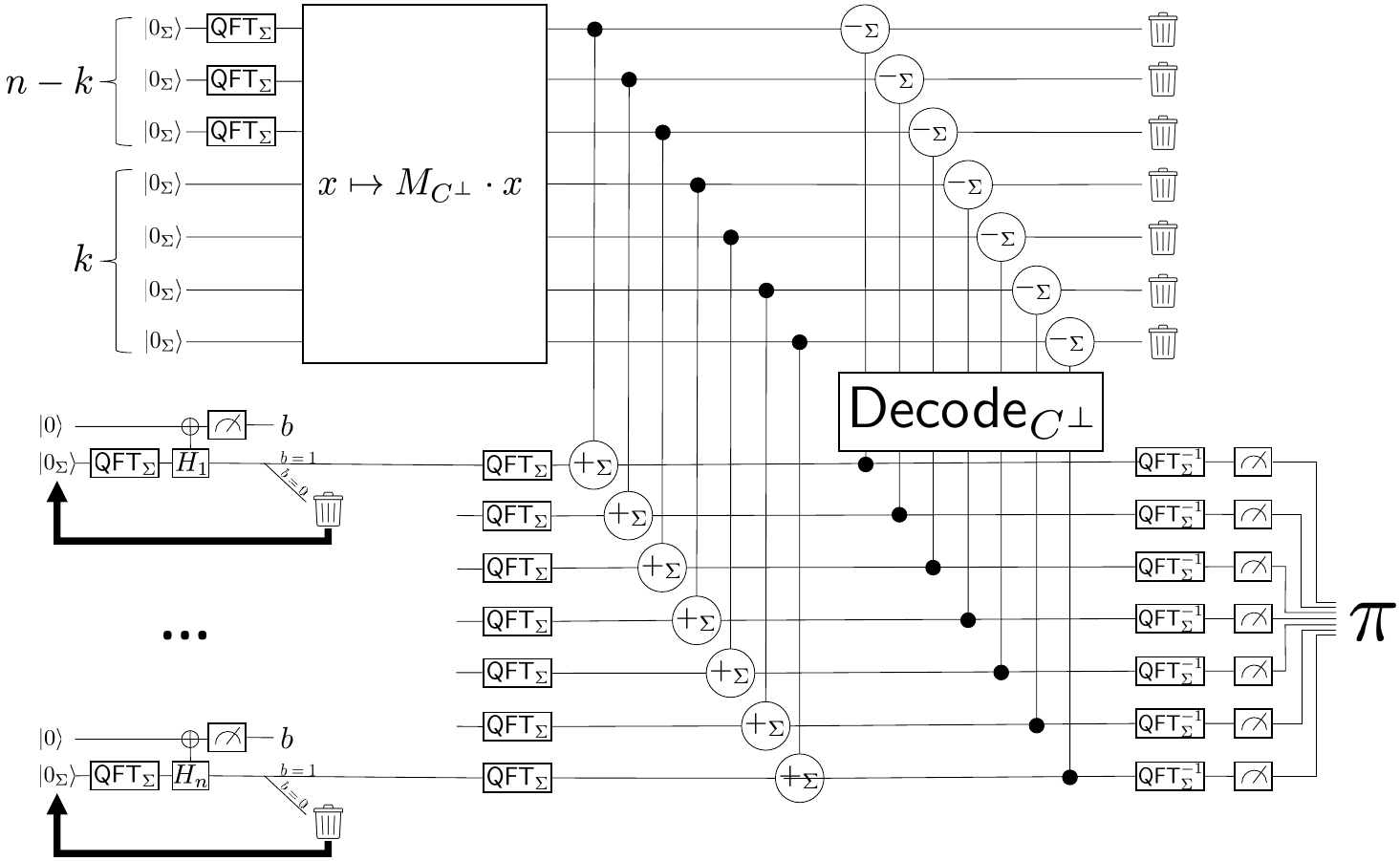}
    \caption{The algorithm $\prove$ for computing $\pi$. Here, $n-k$ is the dimension of $C^\perp$, and $M_{C^\perp}$ is any invertible matrix whose first $n-k$ columns are a basis for $C^\perp$.}
    \label{fig:prove}
\end{figure}

\paragraph{Correctness.}
\begin{lemma}\label{lem:correctness}
$\Pi$ satisfies correctness.
\end{lemma}
\begin{proof}

Let $T_i^{H_i}\subseteq \Sigma$ be the subset consisting of $\ve_i\in \Sigma$ such that $H_i(\ve_i)=1 $ and 
$T^H\defeq T_1^{H_1}\times T_2^{H_2}\times \ldots \times T_n^{H_n}\subseteq \Sigma^n$. 
 Let $\hashset\subseteq \Func(\Sigma,\bit^n)$ be the subset that consists of all $H\in \Func(\Sigma,\bit^n)$ such that $\frac{1}{3}<\frac{|T_i^{H_i}|}{|\Sigma|}<\frac{2}{3}$ for all $i\in [n]$. 
By the Chernoff bound (\Cref{lem:Chernoff}) and union bound, we can see that $(1-n\cdot 2^{-\Omega(|\Sigma|)})$-fraction of $H\in \func(\Sigma,\bit^n)$ belongs to $\hashset$. Since we have $n\cdot 2^{-|\Sigma|}=\negl(\secpar)$ by our parameter choices specified in \Cref{lem:good_codes}, it suffices to prove the correctness assuming that $H$ is uniformly chosen from $\hashset$ instead of from $\Func(\Sigma,\bit^n)$. 
We prove this below.

First, we show that the probability that $\prove$ aborts is negligible. 
In each trial to generate $\ket{\phi_i}$, the success probability is $\frac{|T_i^{H_i}|}{|\Sigma|}<\frac{2}{3}$. 
Thus, the probability that it fails to generate $\ket{\phi_i}$ $\secpar$ times is negligible.  

Let $V:\Sigma^n\ra \mathbb{C}$, $W^{H_i}_i:\Sigma\ra \mathbb{C}$, and $W^H:\Sigma^n\ra \mathbb{C}$ be functions defined as follows:\footnote{Since we assume that $H$ is sampled from $\hashset$, we do not define them when $|T_i^{H_i}|=0$ for some $i$.} 
\begin{align*}
    &V(\vx)=
    \begin{cases}
    \frac{1}{\sqrt{|C|}}& \vx\in C\\
    0& \text{otherwise}
    \end{cases}\\
    &W^{H_i}_i(\ve_i)=
    \begin{cases}
    \frac{1}{\sqrt{|T^{H_i}_i|}}& \ve_i\in T^{H_i}_i\\
    0& \text{otherwise}
    \end{cases}\\
    &W^H(\ve)=
    \begin{cases}
    \frac{1}{\sqrt{|T^H|}}& \ve\in T^H\\
    0& \text{otherwise}
    \end{cases}
\end{align*}

Then we have 
\begin{align*}
    &\ket{\psi}=\sum_{\vx \in \Sigma^n}V(\vx)\ket{\vx}\\
    &\ket{\phi}=\sum_{\ve \in \Sigma^n}W^H(\ve)\ket{\ve}
\end{align*}
where $\ket{\psi}$ and $\ket{\phi}$ are as in the description of $\prove$.
For using \Cref{lem:Regev_like}, 
we prove the following claim.
\begin{claim}\label{cla:conditions}
For an overwhelming fraction of $H\in \hashset$, there is a subset $\good\subseteq \Sigma^n \times \Sigma^n$ such that $\decode_{C^\perp}(\vx+\ve)=\vx$ for any $(\vx,\ve)\in \good$ and we have
\begin{align*}
&\sum_{(\vx,\ve)\in \bad}|\hat{V}(\vx)\hat{W}^H(\ve)|^2\leq \negl(\secpar),\\
&\sum_{\vz\in \Sigma^n}\left|\sum_{(\vx,\ve)\in \bad: \vx+\ve=\vz}\hat{V}(\vx)\hat{W}^H(\ve)\right|^2\leq\negl(\secpar).
\end{align*}
where $\bad=(\Sigma^n \times \Sigma^n)\setminus \good$.
\end{claim}
We prove \Cref{cla:conditions} later. 
We complete the proof of \Cref{lem:correctness} by using \Cref{cla:conditions}.
By \Cref{lem:Regev_like} and \Cref{cla:conditions} where we set $F\defeq \decode_{C^\perp}$, for an overwhelming fraction of $H\in \hashset$, we have
\begin{align} \label{eq:final_state}
    (I\otimes (\QFT_\Sigma^{-1})^{\otimes n})U_{\mathsf{decode}}U_{\mathsf{add}}\ket{\eta}
    \approx
    |\Sigma|^{n/2}\sum_{\vx \in \Sigma^n}(V\cdot W^H)(\vx)\ket{0}\ket{\vx}
\end{align}
where $\ket{\eta}$ is as in the description of $\prove$.
Since $(V\cdot W^H)(\vx)=0$ for $\vx\notin C\cap T^H$, if we measure the second register of the RHS of \Cref{eq:final_state}, the outcome is in $C\cap T^H$ with probability $1$. Thus, 
if we measure the second register of the LHS of \Cref{eq:final_state}, the outcome is in $C\cap S$ with probability $1-\negl(\secpar)$.
This means that an honestly generated proof $\pi$ passes the verification with probability $1-\negl(\secpar)$. 
\end{proof}
To complete the proof of correctness, we prove \Cref{cla:conditions} below.
\begin{proof}[Proof of \Cref{cla:conditions}]
We use the notations defined in the proof of \Cref{lem:correctness} above. 
For each $i\in [n]$, let $\hashset_i\subseteq \Func(\Sigma,\bit)$ be the subset that consists of all $H_i\in \Func(\Sigma,\bit)$ such that $\frac{1}{3}<\frac{|T_i^{H_i}|}{|\Sigma|}<\frac{2}{3}$.\footnote{Mathematically, the set $\hashset_i$ does not depend on $i$. We index it by $i$ for notational convenience.} 
Choosing $H\sample \hashset$ is equivalent to choosing $H_i\sample \hashset_i$ independently for each $i\in [n]$.  
In the following, 
whenever we write $H$ or $H_i$ in subscripts of $\Ex$, they are uniformly taken from $\hashset$ or $\hashset_i$, respectively. 

By \Cref{lem:fourier_dual} and the definition of $V$, we have
\begin{align*}
    &\hat{V}(\vx)=
    \begin{cases}
    \frac{1}{\sqrt{|C^\perp|}}& \vx\in C^\perp\\
    0& \text{otherwise}
    \end{cases}.
\end{align*}

Let $\gooderrors\subseteq \Sigma^n$ be a subset defined as follows:
\begin{align*}
    \gooderrors\defeq \{\ve \in \Sigma^n: \forall \vx\in C^\perp,~\decode_{C^\perp} (\vx+\ve)=\vx\}.
\end{align*}
Let $\baderrors\defeq \Sigma^n \setminus \gooderrors$.  
\Cref{item:dual_decode} of \Cref{lem:good_codes} implies
\begin{align}\label{eq:prob_baderror}
    \Pr_{\ve\sample \dist^n}[\ve\in \baderrors]=\negl(\secpar)
\end{align}
where $\dist$ is the distribution as defined in \Cref{item:dual_decode} of \Cref{lem:good_codes}. 
We define $\good\defeq C^\perp \times \gooderrors$ and  $\bad\defeq (\Sigma^n\times \Sigma^n)\setminus \good$. 
Then, we have $\decode_{C^\perp}(\vx+\ve)=\vx$ for all $(\vx,\ve)\in \good$ by definition. 

Noting that $\hat{V}(\vx)=0$ for $\vx\notin C^{\perp}$, 
it is easy to see that we have the following:
\begin{align}
&\sum_{(\vx,\ve)\in \bad}|\hat{V}(\vx)\hat{W}^H(\ve)|^2
=\sum_{\ve\in \baderrors}|\hat{W}^H(\ve)|^2,
\label{eq:hatV_hatW_construction}\\
&\sum_{\vz\in \Sigma^n}\left|\sum_{(\vx,\ve)\in \bad: \vx+\ve=\vz}\hat{V}(\vx)\hat{W}^H(\ve)\right|^2
=\sum_{\vz\in \Sigma^n}\left|\sum_{\substack{\vx\in C^\perp,\ve\in \baderrors\\: \vx+\ve=\vz}}\hat{V}(\vx)\hat{W}^H(\ve)\right|^2.
\label{eq:hatV_hatW_two_construction}
\end{align}
We should prove that values of \Cref{eq:hatV_hatW_construction,eq:hatV_hatW_two_construction} are negligible for an overwhelming fraction of $H\in \hashset$. 
By a standard averaging argument, it suffices to prove that their expected values are negligible, i.e., 
\begin{align}
&\Ex_{H}\left[\sum_{\ve\in \baderrors}|\hat{W}^H(\ve)|^2\right]\leq \negl(\secpar),
\label{eq:hatV_hatW_exp}\\
&\Ex_{H}\left[\sum_{\vz\in \Sigma^n}\left|\sum_{\substack{\vx\in C^\perp,\ve\in \baderrors\\: \vx+\ve=\vz}}\hat{V}(\vx)\hat{W}^H(\ve)\right|^2\right]\leq \negl(\secpar).
\label{eq:hatV_hatW_two_exp}
\end{align}

Before proving them, we remark an obvious yet useful claim.
\begin{claim}\label{cla:symmetry}
Let $\pi$ be a permutation over $\Sigma$ (resp. $\Sigma^n$). Then, the distributions of $H_i$ and $H_i\circ \pi$ (resp. $H$ and $H\circ \pi$) are identical when $H_i\sample \hashset_i$ (resp. $H\sample \hashset$). 
\end{claim}
\begin{proof}[Proof of \Cref{cla:symmetry}]
Recall that $\hashset_i$ is the set of all $H_i:\Sigma\rightarrow \bit$ such that $\frac{|\Sigma|}{3} < |\{\ve_i\in \Sigma:H(\ve_i)=1\}|<\frac{2|\Sigma|}{3}$.
Clearly, we have $|\{\ve_i\in \Sigma:H(\ve_i)=1\}|=|\{\ve_i\in \Sigma:H\circ \pi(\ve_i)=1\}|$. Thus, $\pi$ induces a permutation over $\hashset_i$, and thus $H_i\circ \pi$ is uniformly distributed over $\hashset_i$ when $H_i\sample \hashset_i$. A similar argument works for $\hashset$ as well. 
\end{proof}

Then, we prove \Cref{eq:hatV_hatW_exp,eq:hatV_hatW_two_exp}. 



\smallskip
\noindent\textbf{Proof of \Cref{eq:hatV_hatW_exp}.}
First, we prove the following claim. 
\begin{claim}\label{cla:hatW}
For all $i\in[n]$ and $\ve,\ve'\in \Sigma\setminus\{0\}$, it hold that  
\begin{align}
\Ex_{H_i}\left[|\hat{W}_i(\vzero)|^2\right] =\frac{1}{2}  \label{eq:hatW_zero}
\end{align}
and 
\begin{align}
\Ex_{H_i}\left[|\hat{W}_i(\ve)|^2\right] =\Ex_{H_i}\left[|\hat{W}_i(\ve')|^2\right].  \label{eq:hatW_others}
\end{align}
\end{claim}
\begin{proof}[Proof of \Cref{cla:hatW}]
\Cref{eq:hatW_zero} is proven as follows. 
\begin{align*}
\Ex_{H_i}\left[|\hat{W}_i(\vzero)|^2\right]
&=\Ex_{H_i}\left[\left|\frac{1}{\sqrt{|\Sigma|}}\sum_{\vz\in \Sigma}W_i^{H_i}(\vz)\right|^2\right]
=\frac{\Ex_{H_i}\left[|T_i^{H_i}|\right]}{|\Sigma|} =\frac{1}{2}.  
\end{align*}
Since $\ve\neq \vzero$, for any $w\in \FF_q$, the number of $\vz\in \Sigma$ such that $\ve\cdot \vz=w$ is $|\Sigma|/q$. 
A similar statement holds for $\ve'$ too. 
Therefore, there is a permutation $\pi_{\ve,\ve'}:\Sigma\rightarrow \Sigma$ such that $\ve\cdot \vz=\ve'\cdot \pi_{\ve,\ve'}(\vz)$ for all $\vz\in \Sigma$. 
Then, \Cref{eq:hatW_others} is proven as follows.
\begin{align*}
\Ex_{H_i}\left[|\hat{W}_i(\ve)|^2\right]
&=\Ex_{H_i}\left[\left|\frac{1}{\sqrt{|\Sigma|}}\sum_{\vz\in \Sigma}W_i^{H_i}(\vz)\omega_p^{\Tr(\ve\cdot \vz)}\right|^2\right]\\
&=\Ex_{H_i}\left[\left|\frac{1}{\sqrt{|\Sigma|}}\sum_{\vz\in \Sigma}W_i^{H_i\circ \pi_{\ve,\ve'}^{-1}}(\pi_{\ve,\ve'}(\vz))\omega_p^{\Tr(\ve'\cdot\pi_{\ve,\ve'}(\vz))}\right|^2\right]\\
&=\Ex_{H_i}\left[\left|\frac{1}{\sqrt{|\Sigma|}}\sum_{\vz\in \Sigma}W_i^{H_i\circ \pi_{\ve,\ve'}^{-1}}(\vz)\omega_p^{\Tr(\ve'\cdot\vz)}\right|^2\right]\\
&=\Ex_{H_i}\left[\left|\frac{1}{\sqrt{|\Sigma|}}\sum_{\vz\in \Sigma}W_i^{H_i}(\vz)\omega_p^{\Tr(\ve'\cdot \vz)}\right|^2\right]\\
&=\Ex_{H_i}\left[|\hat{W}_i(\ve')|^2\right]
\end{align*}
where the fourth equality follows from \Cref{cla:symmetry}. 
\end{proof}

\Cref{cla:hatW} means that we have 
\begin{align} \label{eq:dist_and_Ex}
    \dist(\ve_i)= \Ex_{H_i}\left[|\hat{W}_i(\ve_i)|^2\right]
\end{align}
for all $\ve_i\in \Sigma$ where $\dist(\cdot)$ is the probability density function of the distribution $\dist$ as defined in \Cref{item:dual_decode} of
\Cref{lem:good_codes}. 
Moreover, for any $\ve=(\ve_1,\ldots,\ve_n)\in \Sigma^n$ and $H\in \hashset$, 
since we have
    $W^H(\ve)=\prod_{i=1}^{n}W_i^{H_i}(\ve_i)$, 
by \Cref{lem:QFT_prod}, we have
\begin{align}\label{eq:hatW:mult}
    \hat{W}^{H}(\ve)=\prod_{i=1}^{n}\hat{W}_i^{H_i}(\ve_i).
\end{align}
By combining \Cref{eq:dist_and_Ex,eq:hatW:mult}, we obtain 
\begin{align} \label{eq:bardist_and_Ex}
    \dist^n(\ve)= \Ex_{H}\left[|\hat{W}(\ve)|^2\right]
\end{align}
for all $\ve\in \Sigma^n$ where $\dist^n(\cdot)$ is the probability density function of $\dist^n$. By 
\Cref{eq:prob_baderror}, \Cref{eq:bardist_and_Ex}, and the linearity of expectation, we obtain  \Cref{eq:hatV_hatW_exp}. 

\smallskip
\noindent\textbf{Proof of \Cref{eq:hatV_hatW_two_exp}.}
We define a function $B:\Sigma^n\rightarrow \mathbb{C}$ so that $\hat{B}$ satisfies the following:\footnote{That is, we first define $\hat{B}$ and then define $B$ as its inverse discrete Fourier transform.}
\begin{align*}
\hat{B}(\ve)=
\begin{cases}
1& \ve\in \baderrors\\
0& \text{otherwise}
\end{cases}.
\end{align*}
We prove the following claims.
\begin{claim}\label{cla:hatVhatW_conv_bad}
For any $H\in \hashset$, it holds that 
\begin{align*}
\sum_{\vz\in \Sigma^n}\left|\sum_{\substack{\vx\in C^\perp,\ve\in \baderrors\\: \vx+\ve=\vz}}\hat{V}(\vx)\hat{W}^H(\ve)\right|^2
=\sum_{\vz\in \Sigma^n}\left|(V\cdot (B\ast W^H))(\vz)\right|^2.
\end{align*}
\end{claim}
\begin{proof}[Proof of \Cref{cla:hatVhatW_conv_bad}]
For any $\vz\in \Sigma^n$, we have
\begin{align*}
\sum_{\substack{\vx\in C^\perp,\ve\in \baderrors\\: \vx+\ve=\vz}}\hat{V}(\vx)\hat{W}^H(\ve)
&=\sum_{\substack{
\vx\in \Sigma^n, \ve\in \Sigma^n\\
:\vx+\ve=\vz}}\hat{V}(\vx)(\hat{B}\cdot \hat{W}^H)(\ve)\\
&=(\hat{V}\ast  (\hat{B}\cdot \hat{W}^H))(\vz)\\
&=\widehat{(V\cdot (B\ast W^H))}(\vz)
\end{align*}
where 
we used $\hat{V}(\vx)=0$ for $\vx\notin C^\perp$ in the first equality and 
 the convolution theorem (\Cref{eq:fgh} in \Cref{lem:convolution}) in the third equality. 
\Cref{cla:hatVhatW_conv_bad} follows from the above equation and  Parseval's equality (\Cref{lem:Parseval}). 
\end{proof}
\begin{claim}\label{cla:ex_BW_conv}
For any $\vz\in \Sigma^n$, it holds that  
\begin{align*}
\Ex_{H}\left[|(B\ast W^H)(\vz)|^2\right]\leq \negl(\secpar).
\end{align*}
\end{claim}
\begin{proof}[Proof of \Cref{cla:ex_BW_conv}.]
First, we observe that $\Ex_{H}\left[|(B\ast W^H)(\vz_0)|^2\right]=\Ex_{H}\left[|(B\ast W^H)(\vz_1)|^2\right]$ for any $\vz_0,\vz_1$. 
Indeed, 
if we define a permutation $\pi:\Sigma^n\rightarrow \Sigma^n$ as $\pi(\vz)\defeq \vz+\vz_0-\vz_1$, 
we have 
\begin{align*}
   &\Ex_{H}\left[\left|(B\ast W^H)(\vz_0)\right|^2\right]\\
   =&\Ex_{H}\left[\left|\sum_{\vx\in \Sigma^n}B(\vx)W^H(\vz_0-\vx)\right|^2\right]\\
   =&\Ex_{H}\left[\left|\sum_{\vx\in \Sigma^n}B(\vx)W^{H\circ \pi}(\vz_1-\vx)\right|^2\right]\\
   =&\Ex_{H}\left[\left|\sum_{\vx\in \Sigma^n}B(\vx)W^H(\vz_1-\vx)\right|^2\right]\\
   =&\Ex_{H}\left[\left|(B\ast W^H)(\vz_1)\right|^2\right]
\end{align*}
where the third equality follows from \Cref{cla:symmetry}. 

Then, for any $\vz \in \Sigma^n$, 
we have 
\begin{align*}
   &\Ex_{H}\left[\left|(B\ast W^H)(\vz)\right|^2\right]\\
   =&\frac{1}{|\Sigma|^n}\sum_{\vz\in \Sigma^n}\Ex_{H}\left[\left|(B\ast W^H)(\vz)\right|^2\right]\\
    =&\frac{1}{|\Sigma|^n}\Ex_{H}\left[\sum_{\vz\in \Sigma^n}\left|(B\ast W^H)(\vz)\right|^2\right]\\
     =&\frac{1}{|\Sigma|^n}\Ex_{H}\left[\sum_{\vz\in \Sigma^n}\left||\Sigma|^{n/2}(\hat{B}\cdot \hat{W}^H)(\vz)\right|^2\right]\\
         =&\Ex_{H}\left[\sum_{\vz\in \baderrors}\left| \hat{W}^H(\vz)\right|^2\right]\\
    \leq & \negl(\secpar).     
\end{align*}   
where the third equality follows from the convolution theorem (\Cref{eq:conv_two} in \Cref{lem:convolution}) and Parseval's equality (\Cref{lem:Parseval}) and the final inequality follows from \Cref{eq:hatV_hatW_exp}. 
\end{proof}
Then, we prove \Cref{eq:hatV_hatW_two_exp} as follows: 
\begin{align*}
    &\Ex_{H}\left[\sum_{\vz\in \Sigma^n}\left|\sum_{\substack{\vx\in C^\perp,\ve\in \baderrors\\: \vx+\ve=\vz}}\hat{V}(\vx)\hat{W}^H(\ve)\right|^2\right]\\
    =&\Ex_{H}\left[\sum_{\vz\in \Sigma^n}\left|(V\cdot (B\ast W^H))(\vz)\right|^2\right]\\
    =&\Ex_{H}\left[\sum_{\vz\in C}\frac{1}{|C|}\left|(B\ast W^H))(\vz)\right|^2\right]\\
    =&\frac{1}{|C|}\sum_{\vz\in C}\Ex_{H}\left[\left|(B\ast W^H))(\vz)\right|^2\right]\\
    \leq&\negl(\secpar).
\end{align*}
where the first equality follows from \Cref{cla:hatVhatW_conv_bad}, the second equality follows from the definition of $V$, and the final inequality follows from  \Cref{cla:ex_BW_conv}. 

This completes the proof of \Cref{cla:conditions}.
\end{proof}

\smallskip\noindent\textbf{Soundness.}
\begin{lemma}\label{lem:soundness}
$\Pi$ satisfies $(2^{\secpar^c},2^{-\Omega(\secpar)})$-soundness in the CROM.
\end{lemma}


\begin{proof}
Let $\A$ be an adversary that makes $Q\leq 2^{\secpar^{c}}$ classical queries to $H$. 
Without loss of generality, we assume that $\A$ queries $\vx^*_i$ to $H$ at some point for all $i\in [n]$ where $\vx^*=(\vx^*_1,...,\vx^*_n)\in \Sigma^n$ is $\A$'s final output. 
Since a query to $H$ can be replaced with queries to each of $H_1,\ldots,H_n$, there is an adversary $\A'$ that makes  $Q$ queries to each of $H_1$,...,$H_n$ and succeeds with the same probability as $\A$. 
We denote $\A'$'s total number of queries by $Q'=nQ$.
We remark that $\A'$ queries $\vx^*_i$ to $H_i$ at some point by our simplifying assumption on $\A$. 

For each $i\in[n]$ and $j\in [Q']$, let $S_i^j \subseteq \Sigma$ be the set of elements that $\A'$ ever queried to $H_i$ by the point when it has just made the $j$-th query counting queries to any of $H_1,...,H_n$ in total. 
After the $j$-th query, we say that a codeword $\vx=(\vx_1,...,\vx_n)\in C$ is \emph{$K$-queried} if there is a subset $I\in [n]$ such that $|I|= K$, $\vx_i\in S_i^j$ for all $i\in I$, and $\vx_i\notin S_i^j$ for all $i\notin I$.
By our assumption, the final output $\vx^*$ must be $n$-queried at the end. 
Since a $K$-queried codeword either becomes $(K+1)$-queried or remains $K$-queried 
by a single query, $\vx^*$ must be $K$-queried at some point of the execution of $\A'$ for all $K=0,1,...,n$.

We consider the number of codewords that ever become $K$-queried for $K=\lceil(1-\zeta)n \rceil$ where $\zeta$ is the constant as in \Cref{item:list_recovery} of \Cref{lem:good_codes}. 
If $\vx=(\vx_1,...,\vx_n)\in C$ is $\lceil(1-\zeta)n \rceil$-queried at some point, the number of $i$ such that $\vx_i\in S_i^{Q'}$ is at least $\lceil(1-\zeta)n \rceil$ since  $S_i^j\subseteq S_i^{Q'}$ for all $i,j$. By the construction of $\A'$, we have $|S_i^{Q'}|= Q\leq 2^{\secpar^c}$. 
On the other hand, $C$ is $(\zeta,\ell,L)$-list recoverable for $\ell=2^{\secpar^c}$ and $L=2^{\tilde{O}(\secpar^{c'})}$ as required in \Cref{item:list_recovery} of \Cref{lem:good_codes}.
Thus, the number of  codewords that ever become $\lceil(1-\zeta)n \rceil$-queried is at most $L=2^{\tilde{O}(\secpar^{c'})}$.   

Let $E_i$ be the event that the $i$-th codeword that 
becomes $\lceil(1-\zeta)n \rceil$-queried is finally output by $\A'$. 
Here, if multiple codewords become $\lceil(1-\zeta)n \rceil$-queried at the same time, we order them according to the lexicographical ordering. 
By the above argument, we have 
\begin{align}\label{eq:win}
    \Pr[\A'\text{~wins}]=\sum_{i\in [L]} \Pr[\A'\text{~wins} \land E_i]
\end{align}
where we say that $\A'$ wins if its output passes the verification. 
Moreover, we show that for each $i\in[L]$, 
\begin{align}\label{eq:win_i}
    \Pr[\A'\text{~wins} \land E_i]=2^{-\Omega(\secpar)}. 
\end{align}
This can be seen as follows.  
Suppose that we simulate oracles $H_1,...,H_n$ for $\A'$ via lazy sampling, that is, instead of uniformly choosing random functions at first, we sample  function values whenever they are queried by $\A'$. 
Let $\vx$ be the $i$-th codeword that becomes $\lceil(1-\zeta)n \rceil$-queried in the execution of $\A'$. Since the function values on the unqueried $\lfloor\zeta n \rfloor$ positions are not sampled yet, $\vx$ can become a valid proof only if all those values happen to be $1$, which occurs with probability $\left(\frac{1}{2}\right)^{\lfloor\zeta n \rfloor}=2^{-\Omega(\secpar)}$ by $\zeta=\Omega(1)$ and $n=\Omega(\secpar)$. This implies \Cref{eq:win_i}. 

By combining \Cref{eq:win,eq:win_i} and $L=2^{\tilde{O}(\secpar^{c'})}$ for $c'<1$, 
we complete the proof of \Cref{lem:soundness}.
\end{proof}

\Cref{thm:PoQ} follows from \Cref{lem:correctness,lem:soundness}.

\paragraph{Achieving worst-case correctness.}
Remark that the correctness proven in \Cref{lem:correctness} only ensures that the proving algorithm succeeds with an overwhelming probability over the random choice of the oracle $H$. 
Below, we show a modified protocol for which we can show that the correctness holds for \emph{any} $H$, while still preserving soundness on random $H$. 

The motivation of achieving worst-case correctness is as follows. 
In the query-complexity literature (e.g., \cite{JACM:BBCMW01,BW02,AA14}), it is more common to think of an oracle as an (exponentially large) ``input" rather than a function. In that context, the (classical, randomized, or quantum) query complexity of a task is defined to be the minimum number of queries that is needed to solve the task with probability at least $2/3$ \textbf{for all} inputs. Viewing our problem from this perspective, it is natural to require correctness to hold \textbf{for all} possible oracles $H$.

\paragraph{Construction.}
Let $\{C_\secpar\}_\secpar$ be a family of codes over an alphabet $\Sigma=\FF_q^m$ that satisfies the requirements of \Cref{lem:good_codes} with arbitrary $1<c<c'<1$. 
Let  $H:[t]\times\Sigma\rightarrow \{0,1\}^{n}$ be a random oracle where $t$ is a positive integer specified later. 
For $j\in[t]$, we define $H^{(j)}:\Sigma\rightarrow \{0,1\}^{n}$ by $H^{(j)}(x):=H(j\concat x)$. Let $\mathcal{F}=\{f_K:\Sigma\rightarrow \{0,1\}^{n}\}_{K\in \mathcal{K}}$ be a family of $2(\secpar n+1)$-wise independent hash functions. 
Then, we construct a proof of quantumness $\widetilde{\Pi}=(\widetilde{\prove},\widetilde{\verify})$ based on $\Pi=(\prove,\verify)$ as follows.
\begin{description}
\item[$\widetilde{\prove}^{H}(1^\secpar)$:]
It chooses $K\sample \mathcal{K}$ and defines a function $\widetilde{H}_K^{(j)}:\Sigma\rightarrow \{0,1\}^{n}$ by $\widetilde{H}_K^{(j)}(x):=H^{(j)}(x)\oplus f_K(x)$ for $j\in [t]$. Then, it runs $\pi^{(j)}\sample \prove^{\widetilde{H}_K^{(j)}}(1^\secpar)$ for $j\in [t]$ and outputs a proof $\widetilde{\pi}:=(K,\{\pi^{(j)}\}_{j\in [t]})$.

\item[$\widetilde{\verify}^{H}(1^\secpar,\widetilde{\pi})$:]
It parses $\widetilde{\pi}:=(K,\{\pi^{(j)}\}_{j\in [t]})$ and
outputs $\top$ if $\verify^{\widetilde{H}_K^{(j)}}(1^\secpar,\pi^{(j)})=\top$ for all $j\in [t]$ 
and $\bot$ otherwise.
\end{description}
\paragraph{Correctness.}
\begin{lemma}\label{lem:correctness_variant}
$\widetilde{\Pi}$ satisfies worst-case correctness, i.e., for any $H$, 
\[
\Pr\left[\widetilde{\verify}^{H}(1^\secpar,\widetilde{\pi})=\bot:
\begin{array}{l}
\widetilde{\pi} \sample \widetilde{\prove}^{H}(1^\secpar)
\end{array}
\right]\leq \negl(\secpar).
\]
\end{lemma}
\begin{proof}
For each $j\in [t]$ and fixed $H$, 
by the construction of $\prove$ and the definition of $\widetilde{H}_K^{(j)}$, 
we can view $\prove^{\widetilde{H}_K^{(j)}}$ as an oracle-algorithm that makes $\secpar n$ queries to $f_K$. 
Similarly, we can view $\verify^{\widetilde{H}_K^{(j)}}$ as an oracle-algorithm that makes a single query to $f_K$. 
Since the combination of $\prove^{\widetilde{H}_K^{(j)}}$ and $\verify^{\widetilde{H}_K^{(j)}}$ makes $\secpar n +1$ quantum queires to $f_K$, which is chosen from a family of $2(\secpar n + 1)$-wise independent hash functions, 
by \Cref{lem:simulation_QRO}, 
the probability that $\pi^{(j)}$ generated by   $\prove^{\widetilde{H}_K^{(j)}}$ passes $\verify^{\widetilde{H}_K^{(j)}}$ does not change even if $f_K$ is replaced with a uniformly random function. 
Moreover, if $f_K$ is replaced with a uniformly random function, the correctness of $\Pi$ immediately implies that $\pi^{(j)}$ generated by   $\prove^{\widetilde{H}_K^{(j)}}$ passes  $\verify^{\widetilde{H}_K^{(j)}}$ with an overwhelming probability (for each fixed $H$).
By taking the union bound over $j\in [t]$, 
$\pi^{(j)}$ generated by the  $\prove^{\widetilde{H}_K^{(j)}}$ passes  $\verify^{\widetilde{H}_K^{(j)}}$ for \emph{all} $j\in[t]$ with an overwhelming probability, which means that $\widetilde{\Pi}$ satisfies correctness.
\end{proof}

\smallskip\noindent\textbf{Soundness.}
\begin{lemma}\label{lem:soundness_variant}
$\widetilde{\Pi}$ satisfies $(2^{\secpar^c},|\mathcal{K}|\cdot 2^{-\Omega(t\secpar)})$-soundness in the CROM.
\end{lemma}
\begin{proof}(sketch.)
We observe that the proof of the soundness of $\Pi$ (\Cref{lem:soundness}) can be easily extended to prove $(2^{\secpar^c},2^{-\Omega(t\secpar)})$-soundness for the $t$-parallel repetition of $\Pi$. A similar soundness holds even if we use $\widetilde{H}_K^{(j)}$ as the oracle for the $i$-th instance for each \emph{fixed} $K$ since a random oracle shifted by $f_K$ behaves as another random oracle.
Thus, by taking the union bound over $K\in \mathcal{K}$, we obtain \Cref{lem:soundness_variant}.
\end{proof}

Since $|\mathcal{K}|=2^{\poly(\secpar)}$ for some polynomial $\poly$ that is independent of $t$, we can set $t=\poly(\secpar)$ so that $|\mathcal{K}|\cdot 2^{-\Omega(t\secpar)}=2^{-\Omega(\secpar)}$.

\section{Counterexamples for Cryptographic Primitives}\label{sec:separation_primitive}
In this section, we give constructions of cryptographic primitives that are secure in the CROM but insecure in the QROM. 
They are easy consequences of our proof of quantumness constructed in \Cref{sec:PoQ}. 

\subsection{Counterexample for One-Way Functions}\label{sec:separation_OWF}
We give a construction of a family of functions that is one-way in the CROM but not one-way in the QROM. 
It is easy to generically construct such a one-way function from proofs of quantumness. Indeed, we prove a stronger claim than that in \Cref{sec:CRH}.
Nonetheless, we give a direct construction with a similar structure to the proof of quantumness presented in \Cref{sec:PoQ}.  
An interesting feature of the direct construction which the generic construction does not have is that it is not even \emph{distributionally one-way} in the QROM as explained in \Cref{rem:dist_OW}. 

\begin{theorem}[Counterexample for one-way functions]\label{thm:separation_OWF}
There exists a family $\{f_\secpar\}_\secpar$ of efficiently computable oracle-aided functions that is one-way in the CROM but not one-way in the QROM. 
\end{theorem}
\begin{proof}

The construction of $f_\secpar$ is very similar to that of the proof of quantumness constructed in \Cref{sec:PoQ}. 
We rely on similar parameter settings as in \Cref{sec:PoQ}, and use similar notations as in \Cref{sec:PoQ}. 

We define $f_\secpar^H:C\rightarrow \bit^n$ as follows:
$$
f_\secpar^{H}(\vx_1,...,\vx_n)=(H_1(\vx_1),...,H_n(\vx_n)).
$$  
where  $H_i:\Sigma\rightarrow \bit$ is the function that outputs the $i$-th bit of the output of $H:\Sigma\rightarrow \bit^n$.

The $\prove$ algorithm in \Cref{sec:PoQ} can be understood as an algorithm to invert $f_\secpar$ for the image $1^n$ in the QROM. 
This can be extended to find a preimage
of any image $y\in \bit^n$ in a straightforward manner:
We only need to modify the definition of $T^H_i$ to the subset consisting of $\ve_i\in \Sigma$ such that $H_i(\ve_i)=y_i$ instead of $H_i(\ve_i)=1$ in the proof of \Cref{lem:correctness}.
The rest of the proof works analogously. 
Thus, $\{f_\secpar\}_\secpar$ is not one-way in the QROM. 

The proof of one-wayness in the CROM is similar to that of soundness of the proof of quantumness in \Cref{sec:PoQ}. 
By a straightforward extension of the proof of \Cref{lem:soundness} where we replace $1^n$ with arbitrary $y\in \bit^n$, we obtain the following claim. \begin{claim}\label{cla:OW_uniform}
For any adversary $\A$ that makes $\poly(\secpar)$ classical queries and $y\in \bit^n$, 
\begin{align*}
    \Pr[y= f_\secpar^{H}(\vx')
    :
    \vx'\sample \A^{H}(1^\secpar,y)]<\negl(\secpar).
\end{align*}
\end{claim}

The above claim does not immediately imply one-wayness since in the one-wayness game, $y$ is chosen by first sampling $\vx \sample C$ and then setting $y=f_\secpar^H(\vx)$ instead of fixing $y$ independently of $H$. 
Fortunately, we can show that the distribution of $y$ is almost independent of $H$ as shown in the following claim.
\begin{claim}\label{cla:LHL_like}
We have 
$$
\Delta((H,y),(H,y'))=\negl(\secpar)
$$
where $H\sample \Func(\Sigma,\bit^n)$,$\vx\sample C$, $y=f_\secpar^H(\vx)$, and $y'\sample \bit^n$. 
\end{claim}

By combining \Cref{cla:OW_uniform,cla:LHL_like}, one-wayness in the CROM immediately follows.

For proving \Cref{cla:LHL_like}, we rely on the following well-known lemma that relates the collision probability and statistical distance from the uniform distribution. 
\begin{definition}
For a random variable $X$ over a finite set $\mathcal{X}$, we define its collision probability as $\Col(X)=\sum_{x\in \mathcal{X}}\Pr[X=x]^2$. 
\end{definition}
\begin{lemma}\label{lem:cp_to_sd}
Let $X$ be a random variable over a finite set $\mathcal{X}$.
For $\epsilon>0$, if $\Col(X)\leq \frac{1}{|\mathcal{X}|}(1+\epsilon)$, then 
$$
\Delta(X,U_{\mathcal{X}})\leq \sqrt{\epsilon}/2
$$
where $U_{\mathcal{X}}$ denotes the uniform distribution over $\mathcal{X}$. 
\end{lemma}
See e.g., \cite[Lemma 4.5]{SODA:MitVad08} for the proof of \Cref{lem:cp_to_sd}. 

Then, we prove \Cref{cla:LHL_like} below.
\begin{proof}[Proof of \Cref{cla:LHL_like}]
By \Cref{lem:cp_to_sd}, it suffices to prove $\Col(H,y)=2^{-(|\Sigma|+1)n}\cdot (1+\negl(\secpar))$ where $H\sample \Func(\Sigma,\bit^n),\vx\sample C,y=f_\secpar^H(\vx)$.
We prove this as follows where $H$ and $H'$ are uniformly sampled from $\Func(\Sigma,\bit^n)$ and $\vx$ and $\vx'$ are uniformly sampled from $C$.  
\begin{align*}
    \Col(H,y)&=\Pr_{H,H',\vx,\vx'}[H=H'~\land~f^{H}_\secpar(\vx)=f^{H'}_\secpar(\vx')]\\
    &=2^{-|\Sigma|n}\cdot \Pr_{H,\vx,\vx'}[f^{H}_\secpar(\vx)=f^{H}_\secpar(\vx')]\\
    &=2^{-|\Sigma|n}\cdot \sum_{j=0}^{n}\Pr_{\vx,\vx'}[\hw(\vx-\vx')=n-j]\cdot 2^{-(n-j)}\\
    &=2^{-|\Sigma|n}\cdot \sum_{j=0}^{n}\Pr_{\vx}[\hw(\vx)=n-j]\cdot 2^{-(n-j)}\\
     &\leq 2^{-(|\Sigma|+1)n}\cdot \left(1+\frac{2^n}{|C_\secpar|}+\sum_{j=1}^{n-1}
    \Pr_{\vx}[\hw(\vx)=n-j]
    \cdot 2^{j}\right)\\
    &\leq 2^{-(|\Sigma|+1)n}\cdot \left(1+\frac{2^n}{|C_\secpar|}+\sum_{j=1}^{n-1}
    \left(\frac{2n}{|\Sigma|}\right)^j
    \right)\\
     &\leq 2^{-(|\Sigma|+1)n}\cdot \left(1+\frac{2^n}{|C_\secpar|}+\sum_{j=1}^{\infty}
    \left(\frac{2n}{|\Sigma|}\right)^j
    \right)\\
    &=2^{-(|\Sigma|+1)n}\cdot \left(1+\frac{2^n}{|C_\secpar|}+\frac{\left(\frac{2n}{|\Sigma|}\right)}{
    1-\left(\frac{2n}{|\Sigma|}\right)}
    \right)\\
    &=2^{-(|\Sigma|+1)n}\cdot(1+\negl(\secpar))
\end{align*}

where we used 
$\Pr_\vx[\hw(\vx)=n]\leq 1$ and
$\Pr_\vx[\hw(\vx)=0]=\frac{1}{|C_\secpar|}$ for the fifth line,  
\Cref{item:hw} 
of 
\Cref{lem:good_codes} for the sixth line, and
$|\Sigma|=2^{\secpar^{\Theta(1)}}$, $n=\Theta(\secpar)$, and $|C_\secpar|\geq 2^{n+\secpar}$ for the final line.
This completes the proof of \Cref{cla:LHL_like}.
\end{proof}
This completes the proof of \Cref{thm:separation_OWF}.
\end{proof}
\begin{remark}[On distributional one-wayness]\label{rem:dist_OW}
It is worth mentioning that $\{f_\secpar\}_\secpar$ is not even \emph{distributionally} one-way in the QROM. 
That is, one can find an almost uniformly distributed preimage of $y$ with quantum oracle access to $H$. This can be seen by observing that the proof of \Cref{lem:correctness} actually shows that the proving algorithm outputs an almost uniformly distributed valid proof. This corresponds to finding an almost uniformly distributed preimage of $y$ for the above $f_\secpar$.
\end{remark}

\subsection{Counterexample for Collision-Resistant Hash Functions.}\label{sec:CRH}
We give a construction of a family of compressing functions that is collision-resistant in the CROM but not even one-way in the QROM. It is a generic construction based on proofs of quantumness.
\begin{theorem}[Counterexample for collision-resistant functions]\label{thm:separation_CRH}
There exists a family $\{f_\secpar\}_\secpar$ of efficiently computable oracle-aided compressing keyless (resp. keyed) functions that is collision-resistant against in the CROM (resp. AI-CROM) but not even one-way against oracle-independent adversaries in the QROM. 
\end{theorem}
\begin{proof}
Since the keyed version immediately follows from the keyless version by \Cref{thm:uniform_to_non-uniform_CR}, we prove the keyless version below. 

Let $(\prove,\verify)$ be a keyless proof of quantumness that satisfies soundness in the CROM as given in \Cref{thm:PoQ}.
Let $\pilength$ be its maximum proof length.

We assume that the proof of quantumness uses a random oracle $H:\bit^{\secpar+\pilength} \rightarrow \bit^\secpar$ without loss of generality. 
We construct $f_\secpar^H:\bit^{\secpar+\pilength}\rightarrow \bit^{\secpar}$ as follows:
\begin{align*}
    f_\secpar^H(x,\pi)\defeq 
    \begin{cases}
    x&\text{if~}\verify^H(1^\secpar,\pi)=\top\\
    H(x,\pi)&\text{otherwise}
    \end{cases}
\end{align*}
where the input is parsed as $x\in \bit^\secpar$ and $\pi\in \bit^{\pilength}$. 
Collision-resistance of $\{f_\secpar\}_{\secpar}$ in the CROM is clear from the soundness of the proof of quantumness. Indeed, an adversary with a classical access to $H$ can output $(x,\pi)$ such that  $\verify(1^\secpar,\pi)=\top$ only with a negligible probability. Assuming that this does not happen, an adversary has to find a collision of $H$, which can be done only with probability at most $\frac{Q(Q+1)}{2}\cdot 2^{-\secpar}=\negl(\secpar)$ where $Q=\poly(\secpar)$ is the number of queries to $H$. 
On the other hand, the correctness of the proof of quantumness gives a trivial way to invert $f_{\secpar}^H$ on any target $y\in \bit^\secpar$ with a quantum access to $H$: one can just run $\pi\sample\prove^{H}(1^\secpar)$ and then output $(y,\pi)$.  We have
$f_{\secpar}^H(y,\pi)=y$ except for a negligible probability by the correctness of the proof of quantumness. This means that $\{f_\secpar\}_{\secpar}$ is not one-way 
 in the QROM. 
\end{proof}

\subsection{Counterexamples for Public Key Primitives}
In \cite{EC:YamZha21}, they give counterexamples for public key encryption (PKE) and digital signatures. Since their constructions are generic based on proofs of quantumness, we can plug our proofs of quantumness given in \Cref{sec:PoQ} into their constructions to obtain the following theorems.
\begin{theorem}
If there exists  a PKE scheme that is IND-CPA secure in the CROM, then there exists a PKE scheme that is IND-CCA secure in the CROM but not IND-CPA secure in the QROM.
\end{theorem}
\begin{theorem}
There exists a digital signature scheme that is EUF-CMA secure in the CROM but not EUF-NMA secure in the QROM.
\end{theorem}
See \cite{EC:YamZha21} for the formal definitions of PKE and digital signatures and their security notions. 
Note that \cite{EC:YamZha21} proved similar theorems relative to additional artificial classical oracles and weaker variants of them assuming the LWE assumption. 
We significantly improve them by removing the necessity of additional oracles or complexity assumptions.


\subsection{A Remark on Pseudorandom Generators}
One might think that we can also construct pseudorandom generators (PRGs) that are secure in the CROM but insecure in the QROM because 
\Cref{thm:separation_OWF} gives one-way functions (OWFs) that are secure in the CROM but insecure in the QROM and
there is a black-box construction of PRGs from OWFs~\cite{HILL99}.  
However, we remark that this does not work. 
The reason is that PRGs constructed from OWFs may be secure in the QROM even if the building block OWF is insecure in the QROM. 
For example, there is no obvious attack against the PRG of~\cite{HILL99} even with an inverter for the building block OWF.  

Indeed, we believe that we can show that \emph{any} black-box construction of PRGs from OWFs may remain secure even if the building block OWF is insecure.
We sketch the intuition below. 
Let $f:\calX\rightarrow \calX$ be a OWF. We augment the domain to $\calX\times \calR$ where $\calR$ is an exponentially large space by defining 
\begin{align*}
    f'(x,r)\defeq f(x).
\end{align*}
Then, it is clear that $f'$ is also a OWF. 
Suppose that we construct a PRG $G$ by making black-box use of $f'$. Since $f'$ is a secure OWF, $G^{f'}$ is a secure PRG. 
For each $r^* \in \calR$, we define $f'_{r^*}$ as follows:
\begin{align*}
    f'_{r^*}(x,r)\defeq 
    \begin{cases}
    f(x) & \text{if~}r\neq r^*\\
    x &\text{otherwise}
    \end{cases}.
\end{align*} 
Then, $f'_{r^*}$ clearly does not satisfy the one-wayness: for inverting  $y$, one can just output $(y,r^*)$.
On the other hand, when we run $G$ with respect to $f'_{r^*}$ instead of $f'$ for a randomly chosen $r^*$, there would be a negligibly small chance of calling the second branch of $f'_{r^*}$ if the number of $G$'s queries is polynomial.
This means that $G$ remains secure even though the building block function $f'_{r^*}$ is insecure as a OWF. 

We observe that the (im)possibility of separating CROM and QROM for PRGs is closely related to the Aaronson-Ambainis conjecture~\cite{AA14} (\Cref{conj:AA}). 
Very roughly speaking, the conjecture claims that any single-bit output algorithm in the QROM can be simulated in the CROM with a polynomial blowup on the number of queries. Since a PRG distinguisher's output is a single-bit, it is reasonable to expect that we can prove the equivalence of QROM security and CROM security for PRGs under the Aaronson-Ambainis conjecture. Unfortunately, this does not work as it is because a distinguisher takes a PRG value as its input, which may be correlated with the random oracle, whereas the Aaronson-Ambainis conjecture only captures the case where no side information of the random oracle is given. 
Nonetheless, we conjecture that QROM security and CROM security for PRGs (against polynomial-query unbounded-time adversaries) are equivalent.
It is a fascinating direction for future work to reduce it to the Aaronson-Ambainis conjecture or its reasonable variant.\footnote{Interestingly, a follow-up work by Katz and Sela~\cite{KS24} \emph{unconditionally} proves our conjecture without relying on the Aaronson-Ambainis conjecture.}

\section{Proofs of Randomness}\label{sec:PoR}
In this section, we construct proofs of randomness assuming the Aaronson-Ambainis conjecture~\cite{AA14}. 

Roughly speaking, the Aaronson-Ambainis conjecture claims that for any algorithm $\A$ with a \emph{quantum} access to a random oracle, there is an algorithm $\B$ that  approximates the probability that $\A$ outputs a particular output with a \emph{classical} access to the random oracle, and the number of queries of $\A$ and $\B$ are polynomially related. 
A formal claim is stated below.
\begin{conjecture}[Aaronson-Ambainis conjecture~{\cite[Theorem 22]{AA14}}]\label{conj:AA}
Let $\epsilon,\delta>0$ be reals. Given any quantum algorithm $\A$ that makes $Q$ quantum queries to a random oracle $H:\bit^n\rightarrow \bit^m$, there exists a deterministic classical algorithm $\B$ that makes $\poly(Q,m,\epsilon^{-1},\delta^{-1})$ classical queries and satisfies
\begin{align*}
    \Pr_{H\sample \Func(\bit^n,\bit^m)}[\left|\Pr[\A^H()\rightarrow 1]-\B^{H}()\right|\leq \epsilon]\geq 1-\delta.
\end{align*}
\end{conjecture}
\begin{remark}
We remark that the way of stating the conjecture is slightly different from that in \cite[Theorem 22]{AA14}, but they are equivalent. The difference is that \cite{AA14} considers oracle access to Boolean inputs whereas we consider an oracle access to functions. They are equivalent by considering a function as a bit string concatenating outputs on all inputs. We remark that a straightforward rephrasing would result in an oracle with $1$-bit outputs, but their conjecture is equivalent in the setting with $m$-bit output oracles since an $m$-bit output oracle can be seen as a concatenation of $m$ $1$-bit output oracles. We note that the number of $\B$'s queries in the above conjecture depends on $m$ unlike theirs due to this difference.

We also remark that Aaronson and Ambainis \cite{AA14} reduce the above conjecture to another seemingly unrelated conjecture in Fourier analysis. In the literature,
the latter conjecture is often referred to as Aaronson-Ambainis conjecture. On the other hand, we call \Cref{conj:AA} Aaronson-Ambainis conjecture since this is more relevant to our work. 
\end{remark}

The main theorem we prove in this section is the following. 
\begin{theorem}\label{thm:proof_of_randomness}
If \Cref{conj:AA} is true, there exists keyless (resp. keyed) proofs of randomness in the QROM (resp. AI-QROM).
\end{theorem}

By \Cref{thm:randomness_uniform_to_non-uniform,thm:min-entropy_to_randomness}, it suffices to prove the following theorem for proving \Cref{thm:proof_of_randomness}. 
\begin{theorem}\label{thm:proof_of_min}
If \Cref{conj:AA} is true, there exists keyless proofs of min-entropy that has min-entropy in the QROM.
\end{theorem}
In the following, we prove \Cref{thm:proof_of_min}.

\paragraph{From proofs of quantumness to proofs of min-entropy.}
Our proof of quantumness constructed in \Cref{sec:PoQ} has a large entropy in proofs.  
We can easily show that this is inherent assuming Aaronson-Ambainis conjecture.
This is because if the proving algorithm is almost deterministic, it can be simulated by a polynomial-query classical algorithm, which breaks soundness. 
The following theorem gives a generalization of the above argument. 
\begin{theorem}\label{thm:poq_entropy}
If \Cref{conj:AA} is true, the following holds. 
Let $(\prove,\verify)$ be a keyless proof of quantumness relative to a random oracle $H:\bit^n\rightarrow \bit^m$ that satisfies $(Q_\poq(\secpar),\epsilon_\poq(\secpar))$-soundness. 
Let $\A$ be an adversary that makes $Q_\A(\secpar)$ quantum queries.
Let $\epsilon_\A(\secpar),\delta_\A(\secpar)>0$ be reals. 
There exists a polynomial $p$ such that if we have
\[
Q_\poq(\secpar)\geq p(\secpar,Q_\A(\secpar),\epsilon_\A(\secpar)^{-1},\delta_\A(\secpar)^{-1})
\]
and 
\[
\epsilon_\poq(\secpar)\leq \delta_\A(\secpar)/4,\footnote{In fact, it suffices to require $\epsilon_\poq(\secpar)\leq c\delta_\A(\secpar)$ for any constant $c<1$.}
\] 
for all $\secpar\in \mathbb{N}$, 
then we have
$$
\Pr_{H\sample \Func(\bit^n,\bit^m)}\left[
\max_{\pi^*\text{~s.t.~}\verify^H(1^\secpar,\pi^*)=\top}\Pr[\A^H(1^\secpar)\rightarrow \pi^*]\le \epsilon_\A(\secpar)
\right]\geq 1-\delta_\A(\secpar).
$$
\end{theorem}


We defer the proof of \Cref{thm:poq_entropy} to the end of this section. 
By plugging the proofs of quantumness in \Cref{sec:PoQ} into \Cref{thm:poq_entropy}, we obtain proofs of min-entropy, which proves \Cref{thm:proof_of_min}.  
\begin{proof}[Proof of \Cref{thm:proof_of_min}]
For any polynomial $h(\secpar)$, there exists a constant $C$ such that $Q_\poq(\secpar)=2^{C (h(\secpar)+\secpar)}$ and $\epsilon_\poq(\secpar)=2^{-\secpar-2}$ satisfy the requirements of \Cref{thm:poq_entropy} for $Q_\A(\secpar)=\poly(\secpar)$, $\epsilon_\A(\secpar)=2^{-(h(\secpar)+\secpar)}$, and $\delta_\A(\secpar)=2^{-\secpar}$. 
As shown in \Cref{lem:soundness}, our proof of quantumness constructed in \Cref{sec:PoQ}, which we denote by $(\prove_\poq,\verify_\poq)$, satisfies subexponential security. 
Thus, by standard complexity leveraging, there is a polynomial $q(\secpar)$ such that if we replace the security parameter with $q(\secpar)$ in $(\prove_\poq,\verify_\poq)$, then it satisfies $(2^{C (h(\secpar)+\secpar)}, 2^{-\secpar-2})$-soundness. By \Cref{thm:poq_entropy}, for any adversary $\A$ that makes $\poly(\secpar)$ quantum queries, we have 
\begin{align}\label{eq:min-entropy_lower_bound}
\Pr_{H\sample \Func(\bit^n,\bit^m)}\left[
\max_{\pi^*\text{~s.t.~}\verify^H_\poq(1^{q(\secpar)},\pi^*)=\top}\Pr[\A^H(1^\secpar)\rightarrow \pi^*]\le 2^{-(h(\secpar)+\secpar)}
\right]\geq 1-2^{-\secpar}.
\end{align}

Then, we construct proofs of min-entropy $(\prove,\verify)$ as follows.
\begin{description}
\item[$\prove^H(1^\secpar,1^{h(\secpar)}):=\prove^H_\poq(1^{q(\secpar)})$]
\item[$\verify^H(1^\secpar,1^{h(\secpar)},\pi)$:]
If $\verify^H_\poq(1^{q(\secpar)},\pi)=\bot$, it outputs $\bot$. Otherwise, it outputs $\pi$. 
\end{description}
Suppose that $(\prove,\verify)$  does not have min-entropy in the QROM. 
Then, there exist an adversary $\B$ that makes $\poly(\secpar)$ quantum queries and a polynomial $h(\secpar)$ such that we have  
\begin{align} \label{eq:min-entropy_broken}
\Pr[\verify^H(1^{\secpar},h(\secpar),\B^H(1^\secpar,1^{h(\secpar)}))\neq\bot]\geq1/\poly(\secpar)\wedge H_\infty\left(\B^H_{\top}(1^\secpar,1^{h(\secpar)})\right)\leq h(\lambda)
\end{align}
for a non-negligible fraction of $H$.
It is easy to see that \Cref{eq:min-entropy_broken} implies 
\begin{align*}
\max_{\pi^*\text{~s.t.~}\verify^H_\poq(1^{q(\secpar)},\pi^*)=\top}\Pr[\B^H(1^\secpar,1^{h(\secpar)})\rightarrow \pi^*]\ge  2^{-h(\secpar)}/\poly(\secpar).
\end{align*}
Since this holds for a non-negligible fraction of $H$, 
if we consider $\A^H(1^\secpar):=\B^H(1^\secpar,1^{h(\secpar)})$, 
this contradicts \Cref{eq:min-entropy_lower_bound}.
Therefore, $(\prove,\verify)$  has min-entropy in the QROM.
\end{proof}

\paragraph{Intuition for the proof of \Cref{thm:poq_entropy}.} 
In the following, we often omit dependence on $\secpar$ and simply write e.g., $\epsilon_\A$ to mean $\epsilon_\A(\secpar)$ for brevity. 

Towards a contradiction, we assume that 
\begin{align*}
\Pr_{H\sample \Func(\bit^n,\bit^m)}\left[
\max_{\pi^*\text{~s.t.~}\verify^H(1^\secpar,\pi^*)=\top}\Pr[\A^H(1^\secpar)\rightarrow \pi^*]> \epsilon_\A
\right]> \delta_\A.
\end{align*}
We have to construct a classical adversary that breaks the soundness of the proof of quantumness. 
If $\epsilon_\A\approx 1$, it is easy: We consider an algorithm $\A_j$ that outputs the $j$-th bit of $\A$'s output for $j\in [\ell_\pi]$ where $\ell_\pi$ is the length of a proof in the proof of quantumness. For $\delta_\A$-fraction of $H$, $\A_j$'s output is almost deterministic for all $j$. Then, we can classically simulate $\A_j$ for all $j$ by invoking \Cref{conj:AA} for $\epsilon\ll 1$ and $\delta \ll \delta_\A/\ell_{\pi}$. This breaks the soundness of the proof of quantumness.

When $\epsilon_\A\ll 1$, such a simple bit-by-bit simulation attack does not work. The reason is that mixing up bits of multiple valid proofs does not result in a valid proof in general.  
To deal with such a case, we attempt to convert $\A$ into an almost deterministic attacker. If this is done, the same idea as the case of $\epsilon_\A\approx 1$ works. For making  $\A$ almost deterministic, our first idea is to consider an modified adversary $\A'$ that outputs the smallest valid proof $\pi$ in the lexicographical order such that $\A$ outputs $\pi$ with probability at least $\epsilon_\A$. If we can efficiently construct such $\A'$, then this idea works. However, the problem is that $\A'$ cannot exactly compute the probabilities that $\A$ outputs each $\pi$ with a limited number of queries. What $\A'$ can do is to run $\A$ many times to approximate the probabilities up to a $1/\poly$ error.\footnote{$\poly$ means a polynomial in the number of repetition of $\A$ run by $\A'$.} Now, a problem occurs if there are multiple $\pi$ such that the probability that $\A$ outputs $\pi$ is within $\epsilon_\A \pm 1/\poly$. 

To deal with this issue, we rely on an idea to randomly decide the threshold.\footnote{This idea is inspired by \cite{TCC:ChiChuYam20}.} That is, $\A'$ outputs the lexigographically smallest valid proof $\pi$ such that the approximated probability that $\A$ outputs $\pi$ is at least $t$ for some randomly chosen threshold $t\in (\epsilon_\A/2,\epsilon_\A)$. 
If we choose $t$ from a sufficiently large set and set the approximation error to be sufficiently small, we can show that it is impossible that there are multiple $\pi$ such that the probability that $\A$ outputs $\pi$ is within $t \pm 1/\poly$ for a large fraction of $t$ by a simple counting argument. This resolves the above problem. 

\paragraph{Proof of \Cref{thm:poq_entropy}.}
In the rest of this section, we give a formal proof of \Cref{thm:poq_entropy}.
We first show the following simple lemma. 
\begin{lemma}\label{lem:probability_approximation}
Let $\A$ be a (possibly quantum) algorithm that outputs an $\ell$-bit string $z$. 
For any $\epsilon,\delta>0$, there is an algorithm $\Approx(\A,\epsilon,\delta)$ that runs $\A$ $O(\ell\log(\delta^{-1}) \epsilon^{-2})$ times and 
outputs a tuple $\{P_z\}_{z\in \bit^\ell}$ such that 
\begin{align*}
    \Pr\left[\forall z\in \bit^\ell~\left|P_z-\Pr[\A()\rightarrow z]\right|\leq \epsilon\right]\geq 1-\delta
\end{align*}
where $\{P_z\}_{z\in \bit^\ell}\sample \Approx(\A,\epsilon,\delta)$. 
We say that $\Approx(\A,\epsilon,\delta)$ succeeds if the event in the above probability occurs. 
\end{lemma}
\begin{proof}
$\Approx(\A,\epsilon,\delta)$ works as follows.
It runs $\A()$ $N$ times where $N$ is an integer specified later.
For each $z$, let $K_z$ be the number of executions where $\A$ outputs $z$. 
Then it outputs $\{P_z\defeq \frac{K_z}{N}\}_{z\in \bit^\ell}$.

If we set $N\geq  C\ell\log(\delta^{-1})\epsilon^{-2}$ for a sufficiently large constant $C$, 
by the Chernoff bound (\Cref{lem:Chernoff}), 
for each $z$, we have 
\begin{align*}
    \Pr\left[ \left|P_z-\Pr[\A()\rightarrow z]\right|\leq \epsilon\right]\geq 1-\frac{\delta}{2^\ell}.
\end{align*}
By the union bound, we obtain \Cref{lem:probability_approximation}.
\end{proof}
Then, we prove \Cref{thm:poq_entropy}. 
\begin{proof}[Proof of \Cref{thm:poq_entropy}.]
Towards a contradiction, we assume that 
\begin{align}
\label{eq:assumption}
\Pr_{H\sample \Func(\bit^n,\bit^m)}\left[
\max_{\pi^*\text{~s.t.~}\verify^H(1^\secpar,\pi^*)=\top}\Pr[\A^H(1^\secpar)\rightarrow \pi^*]> \epsilon_\A
\right]> \delta_\A.
\end{align}
It suffices to prove that there exists a classical adversary $\B$ that makes $p(Q_\A,m,\epsilon_\A^{-1},\delta_\A^{-1})$ quantum queries and satisfies  
$$
\Pr_{H\sample \Func(\bit^n,\bit^m)}[\verify^H(1^\secpar,\pi)=\top:\pi \sample  \B^{H}(1^\secpar)]\geq \delta_\A/4
$$
for some polynomial $p$. 
Let $M\defeq \lceil \frac{4}{\epsilon_\A}\rceil$. 
For $i\in [M]$, we consider a quantum adversary $\A_i$ that works as follows.
\begin{description}
\item[$\A_i^H(1^\secpar)$:]
It runs $\{P_\pi\}_{\pi\in \bit^{\ell_\pi}}\sample \Approx(\A,\frac{\epsilon_\A}{4M},\frac{1}{5})$ where $\ell_\pi$ is the length of a proof.
Then it outputs the smallest $\pi$ in the lexicographical order that satisfies 
$$
\verify^H(1^\secpar,\pi)=\top
$$
and 
$$
P_\pi>\frac{\epsilon_\A}{2}\left(1+\frac{2i-1}{2M}\right).
$$
\end{description}
The number of queries by $\A_i$ is $Q_{\A_i}=\poly(\secpar,Q_\A, \epsilon_\A^{-1})$ since $\ell_\pi=\poly(\secpar)$.
For each $H$, let $\pi_i^H$ be the most likely output of $\A^H_{i}(1^\secpar)$.\footnote{If there is a tie, we choose the smallest one in the lexicographical order.}
We prove the following claim. 
\begin{claim}\label{cla:A_i_success_prob}
For at least $\left(\frac{\delta_A}{2}\right)$-fraction of $H\in \Func(\bit^n,\bit^m)$ and $i\in [M]$, it holds that
$$
\Pr[\A_i^H(1^\secpar)\rightarrow \pi_i^H]> 4/5.
$$
\end{claim}
\begin{proof}[Proof of \Cref{cla:A_i_success_prob}]
By \Cref{eq:assumption}, at least $\delta_\A$-fraction of $H$ satisfies   
\begin{align}\label{eq:max_prob}
\max_{\pi^*\text{~s.t.~}\verify^H(1^\secpar,\pi^*)=\top}\Pr[\A^H(1^\secpar)\rightarrow \pi^*]> \epsilon_\A.
\end{align}
Fix such $H$. 
Then, for at least $\frac{1}{2}$-fraction of $i\in[M]$, 
there does not exist $\pi$ satisfying 
\begin{align}\label{eq:prob_in_interval}
\left|
\Pr[\A^H(1^\secpar)\rightarrow \pi]-\frac{\epsilon_\A}{2}\left(1+\frac{2i-1}{2M}\right)
\right|
<
\frac{\epsilon_\A}{4M}.
\end{align}
This can be seen by a simple counting argument.
First, we remark that if $\pi$ satisfies \Cref{eq:prob_in_interval} for some $i\in [M]$, then we have 
$\Pr[\A^H(1^\secpar)\rightarrow \pi]>\epsilon_\A/2$. 
Therefore, the number of such $\pi$ is at most $2/\epsilon_\A$. 
Second, we remark that each $\pi$ can satisfy  \Cref{eq:prob_in_interval} for at most one $i$. 
Therefore, the fraction of $i\in [M]$ such that there is $\pi$ that satisfies \Cref{eq:prob_in_interval} is at most $2/(\epsilon_\A M)\leq 1/2$. 

Therefore, for at least $\left(\frac{\delta_\A}{2}\right)$-fraction of $H$ and $i$, \Cref{eq:max_prob} holds and there does not exist $\pi$ satisfying \Cref{eq:prob_in_interval}. 
For such $H$ and $i$, if $\Approx(\A,\frac{\epsilon_\A}{4M},\frac{1}{5})$ succeeds, which occurs with probability at least $\frac{4}{5}$,  then $\A_i$ outputs the smallest $\pi$ in the lexicographical order that satisfies 
$$
\verify^H(1^\secpar,\pi)=\top
$$
and 
$$
\Pr[\A^H(1^\secpar)\rightarrow \pi]>\frac{\epsilon_\A}{2}\left(1+\frac{2i-1}{2M}\right).
$$
Since the above $\pi$ is output with probability larger than $4/5$, this is the most likely output $\pi_i^H$. 
Thus, for at least $\left(\frac{\delta_\A}{2}\right)$-fraction of $H$ and $i$, $\A_i^H$ returns $\pi_i^H$ with probability larger than $4/5$.
This completes the proof of  \Cref{cla:A_i_success_prob}.
\end{proof}

For $j\in [\ell_\pi]$, let $\A_{i,j}$ be the algorithm that runs $\A_{i}$ and outputs the $j$-th bit of the output of $\A_i$.
Since $\A_{i,j}$ makes the same number of queries as $\A_i$, its number of queries is 
$Q_{\A_{i,j}}=Q_{\A_i}=\poly(\secpar,Q_\A,\epsilon_\A^{-1})$. 
We apply \Cref{conj:AA} to $\A_{i,j}$ where 
$\epsilon:=1/5$ and 
$\delta:=\frac{\delta_\A}{4\ell_\pi}$.
Then, \Cref{conj:AA} ensures that there exists a deterministic classical algorithm $\B_{i,j}$ that makes $\poly(Q_{\A_{i.j}},m,\epsilon^{-1},\delta^{-1})=\poly(\secpar,Q_\A, \epsilon_\A^{-1},\delta_\A^{-1})$ classical queries and satisfies
\begin{align*}
    \Pr_{H\sample \Func(\bit^n,\bit^m)}\left[\left|\Pr[\A^H_{i,j}(1^\secpar)\rightarrow 1]-\B^{H}_{i,j}(1^\secpar)\right|\leq 1/5\right]\geq 1-\frac{\delta_\A}{4\ell_\pi}.
\end{align*}
By the union bound, we have 
\begin{align}\label{eq:B_ij}
    \Pr_{H\sample \Func(\bit^n,\bit^m)}\left[
    \forall j\in[\ell_\pi]~
    \left|\Pr[\A^H_{i,j}(1^\secpar)\rightarrow 1]-\B^{H}_{i,j}(1^\secpar)\right|\leq 1/5\right]\geq 1-\frac{\delta_\A}{4}.
\end{align}

Now, we  give the classical adversary $\B$. 
\begin{description}
\item[$\B^H(1^\secpar)$:]
It randomly chooses $i\sample [M]$. For $j=1,2,...,\ell_\pi$, it runs $\B^H_{i,j}(1^\secpar)$ and sets $\pi_j\defeq 1$ if the output is larger than $1/2$ and $\pi_j\defeq 0$ otherwise. 
Then it outputs $\pi\defeq \pi_1\concat \pi_2\concat...\concat \pi_{\ell_\pi}$.
\end{description}
By the construction, we can see that $\B$ makes $\poly(\secpar,Q_\A,\epsilon_{\A}^{-1},\delta_{\A}^{-1})$ queries. 
By combining \Cref{cla:A_i_success_prob} and \Cref{eq:B_ij}, for at least $\left(\frac{\delta_A}{4}\right)$-fraction of $H\in \Func(\bit^n,\bit^m)$ and $i\in [M]$, 
for all $j\in[\ell_\pi]$, 
if the $j$-th bit of $\pi^H_i$ is $1$, we have
$$
\B^{H}_{i,j}(1^\secpar)>3/5
$$
and otherwise
$$
\B^{H}_{i,j}(1^\secpar)<2/5.
$$
Thus, for such $H$ and $i$, $\B^{H}(1^\secpar)$ outputs $\pi^H_i$. 
Since we have $\verify^H(1^\secpar,\pi^H_i)=\top$ for all $i\in [M]$, we have 
$$
\Pr_{H\sample \Func(\bit^n,\bit^m)}[\verify^H(1^\secpar,\pi)=\top:\pi \sample  \B^{H}(1^\secpar)]\geq \frac{\delta_\A}{4}.
$$
This contradicts the soundness of the proof of quantumness in the CROM. 
This completes the proof of \Cref{thm:poq_entropy}.
\end{proof}

\section{Proof of Theorem \ref{thm:randomness_uniform_to_non-uniform}}\label{sec:proof_randomness_uniform_to_non-uniform}
In this section, we prove \Cref{thm:randomness_uniform_to_non-uniform}. For the reader's convenience, we restate the theorem below. 
\begin{theorem}[Restatement of \Cref{thm:randomness_uniform_to_non-uniform}] 
If $(\prove_0,\verify_0)$ is a proof of min-entropy (resp. proof of randomness) in the QROM, then $(\prove,\verify)$ is a proof of min-entropy (resp. proof of randomness) in the AI-QROM, where $\prove^H(1^\secpar,k_0||k_1,1^h)=\prove_0^{H(k_1||\cdot)}(1^\secpar,k_0,1^{h+1})$ and $\verify^H(1^\secpar,k_0||k_1,1^h,\pi)=\verify_0^{H(k_1||\cdot)}(1^\secpar,k_0,1^{h+1},\pi)$ and where $k_1\in\bit^\secpar$.
\end{theorem}
\begin{proof}We prove the case of proof of min-entropy, the case of proofs of randomness being essentially identical. Consider a non-uniform oracle-dependent adversary $\A$ for the min-entropy of $(\prove,\verify)$, with advice function $a(H)$ of polynomial output length.

To get an intuition for our proof, consider two possible advice strings $a(H)$. The first is where $a(H)$ is computed by choosing an arbitrary $k_1^*$, and setting $a(H)$ to be some function of $H(k_1^*||\cdot)$, the portion of the truth table that uses the prefix $k_1^*$. The second is where $a(H)$ is, say, $H(0||x)\oplus H(1||x)\oplus H(2||x),\cdots$ for some $x$, which depends on $H$ evaluated at all possible prefixes.

In the first case, $a(H)$ is only useful if $k_1=k_1^*$, which occurs with exponentially-small probability. If $k_1\neq k_1^*$, then since $\verify_0^{H(k_1||\cdot)}$ only queries $H$ on inputs that are independent of $a(H)$, security follows from the underlying security of $(\prove_0,\verify_0)$ in the ordinary QROM.

The second case is slightly trickier, since now $a(H)$ depends on all possible prefixes. Here, however, we can come up with a simple fix: choose a uniform $k_1^*$, and \emph{re-sample} $H$ on all inputs of the form $k_1^*||\cdot$. Let the resulting oracle be $H'$. Because $k_1^*$ is random and independent of the adversary's view, it is straightforward to show that this change negligibly impacts the adversary's output distribution. Now, however, $a(H)$ is actually independent of $H'$, since the re-sampled parts eliminate any dependency.

Our proof will follow similar lines, but work more generally. We re-sample a large-but-not-too-large number of prefixes, and show that this does not change the adversaries output distribution by much. Intuitively, if $a(H)$ depended globally on many prefixes (as in our second example), then by re-sampling a few prefixes we make $a(H)$ close to independent of $H'$. On the other hand, if $a(H)$ depends on just a few prefixes, it is anyway exponentially unlikely that $k_1$ will be among the prefixes. The result in either case is that $H(k_1||\cdot)$ will be close to independent of $a(H)$, which allows us to base security on the underlying security of $(\prove_0,\verify_0)$ in the ordinary QROM.

The above argument would work for ``typical'' cryptographic games. One wrinkle, however, with applying it to proofs of min-entropy is that a negligible change in the adversary's output distribution can result in a non-negligible change in the entropy. It is therefore insufficient to argue simply that the adversary's output distribution is negligibly close. We utilize a careful argument to show that, indeed, entropy is preserved in our reduction. The intuition is that instead of an additive error, we show that the probability of each outcome incurs only a small multiplicative change moving from $H$ to $H'$. Such a small multiplicative change will indeed preserve entropy. We now give the proof.

\medskip

Suppose $\A$ breaks min-entropy. This means there is a polynomial $h$, an inverse polynomial $\delta$ and a non-negligible $\epsilon$ such that the following simultaneously hold with probability at least $\epsilon$ over the choice of $H,k_0,k_1$:
\begin{align}
    \Pr[\verify^{H}(1^\secpar,k_0||k_1,1^h,\A^{H}(1^\secpar,a(H),k_0||k_1,1^h))\neq\bot]&\geq\delta(\lambda)\label{nonuniform1}\\
    H_\infty\left(\A^{H}_{\top}(1^\secpar,a(H),k_0||k_1,1^h)\right)&\leq h\label{nonuniform2}
\end{align}

We now implement the re-sampling process outlined above. Choose a second random oracle $J$. Moreover, choose a random set of salts $S\subseteq\bit^\secpar$. $S$ will be chosen as follows. First choose a size $\ell\in[2^{\secpar}]$ according to a distribution $D$, which will be specified later. Then choose $S$ to be a uniform random subset of size $\ell$. Define $H'$ as
\[H'(s,x)=\begin{cases}J(s,x)&\text{if }s\in S\\H(s,x)&\text{otherwise}\end{cases}\]

We now specify two different distributions $D_1,D_2$ for $\ell$, which induce distributions $E_1,E_2$ over $H'$. Let $k,d,n$ be non-negative integers with $dn\leq 2^\lambda$. We will think of $d,n$ as being super-polynomial, and $k$ as being polynomial. Define the matrix $\matA\in\mathbb{Z}^{(k+1)\times n}$ as follows:
\[\matA=\begin{pmatrix}
    1&1&1&1&\cdots&1\\
    0&1&2&3&\cdots&n\\
    0&1&4&9&\cdots&n^2\\
    \vdots&\vdots&\vdots&\vdots&\ddots&\vdots\\
    0&1&2^k&3^k&\cdots&n^k
\end{pmatrix}\]
Let $\vecx$ be the $n$-dimensional vector $\vecx=(1\;\;-1\;\;0\;\;0\;\;\cdots\;\;0)$, and let $\vecy$ be the orthogonal projection of $\vecx$ onto the space orthogonal to the rows $\matA$. Let $\vecy^+$ be the vector obtained from $\vecy$ by replacing all the negative entries with 0 and keeping all the positive entries. Let $\vecy^-$ be the vector obtained from $\vecy$ by replacing all the positive entries with 0, and negating all the negative entries (thereby making them positive). That is, 
\begin{align*}
    \vecy^+_i&=\max(\vecy_i,0)\\
    \vecy^-_i&=\max(-\vecy_i,0)
\end{align*}
This means $\vecy^+,\vecy^-$ have only non-negative entries, and $\vecy=\vecy^+-\vecy^-$. We will 0-index the coordinates of $\vecy,\vecy^+,\vecy^-$, so that the first entry has position $i=0$, the second has position $i=1$, etc.

Now define $D_1$ as the distribution which samples $i\cdot d$ with probability proportional to $\vecy^+_i$ (namely, with probability $\vecy^+_i/|\vecy^+|_1$ where $|\cdot|_1$ represents the 1-norm), and $D_2$ as the distribution which samples $i\cdot d$ with probability proportional to $\vecy^-_i$  (namely, with probability $\vecy^-_i/|\vecy^-|_1$). We call $E_1,E_2$ the distributions over $H'$ that result from sampling $\ell$ from $D_1,D_2$, respectively.

The intuition for these distributions is that $\vecy^+$ will be very close to $(1\;\;0\;\;\cdots\;\;0)$ while $\vecy^-$ will be very close to $(0\;\;1\;\;0\;\;\cdots\;\;0)$. This means $D_1$ will place the bulk of its weight on 0, meaning $|S|=0$ with high probability, in which case $H'=H$. The small probability that $H\neq H'$ means that the probability of any output of $\A$ could only have changed by a small multiplicative amount, meaning the min-entropy stays low (we want the entropy to stay low since we are ultimately going to use the adversary to break $(\prove_0,\verify_0)$). On the other hand, $D_2$ places \emph{all} of its weight on values at least $d$, meaning $|S|\geq d$. In this case, we will show that for a random choice of $s\notin S$, the truth table of $H(s,\cdot)$ is essentially independent of $a(H)$ given $H'$. This allows us to show that if $\A$ breaks min-entropy under the distribution $D_2$, then we can turn $\A$ into an adversary $\B$ for $H(s,\cdot)$ in the setting where $\B$ is given no auxiliary input. This would contradict the assumed security of $(\prove_0,\verify_0)$. The proof is then completed by showing that, since $\matA\cdot(\vecy^+-\vecy^-)=0$, the output distributions under $D_1$ and $D_2$ are identical. We now prove the above facts.

\paragraph{Part 1: Small entropy difference for $E_1$.} We now show that in the case $H'$ is sampled from $E_1$ (that is, $\ell$ sampled from $D_1$), that the resulting distribution is very close to $H$. More precisely:

\begin{lemma}\label{lem:smallentropychange} Fix $H,k_0,k_1$, which in turn fixes $a(H)$. Let $z$ be any possible output of $\A$. Then \[\Pr_{H'\gets E_1}[z\gets\A^{H'}(1^\secpar,a(H),k_0||k_1,1^h)]\geq \left(1-O(k^3/n^{1/2})\right)\Pr[z\gets\A^{H}(1^\secpar,a(H),k_0||k_1,1^h)]\]
\end{lemma}
This means that the most likely outcome $z$ is only negligibly effected by moving from $H$ to $H'$, when $\ell$ is sampled from $D_1$. Hence the min-entropy of the output distribution of $\A$ can only increase by a negligible amount.

Since $H'=H$ when $\ell=0$, Lemma~\ref{lem:smallentropychange} is an immediate consequence of the following lemma:
\begin{lemma}$\Pr[0\gets D_1]\geq 1-O(k^3/n^{1/2})$
\end{lemma} 
\begin{proof}Let $\vecz$ be the projection of $\vecx=(1\;\;-1\;\;0\;\;0\;\;\cdots\;\;0)$ onto the row-span of $\matA$, meaning $\vecz+\vecy=\vecx$ and $\vecy,\vecz$ are orthogonal. Hence $2=|\vecx|^2=|\vecz|^2+|\vecy|^2$. Our goal will be to bound $|\vecz|$ to being negligible. This will imply that $\vecy$ is very close to $\vecx$, and hence $\vecy^+$ is very close to $(1\;\;0\;\;0\;\;\cdots\;\;0)$. This in turn means most of the mass of $D_1$ is on 0, as desired.

Consider the matrix $\matB=\matA\cdot\matA^T$. Then $\matB_{i,i'}=\sum_{j=0}^n j^{i+i'}$ (where we 0-index $i,i'$). This sum very closely approximates $n^{i+i'+1}/(i+i'+1)$. To keep the following analysis simpler, we will take $\matB_{i,i'}=n^{i+i'+1}/(i+i'+1)$; the error caused by this will be small and therefore will be absorbed into the big-O. We can then write $\matB=n\cdot \matD\cdot\matB'\cdot\matD$
where
\[\matD=\begin{pmatrix}1&&&\\&n&&\\&&n^2&\\&&&\ddots\end{pmatrix}\;\;\;\;\matB'=\begin{pmatrix}1&\frac{1}{2}&\frac{1}{3}&\cdots&\frac{1}{k+1}\\
\frac{1}{2}&\frac{1}{3}&\frac{1}{4}&\cdots&\frac{1}{k+2}\\
\frac{1}{3}&\frac{1}{4}&\frac{1}{5}&\cdots&\frac{1}{k+3}\\
\vdots&\vdots&\vdots&\ddots&\vdots\\
\frac{1}{k+1}&\frac{1}{k+2}&\frac{1}{k+3}&\cdots&\frac{1}{2k+1}\end{pmatrix}\]

Observe that the matrix representing the orthogonal projection onto the row-span of $\matA$ is $\matA^T\cdot\matB^{-1}\cdot\matA$. Therefore, we have that 
\begin{align*}|\vecz|^2 &= \vecz^T\cdot\vecz = \vecx^T\cdot \matA^T\cdot\matB^{-1}\cdot\matA\cdot\vecx=(0\;\;-1\;\;-1\;\;\cdots\;\;-1)\cdot\matB^{-1}\cdot\begin{pmatrix}0\\-1\\-1\\\vdots\\-1\end{pmatrix}\\&=\frac{1}{n}\cdot \left(0\;\;\frac{1}{n}\;\;\frac{1}{n^2}\;\;\cdots\;\;\frac{1}{n^k}\right)\cdot(\matB')^{-1}\cdot\begin{pmatrix}0\\1/n\\1/n^2\\\vdots\\1/n^k\end{pmatrix}\end{align*}

We therefore must compute $(\matB')^{-1}$. Fortunately, the inverse is known. $\matB$ is known as the Hilbert matrix, and it's inverse is given by:

\begin{lemma}[\cite{choi83}]\label{claim:Binv}$(\matB')^{-1}_{i,i'}=(-1)^{i+i'}(i+i'+1)\binom{i+i'}{i}^2\binom{k+i}{i+i'+1}\binom{k+i'}{i+i'+1}$, where again $i,i'$ are 0-indexed.
\end{lemma}
With Lemma~\ref{claim:Binv}, we have that $|\vecz|^2=\sum_{j=2}^{2k}\frac{(-1)^j(j+1)}{n^{j+1}}\sum_i \binom{j}{i}\binom{k+i}{j+1}\binom{k+j-i}{j+1}$. We can lower-bound the sum over $i$ by 0 and upper bound it by $\sum_i\binom{j}{i}(2k)^{2(j+1)}=2^j\cdot (2k)^{2(j+1)}\leq (4k)^{2j+2}$. Thus,
\[|\vecz|^2\leq\sum_{j'=1}^k (2j'+1) \left(\frac{16k^2}{n}\right)^{2j'+1}\leq \sum_{j'=1}^\infty (2j'+1) \left(\frac{16k^2}{n}\right)^{2j'+1}=\frac{(3-\alpha^2)\alpha^3}{(1-\alpha^2)^2}\leq 12\left(\frac{16k^2}{n}\right)^3\]
where above we set $j'=j/2$ for even $j$ (the odd $j$ terms being bounded by 0), $\alpha=16k^2/n$, and we assume $\alpha\leq 1/2$.

Thus we have that $|\vecz|=O(k^3/n^{3/2})$, which in turn implies that $|\vecz|_1\leq n|\vecz|=O(k^3/n^{1/2})$. Since we have $\vecy=(1\;\;-1\;\;0\;\;\cdots\;\;0)-\vecz$, and $\vecy^+$ contains all the non-negative entries of $\vecy$, we therefore have that $\vecy^+_1\geq 1-O(k^3/n^{1/2})$, and all the remaining entries of $\vecy^+$ sum to less that $O(k^3/n^{1/2})$. Thus $|\vecy^+|_1= 1\pm O(k^2/n^{1/2})$, and so $\vecy^+_1/|\vecy^+|\geq 1-O(k^3/n^{1/2})$. Thus the distribution $D_1$ will output zero with probability at least $1-O(k^3/n^{1/2})$, as desired.\end{proof}

\paragraph{Part 2: Equivalent Output Distributions.} We next show that the output distributions are equivalent under $E_1$ and $E_2$:
\begin{lemma}\label{lem:equivoutput} Fix $H,k_0,k_1$, which in turn fixes $a(H)$. Let $z$ be any possible output of $\A$. Let $q$ be the number of queries made by $\A$, and assume $k\geq 4q$. Then $\Pr_{H'\gets E_1}[z\gets\A^{H'}(1^\secpar,a(H),k_0||k_1,1^h)]=\Pr_{H'\gets E_2}[z\gets\A^{H'}(1^\secpar,a(H),k_0||k_1,1^h)]$. In other words, output distributions of $\A^{H'}(1^\secpar,a(H),k_0||k_1,1^h)$ is identical whether $H'$ is sampled from $E_1$ or $E_2$.
\end{lemma}
\begin{proof}Our proof will use the polynomial method~\cite{JACM:BBCMW01}. Specifically, we will make use of the following formulation, shown in~\cite{C:Zhandry12}:
\begin{lemma}\label{lem:lincomb}Let $\A$ be a quantum algorithm making $q'$ quantum queries to an oracle $H:\mathcal{X}\rightarrow\mathcal{Y}$. If we draw $H$ from some distribution $D$, then for every $z$, the quantity $\Pr_{H\gets D}[z\gets \A^H()]$ is a linear combination of the quantities $\Pr_{H\gets D}[H(x_i)=r_i\forall i\in[2q']]$ for all possible settings of the $x_i$ and $r_i$. The coefficients in the linear combination are independent of the distribution $D$.
\end{lemma}
In the case $\mathcal{Y}=\{0,1\}$, by inclusion-exclusion, we can in turn write the quantities $\Pr_{H\gets D}[H(x_i)=r_i\forall i\in\{1,\cdots 2q'\}]$ as linear combinations of the quantities $\Pr_{H\gets D}[H(x_i)=1\forall i\in[k]]$ for all possible $k\leq 2q'$. 

We abuse notation, and let $S$ also denote the membership oracle for $S$, namely $S(s)=1$ if and only if $s\in S$. Now consider the distributions $D_1,D_2$, which induce distributions over $S$ that we will call $S_1$ and $S_2$, respectively. These in turn induce distributions $E_1,E_2$ over $H'$. Consider the algorithm that simulates $\A^{H'}$ by making queries to $S$, where $S$ is drawn from either $S_1$ or $S_2$, meaning that $H'$ is drawn from either $E_1$ or $E_2$. This simulation must make two queries to $S$ for each query $\A$ makes to $H'$: one to compute whether $s\in S$, and then one to un-compute the value at the end of the query. Thus, if $\A$ makes $q$ queries, the total number of queries the simulation makes to $S$ is $q'=2q$. Observe also that $S$ is independent of $H,k_0,k_1,a(H)$. Thus, after fixing $H,k_0,k_1,a(H)$, we can apply Lemma~\ref{lem:lincomb} to the simulation of $\A$, and see that the probability $\A$ outputs any given value $z$ is a linear combination of $\Pr_S[S(s_i)=1\forall i\in[k']]$ for $k'\leq 4q\leq k$, where the coefficients of the linear combination are independent of the distribution over $S$. It suffices, therefore, to prove that for all $k'\leq k$ and for all $s_1,\cdots,s_{k'}$, that 
\[\Pr_{S\gets S_1}[S(s_i)=1\forall i\in[k']]=\Pr_{S\gets S_2}[S(s_i)=1\forall i\in[k']]\]
Toward that end, we observe that, for any $s_1,\cdots,s_{k'}$, the event $S(s_i)=1\forall i\in[k']$ means that each $s_i\in S$. For a given size $\ell$ of $S$, there are $\binom{2^\lambda-k'}{\ell-k'}$ ways to choose such an $S$. Since for both $S_1,S_2$ we have that $S$ is uniform once we choose $\ell$, this means that for a given $\ell$, 
\[\Pr_S[S(s_i)=1\forall i\in[k']]=\binom{2^\lambda-k'}{\ell-k'}/\binom{2^\lambda}{\ell}=\frac{(2^\lambda-k')!\ell!}{(2^\lambda)!(\ell-k')!}=\frac{(2^\lambda-k')!}{(2^\lambda)!}\ell(\ell-1)\cdots(\ell-k'+1),\] which is a polynomial in $\ell$ of degree at most $k'\leq k$. 

This in turn means the probability of any outcome $z$, once we have fixed $z$, is a polynomial $p_z$ in $\ell$ of degree at most $k$. Averaging over all $\ell$, the probability of outcome $z$ is $\sum_\ell \Pr[\ell]p_z(\ell)$. We must therefore show that $\sum_\ell \Pr[\ell\gets D_1]p_z(\ell)=\sum_\ell \Pr[\ell\gets D_2]p_z(\ell)$, for which is suffices to show that $\sum_\ell (\Pr[\ell\gets D_1]-\Pr[\ell\gets D_2])\ell^j=0$ for all $j\in[0,k]$. Recall that $\ell$ is always a multiple of $d$, so this is equivalent to showing $\sum_i (\Pr[i\cdot d\gets D_1]-\Pr[i\cdot d\gets D_2])(i\cdot d)^j=0$ 

We now observe that $\vecy$ is in the kernel of $\matA$, meaning the sum of its components is 0. As such, we must have that $|\vecy^+|_1=|\vecy^-|_1=:R$. Therefore, when we re-normalize $\vecy^+$ and $\vecy^-$ to get the distributions $D_1,D_2$, the re-normalization is the same in both cases: dividing by $R$. Thus $\vecy_i/R=\vecy^+_i/R-\vecy^-_i/R=\Pr[i\cdot d\gets D_1]-\Pr[i\cdot d\gets D_2]$, meaning
$\sum_i (\Pr[i\cdot d\gets D_1]-\Pr[i\cdot d\gets D_2])(i\cdot d)^j=\frac{d^j}{R} (\matA\cdot \vecy)_j=0$, as desired.
\end{proof}

\paragraph{Part 3: Statistical independence for $E_2$.} Here, we show that when $H'$ is sampled from $E_2$, but when the adversary is still provided the advice $a(H)$, then for most choices of the salt $k_1$, $H(k_1||\cdot)$ is statistically close to uniform even given $a(H)$ and $H'$.

Let $H(s||\cdot)$ denote the slice of the truth table of $H$ corresponding to salt $s$. Let $\overline{H}(s||\cdot)$ denote the remaining truth table not included in $H(s||\cdot)$.

\begin{lemma}\label{lem:statrandom} Consider sampling a uniform $H$, and then sampling $H'\gets E_2$ and letting $k_1\gets\{0,1\}^\lambda\setminus S$. Then the distributions $(a(H),k_1,H(k_1||\cdot),\overline{H'}(k_1||\cdot)$ and $(a(H),k_1,R,\overline{H'}(k_1||\cdot)$ are $\sqrt{|a(H)|/2d}$-close in statistical distance.
\end{lemma}
\begin{proof}
In order to prove Lemma~\ref{lem:statrandom}, we will need the following technical lemma:

\begin{lemma}\label{lem:stats} Let $D$ be a distribution and $X_1,\dots,X_g,Y$ be iid random variables sampled from $D$. Let $F$ be a function with co-domain of size $2^r$. Then 
\[\Delta(\;(\calI,X_\calI,F(X_1,\dots,X_g))\;,\;(\calI,Y,F(X_1,\dots,X_g))\;)\leq \sqrt{r/2g}\]
Above, $\calI$ is uniform in $[g]$, and $\Delta$ denotes statistical distance.
\end{lemma}
\begin{proof}Let $I(X;Y)$ denote the mutual information between random variables $X$ and $Y$. Then
	\[r\geq I(\;F(X_1,\dots,X_g)\;;\;X_1,\dots,X_t\;)\geq \sum_{i=1}^g I(\;F(X_1,\dots,X_g)\;;\; X_i\;)\]
where the second inequality is due to the independence of the $X_i$. Let $\delta_i$ be the statistical distance between the distributions $(F(X_1,\dots,X_g),X_i)$ and $(F(X_1,\dots,X_g),Y)$. Let $\delta$ be the statistical distance between $(\calI,X_\calI,F(X_1,\dots,X_g))$ and $(\calI,Y,F(X_1,\dots,X_g))$; our goal is to bound $\delta$. $I(\;F(X_1,\dots,X_g)\;;\; X_i\;)$ is just the KL divergence between $(F(X_1,\dots,X_g),X_i)$ and $(F(X_1,\dots,X_g),Y)$. By Pinsker's inequality, we therefore have that $\delta_i\leq \sqrt{I(\;F(X_1,\dots,X_t)\;;\; X_i\;)/2}$. This implies
\[r\geq 2\sum_{i=1}^g \delta_i^2\]
On the other hand, $\delta=(\sum_i \delta_i)/g$. Jensen's inequality then gives that
\[\delta\leq \sqrt{\sum_i \delta_i^2/g}\leq \sqrt{r/2g}\qedhere\]	
\end{proof}

We now apply Lemma~\ref{lem:stats} to our setting. Consider sampling a random $S$ of size $\ell$ where $\ell$ is sampled from $D_2$. $D_2$ only has support on $\ell$ of size at least $d$. Now consider sampling a random $k_1\notin S$. It is equivalent to sample a random set $S'$ of size $\ell+1$, and then let $k_1$ be uniform in $S'$, and $S=S'\setminus\{k_1\}$.

Therefore let $g=\ell+1$, and let $X_1,\cdots,X_g$ denote the slices $H(s||\cdot)$ of the truth table of $H$, for $s\in S\cup\{k_1\}$. Now fix $H(s||\cdot)$ for $s\notin S\cup\{k_1\}$; call this partial truth table $H_{\sf part}$. Let $F$ be the function from $X_1,\cdots,X_g$ which computes $a(H)$ ($H$ being fully specified by $H_{\sf part}$ together with $X_1,\cdots X_g$). Lemma~\ref{lem:stats} now says that the tuples $(k_1,H(k_1||\cdot),a(H))$ and $(k_1,R,a(H))$ are $\sqrt{|a(H)|/2d}$-close given $H_{\sf part}$, where $R$ is a independent uniform truth table. To complete the proof of Lemma~\ref{lem:statrandom}, we simply observe that $\overline{H'}(k_1||\cdots)$ consists of $H_{\sf part}$ together with $H'(s||\cdots)$ for $s\in S$. But recall that for $s\in S$, we set $H'(s||\cdots)=J(s||\cdots)$ where $J$ is an independent random oracle, meaning all information about $H(s||\cdots)$ is erased from $H'$. Therefore, even conditioned on $\overline{H'}(k_1||\cdots)$, the tuples $(k_1,H(k_1||\cdot),a(H))$ and $(k_1,R,a(H))$ remain statistically close. Averaging over all choices of $\overline{H'}$ gives the lemma. \end{proof}

\paragraph{Part 4: Putting it all together.} We now put everything together, obtaining an adversary for $\prove_0,\verify_0$. To create our adversary $\B$, we do the following:
\begin{itemize}
    \item Choose a random $H$ and compute $a(H)$.
    \item Choose a random set $S$ from $D_2$. Choose a random $J$ and compute $H'$.
    \item Choose a random $k_1\in \{0,1\}^\lambda\setminus S$.
\end{itemize}
We will fix $H,a(H),S,k_1,H'$ in the description of $\B$; alternatively we could imagine $\B$ choosing the $H,a(H),S,k_1,H'$ which maximize its success probability.

$\B^{H_0}(1^\secpar,k_0,1^h)$ runs $\A^{H''}(1^\secpar,a(H),k_0||k_1,1^h)$ and outputs whatever $\A$ outputs, where $H_0$ is the random oracle $\B$ is given, and $H''$ is the oracle:
\[H''(s,x)=\begin{cases}H'(s,x)&\text{ if }s\neq k_1\\H_0(x)&\text{ if }s=k_1\end{cases}\]

\begin{lemma}With non-negligible probability over the choice of $H,a(H),S,k_1,H'$ as sampled above, there is a non-negligible $\delta'$ such that the following is true:
\begin{align}
    \Pr[\verify_0^{H_0}(1^\secpar,k_0,1^h,\B^{H_0}(1^\secpar,k_0,1^h))\neq\bot]&\geq\delta'(\lambda)\label{nonuniform1}\\
    H_\infty\left(\B^{H_0}_{\top}(1^\secpar,k_0,1^h)\right)&\leq h+1\label{nonuniform2}
\end{align}
where the probabilities above are taken over the choice of uniform $H_0,k_0$. In particular, there exists such a choice of $H,a(H),S,k_1,H'$ which makes $\B$ break the security of $\prove_0,\verify_0$.
\end{lemma}
\noindent This lemma therefore completes the proof of Theorem~\ref{thm:randomness_uniform_to_non-uniform}.
\begin{proof}We first consider setting $H_0$ to be $H'(k_1||\cdot)$. In this case, $H''=H'$ so $\B$ runs $\A$ on $H'$, and by Lemmas~\ref{lem:smallentropychange} and ~\ref{lem:equivoutput}, the entropy of the output of $\A$ and hence $\B$ is less than $h+1$ with non-negligible probability over the choice of $H,a(H),S,k_1,H'$. 

Now we actually set $\B$'s oracle to $H_0$. By Lemma~\ref{lem:statrandom}, $H_0$ and $H'(k_1||\cdot)$ are statistically close, even given $a(H),S,k_1,\overline{H'}(k_1||\cdot)$. Since the min-entropy of $\B$ is a property of the oracle it sees (and $k_0$), even after changing to $H_0$, the probability $\B$'s entropy is less than $h+1$ is only negligibly affected, and is hence still non-negligible.
\end{proof}
This completes the proof of Theorem~\ref{thm:randomness_uniform_to_non-uniform}.\end{proof}

\bibliographystyle{alpha}
\bibliography{abbrev3,crypto,reference}
\appendix 
\end{document}